\documentclass[english,11pt]{article}
\providecommand{\keywords}[1]{\textbf{\textit{Key words---}} #1}
\usepackage{hyperref}
\usepackage{algpseudocode}
\usepackage{algorithmicx}
\usepackage{caption}
\usepackage{mathbbol}
\usepackage[english]{babel}
\usepackage{hyperref}

\usepackage[T1]{fontenc}
\usepackage{graphicx}
\usepackage{subfigure}
\usepackage[utf8]{inputenc}
\usepackage[a4paper]{geometry}
\usepackage{babel}
\usepackage{algorithm}
\usepackage{bbm}
\usepackage{color}
\usepackage{bm}
\usepackage{appendix}
\usepackage{textcomp}
\usepackage{amsthm}
\usepackage{amsmath}
\usepackage{amssymb}
\usepackage{esint}
\usepackage{enumerate}
\usepackage{cite}
\usepackage{pifont}
\usepackage{dsfont}

\newcommand{\mC}{\mathbb{C}}

\allowdisplaybreaks[4] 

\makeatletter

\newtheorem{thm}{Theorem}[section]
\newtheorem{lemma}[thm]{Lemma}
\newtheorem{pro}[thm]{Proposition}

\newtheorem{definition}[thm]{Definition}

\newtheorem{rem}{Remark}

\allowdisplaybreaks[4]

\ifx\proof\undefined
\newenvironment{proof}[1][\protect\proofname]{\par
\normalfont\topsep6\p@\@plus6\p@\relax
\operatorname{tr}ivlist
\itemindent\parindent
\item[\hskip\labelsep\scshape #1]\ignorespaces
}{%
\endtrivlist\@endpefalse
}
\providecommand{\proofname}{Proof}
\fi
\theoremstyle{plain}

\theoremstyle{plain}

\allowdisplaybreaks
\numberwithin{thm}{section}

\makeatother

\providecommand{\lemmaname}{Lemma}
\providecommand{\propositionname}{Proposition}

\begin{document}
\title{Weighted Riemannian Optimization for Solving Quadratic Equations from Gaussian Magnitude  Measurements}
\date{}
\author{Jian-Feng Cai\thanks{J.-F. Cai is with the Department of Mathematics, Hong Kong University of Science and Technology, Clear Water Bay, Kowloon, Hong Kong. Email: jfcai@ust.hk. The work is partially supported by Hong Kong Research Grant Council (HKRGC) GRFs 16307325, 16306124, and 16307023. This work is also supported by Zhejiang Provincial Natural Science Foundation of China under Grant No. LMS25A010008 and NSFC of China under grant number 12271133.},
Huiping Li %
\thanks{H.-P. Li is with the School of Mathematics, Hangzhou Normal University, huipingli@hznu.edu.cn. She is the corresponding author.%
} , Jiayi Li \thanks{J.-Y. Li is with the Department of Mathematics, Hong Kong University of Science and Technology, jligo@connect.ust.hk.
}
}

\maketitle
\begin{abstract}
This paper explores the problem of generalized phase retrieval, which involves reconstructing a length-$n$ signal $\bm{x}$ from its $m$ phaseless samples $y_k = \left|\langle \bm{a}_k,\bm{x}\rangle\right|^2$, where $k = 1,2,...,m$, and $\bm{a}_k$ are the measurement vectors. This problem can be reformulated into recovering a positive semidefinite rank-$1$ matrix $\bm{X}=\bm{x}\bm{x}^*$ from linear samples $\bm{y}=\mathcal{A}(\bm{X})\in\mathbb{R}^m$, thereby requiring us to find a rank-$1$ solution of the linear equations. We demonstrate that several existing phase retrieval algorithms, including Wirtinger Flow (WF) and the canonical Riemannian gradient descent (RGD), actually solve the least-squares fitting of this linear equation on the Riemannian manifold of rank-$1$ matrices, but utilize different metrics on this manifold. Nevertheless, these metrics only allow for a stable and far-apart-from-isometric embedding of rank-$1$ matrices to $\mathbb{R}^m$ by $\mathcal{A}$, resulting in a linear convergence with a considerably large convergence factor. To expedite the convergence, we establish a new metric on the rank-$1$ matrix manifold that facilitates the nearly isometric embedding of rank-$1$ matrices into $\mathbb{R}^m$ through $\mathcal{A}$. A RGD algorithm under this new metric, termed Weighted RGD (WRGD), is proposed to tackle the phase retrieval problem. Owing to the near isometry, we prove that our WRGD algorithm, initialized by spectral methods, can linearly converge to the underlying signal $\bm{x}$ with a small convergence factor. Empirical experiments strongly validate the efficiency and resilience of our algorithms compared to the truncated Wirtinger Flow (TWF) algorithm and the canonical RGD algorithm.
\end{abstract}
\section{Introduction}
In this paper, we are interested in recovering a complex signal $\bm{x} \in \mathbb{C}^n$ from the following system of phaseless equations:
\begin{equation}\label{eq:pf}
y_k=\left|\left\langle \bm{a}_k, \bm{x}\right\rangle\right|^2, \quad k=1, \cdots, m,
\end{equation}
where the data $\bm{y}=\left[y_k\right]_{1 \leq k \leq m}\in\mathbb{R}^n_+$ and the measurement vectors $\bm{a}_k \in \mathbb{C}^n$, $k=1,\ldots,m$, are all known. It is clear that the phase information is missing in \eqref{eq:pf} while only the amplitude is observed due to some physical limitations of optical sensors. This kind of problem is commonly known as the phase retrieval problem, which has been widely applied in many fields, including X-ray crystallography\cite{RefWorks:RefID:7-harrison1993phase,RefWorks:RefID:8-millane1990phase}, diffraction imaging\cite{RefWorks:RefID:3-bandeira2014phase}, microscopy\cite{RefWorks:RefID:6-1987array}, and even quantum mechanics\cite{RefWorks:RefID:5-corbett2006pauli}. 

Significant advancements have been achieved in the theories and algorithms of phase retrieval. In addition to classical algorithms \cite{Fienup:78,Fienup:82,Gerchberg1972APA,Luke:2002}, recent algorithms that provably solve \eqref{eq:pf} have attracted considerable attention \cite{pmlr-v54-bahmani17a,Goldstein:2016,jeong2017convergence,Ma_2019,Netrapalli:2015,tan2018phase,Waldspurger:2015,wang2017solving,Cai_2023,RefWorks:RefID:55-gao2017phaseless}. The PhaseLift method by Candès et al. \cite{RefWorks:RefID:36-candes2013phase,RefWorks:RefID:11-candes2013phaselift}, in particular, has gained significant interest. It involves lifting the unknown vector $\bm{x}$ in $\mathbb{C}^n$ to a rank-$1$ positive semidefinite matrix $\bm{X}:=\bm{x}\bm{x}^*$, transforming the phase retrieval problem \eqref{eq:pf} into a system of linear equations for the matrix $\bm{X}$, as shown below:
\begin{equation}\label{eq:pfmatrix}
y_k=\langle \bm{a}_k\bm{a}_k^*,\bm{X}\rangle, \quad k=1, \cdots, m.
\end{equation}
Let $\mathbb{H}^{n \times n}$ denote the space of all $n \times n$ Hermitian matrices. For $\bm{A}, \bm{B} \in \mathbb{H}^{n \times n}$, define the inner product $\langle \bm{A}, \bm{B} \rangle = \operatorname{tr}(\bm{A}\bm{B})$, which always yields a real value.
With the help of a linear operator defined as $\mathcal{A}~:~\mathbb{H}^{n\times n}\mapsto\mathbb{R}^m$, we can reformulate \eqref{eq:pfmatrix} as
\begin{equation}\label{eq:defA1}
\mathcal{A}(\bm{Z})=\Bigg[\langle \bm{A}_1, \bm{Z} \rangle ,\langle \bm{A}_2, \bm{Z} \rangle ,\ldots,\langle \bm{A}_m, \bm{Z} \rangle \Bigg]^{\top},\quad\forall~\bm{Z}\in\mathbb{H}^{n\times n},
\end{equation}
where $\bm{A}_k:=\bm{a}_k\bm{a}_k^*\in\mathbb{H}^{n\times n}$, $k = 1, 2, \dots, m$, are measurement matrices. PhaseLift then applies convex nuclear norm minimization and positive semidefinite condition to find a positive semidefinite and rank-$1$ solution to satisfy \eqref{eq:pfmatrix}. Namely, 
\[
\min _{\bm{Z} \in \mathbb{H}^{n\times n}}\|\bm{Z}\|_{\ast},\quad \text{s.t.}\quad\mathcal{A}(\bm{Z})=\bm{y},~~\bm{Z}\succeq0.    
\]
Although customized algorithms, e.g., SDP methods, can solve this nuclear norm minimization problem, they suffer from lower computational efficiency due to the increase in the number of unknowns from $n$ to $n^2$.

To enhance computational efficiency, researchers have developed provable non-convex phase retrieval algorithms. One such category of algorithms seeks a solution to the lifted problem \eqref{eq:pfmatrix} within the smooth Riemannian manifold $\mathcal{M}_1$, which comprises all Hermitian rank-$1$ matrices embedded in $\mathbb{H}^{n\times n}$ over $\mathbb{R}$, namely, 
\begin{equation}\label{rm11}
 \min_{\bm{Z} \in \mathcal{M}_1} \frac{1}{2m} \| \mathcal{A} (\bm{Z})- \bm{y} \|_2^2.
\end{equation}
Then a canonical RGD algorithm is applied to solve $\eqref{rm11}$, which is the gradient descent with respect to the canonical metric on the embedded Riemannian manifold $\mathcal{M}_1$. Using the canonical metric from the ambient space $\mathbb{H}^{n\times n}$ over $\mathbb{R}$, the canonical RGD and RCG have been proposed and studied \cite{RefWorks:RefID:34-cai2018solving,RefWorks:RefID:33-li2021riemannian} for phase retrieval problems. By exploiting the special structure of $\mathcal{M}_1$ with the canonical metric, the computational cost of canonical RGD per iteration is significantly reduced --- there is no large-scale singular value decomposition (SVD) computation involved. Meanwhile, when initialized by the spectral method, it is theoretically proved that canonical RGD converges linearly to the underlying positive semidefinite and rank-1 matrix $\bm{X}$ and the convergence factor is dependent on the condition number of the sampling operator $\mathcal{A}$ when restricted on the tangent space of the manifold $\mathcal{M}_1$. All those make canonical RGD one of the most efficient algorithms for solving the phase retrieval problem. Besides, the performance of canonical RGD can be further improved by Riemannian conjugate gradient descent (RCG) algorithms. Another category of non-convex algorithms aims to minimize loss functions, e.g.
\begin{equation}\label{leastsquare2}\nonumber
 \min_{\bm{z} \in \mathbb{C}^{n}}\frac{1}{2m}\sum_{k=1}^{m}( y_k-|\langle\bm{a}_k,\bm{z}\rangle|^2)^2,   
\end{equation}
in the vector space for \eqref{eq:pf}, exemplified by Wirtinger Flow (WF) \cite{RefWorks:RefID:14-candes2015phase} and Truncated Wirtinger Flow (TWF) \cite{RefWorks:RefID:12-chen2017solving,RefWorks:RefID:41-li2022sampling}. Although these methods optimize the vector $\bm{x}$, they can be viewed as non-convex approaches to find a positive semidefinite and rank-$1$ solution of the lifted problem \eqref{eq:pfmatrix}, by factorizing the unknown matrix $\bm{X}$ as $\bm{x}\bm{x}^*$.

Actually, these two categories of non-convex algorithms can be reconciled through the Riemannian optimization methods. We have demonstrated that gradient-type algorithms, including the canonical RGD, a manifold-based gradient descent algorithm, and the WF, a factorization-based gradient descent algorithm, are gradient descent algorithms on the manifold $\mathcal{M}_1$ with different Riemannian metrics and retraction operators, see \cite{RefWorks:RefID:24-jiayili2024phase} for details. We then found that the convergence rate of Riemannian gradient descent algorithms is significantly influenced by the condition number of the sampling operator $\mathcal{A}$ when restricted to the tangent space of the manifold $\mathcal{M}_1$. 
Nonetheless, neither the canonical RGD nor WF employs a metric that optimizes the condition number, resulting in a relatively slow linear convergence rate.  Thus, a natural question is:

\begin{center}
\emph{Can we find a more suitable metric on $\mathcal{M}_1$ that leads to faster convergence of general RGD?}
\end{center}

We answer this question in the affirmative by constructing a new Riemannian metric through analyzing $\mathbb{E}\big[\langle \mathcal{A}(\bm{W}_1), \mathcal{A}(\bm{W}_2) \rangle\big]$ for any $\bm{W}_1,~\bm{W}_2 \in\mathbb{T}_{\bm{Z}}$ and setting up the new metric for the tangent space $\mathbb{T}_{\bm{Z}}$ of $\mathcal{M}_1$ at each point $\bm{Z}$. The resulting algorithm, which is RGD with the new metric, is called Weighted RGD (WRGD). We demonstrate through numerical experiments that WRGD is significantly more efficient and faster than RGD (with canonical metric), TAF, and TWF. On the theoretical side, we also prove that WRGD converges linearly to the underlying rank-1 matrix $\bm{X}$ under random complex Gaussian measurements, when initialized by truncated spectral methods. Under the new metric, the condition numbers remain consistently close to 1 in different models. Crucially, this metric directly forces the contracting factor toward 0 - a property that can be rigorously proven, resulting in accelerated convergence and improved linear convergence rates.

The main contributions of this paper can be summarized as follows: 
\begin{itemize}
\item \emph{Optimal metric.}  We derive a simple, computationally efficient Riemannian metric $\mathfrak{g}$ from the sampling operator $\mathcal{A}$. For each tangent space $\mathbb{T}_{\bm{Z}}$ of the smooth manifold $\mathcal{M}_1$, we define $\mathfrak{g}$ through the expectation:
\[
\mathbb{E}\big[\langle \mathcal{A}(\bm{W}_1), \mathcal{A}(\bm{W}_2) \rangle\big] = \langle \bm{W}_1, \bm{W}_2 \rangle_{\mathfrak{g}} \quad \text{for any} \quad \bm{W}_1, \bm{W}_2 \in\mathbb{T}_{\bm{Z}}.
\]
This weighted metric ensures that when $\langle \mathcal{A}(\bm{W}_1), \mathcal{A}(\bm{W}_2) \rangle$ concentrates strongly around its expectation, we have $ \|\mathcal{A}(\bm{ W})\|_2^2\approx \|\bm{W}\|_{\mathfrak{g}}^2$, for any $\bm{W} \in\mathbb{T}_{\bm{Z}}$, thus the condition number of $\mathcal{A}$ approaches 1. This will directly enable faster convergence. Under this metric, we develop one novel algorithm for phase retrieval: Weighted Riemannian Gradient Descent (WRGD), derived by applying the Riemannian gradient descent (RGD) method.

\item \emph{Recovery guarantee.}  We prove the local convergence of WRGD in terms of the restricted near-isometric property of the sensing operator $\mathcal{A}$. Consequently, WRGD converges swiftly to the rank-1 and positive semidefinite matrix related to the original phase retrieval solution. This is supported by our theoretical and experimental results. Theoretically, we prove that despite the non-convex nature of the algorithm, the proposed WRGD converges linearly to the global minimum with a convergence factor close to $0$, given $m=O(n)$ random complex Gaussian samples. These sample complexities are optimal for the phase retrieval problem. 
\item \emph{Empirical experiments.} We evaluate the performance of the WRGD compared to canonical RGD, WF-typed, and TAF algorithms. Our numerical results demonstrate that WRGD outperforms the canonical RGD, WF, and TAF algorithms in terms of both iteration counts and computational times. As a result, WRGD is much more efficient than RGD with the canonical metric.
\end{itemize}

The paper is organized as follows. In Section~\ref{sec2}, we unify two distinct approaches --- the manifold-based canonical RGD algorithm and the factorization-based WF method --- within a common Riemannian gradient descent framework, and compare their convergence.
In Section \ref{sec3}, we present the truncated WRGD algorithm, along with exact recovery guarantees for the truncated WRGD algorithm. We prove the exact recovery of our algorithms in Sections \ref{sec4} and \ref{sec5}. Section \ref{sec4} presents the necessary technical lemmas, and Section \ref{sec5} contains the proofs of the main results. Section \ref{sec6} is dedicated to a comparative analysis of our algorithms with TAF, TWF, and the canonical Riemannian gradient descent algorithm through a series of numerical experiments. Finally, in Section \ref{sec7}, we conclude the paper and discuss potential future directions.

To state and prove the main results of this paper, we shall introduce some basic notations. Let $\bm{y}^*$ be the conjugate transpose of a vector $\bm{y}\in\mC^n$. We denote the space of all Hermitian matrices by $\mathbb{H}^{n\times n}$. The inner product $\langle\cdot,\cdot \rangle$ denotes the standard inner product, and the $\operatorname{tr}(\cdot)$ indicates the trace of a matrix. For the matrix-valued operator $\mathcal{M}$, the operator norm corresponding to the Frobenius norm becomes
\begin{equation} \label{eq:13}
\|\mathcal{M}\|=\sup _{\bm{W} \in \mathbb{H}^{n \times n} }\frac{|\langle\bm{W},\mathcal{M} (\bm{W})\rangle|}{\|\bm{W}\|_F^2}.
\end{equation}
Similarly, the operator norm according to the newly defined metric $\mathfrak{g}$ becomes
\begin{equation}\label{eq:14}
 \|\mathcal{M}\|_{o p}^{\mathfrak{g}}=\sup _{\bm{W} \in \mathbb{H}^{n \times n} }\frac{|\langle \mathcal{M} (\bm{W}), \bm{W} \rangle_{\mathfrak{g}}|}{\|\bm{W}\|_{\mathfrak{g}}^2}.   
\end{equation}
We denote $\mathcal{I}$ as the identity operator from $\mathbb{C}^{n\times n}$ to $\mathbb{C}^{n\times n}$. We exploit $``\geq,~ \leq"$ and $``\succeq,~\preceq"$ to show the positive or negative semi-definiteness of matrices and matrix-valued operators, respectively. $A\lesssim B$ denotes that $A$ is approximately equal to $B$. $\{\bm{e}_j\}_{1\leq j\leq n}$ of $\mathbb{C}^n$ are the standard orthogonal basis. Denote $\|\cdot\|_{\ast},~\|\cdot\|_{F},~\|\cdot\|$ as the nuclear norm, Frobenius norm, and operator norm of matrices. Also, we denote $\|\cdot\|_{p}$ as the $\ell_{p}$-norm of vectors. 

\section{General Riemannian Gradient Descent Algorithms}\label{sec2}
In this section, we develop a general Riemannian gradient descent method unifying manifold-based and factorization-based approaches for solving the phase retrieval problem. To begin with, we review the Riemannian gradient descent method for solving general optimization problems on the Riemannian manifold. Then, by selecting distinct metrics and retraction operators, we found that both the canonical RGD and the WF algorithms can be unified as the Riemannian gradient descent algorithm. Finally, we compare their convergence speed using different metrics on the Riemannian manifold $\mathcal{M}_1$, which comprises all Hermitian rank-$1$ matrices embedded in $\mathbb{H}^{n\times n}$ over $\mathbb{R}$.

We present that a general optimization on the Riemannian manifold  $\mathcal{M}$ embedded in a linear space is
\begin{equation}\label{GM-model1}
\min_{\bm{Z} \in \mathcal{M}} F(\bm{Z}),
\end{equation}
where $F:\mathbb{E}\mapsto\mathbb{R}$ is an objective function. The Riemannian gradient descent (RGD) \cite{boumal2023introduction} is a popular first-order algorithm for solving \eqref{GM-model1}. For general RGD method, starting from an initial guess $\bm{Z}_0$, RGD generates a sequence of iterations by the following update rule:
\begin{equation}\label{eq:RGDgeneral}
\bm{Z}_{t+1} = \mathcal{R}_{\bm{Z}_t} \left( \bm{Z}_t - \alpha_t \nabla^{(\mathfrak{g})}_{\mathcal{M}} F(\bm{Z}_t) \right),\quad t=0,1,2,\ldots,
\end{equation}
where $\alpha_t > 0$ is the step size, $\nabla^{(\mathfrak{g})}_{\mathcal{M}} F(\bm{Z}_t)$ is the gradient of $F$ at $\bm{Z}_t$ with respect to the Riemannian metric $\mathfrak{g}$ of the manifold $\mathcal{M}$, and $\mathcal{R}_{\bm{Z}_t}$ is a retraction operator. At each iteration, RGD first performs one step of the standard gradient descent in the tangent space of $\mathcal{M}$ at $\bm{Z}_t$, and then it retracts the iterate from the tangent space back to $\mathcal{M}$ by the retraction operator $\mathcal{R}_{\bm{Z}_t}$. Here, $\mathcal{R}_{\bm{Z}_t}$ is a retraction operator that maps a point on the tangent space $\mathbb{T}_{\bm{Z}}\mathcal{M}$ of $\mathcal{M}$ at $\bm{Z}_t$, back to a point on $\mathcal{M}$. We present the formal definition of \textit{retraction} below.
\begin{definition}[{\cite[Chapter 4.1]{absil2008optimization}}]\label{defre}
Given $\bm{Z} \in \mathcal{M}$, a retraction is a smooth mapping $\mathcal{R}_{\bm{Z}}:\mathbb{T}_{\bm{Z}}\mathcal{M} \rightarrow \mathcal{M}$ with 
\begin{enumerate}
    \item $\mathcal{R}_{\bm{Z}}(\bm{Z})=\bm{Z}$, and
    \item $\mathcal{J}_{\mathcal{R}_{\bm{Z}}}(\bm{Z})= \mathcal{I} $ , where $ \mathcal{J}_{\mathcal{R}_{\bm{Z}}}(\bm{Z})$ denotes the differential of $\mathcal{R}_{\bm{Z}}$ at the point $\bm{Z}$.
\end{enumerate}
\end{definition}

\subsection{Canonical Riemannian Gradient Descent}
Based on the lifting technique and the linear operator $\mathcal{A}$, the phase retrieval problem \eqref{eq:pf} can be regarded as seeking a solution to the lifted problem within the smooth manifold $\mathcal{M}_1$, namely, 
\begin{equation}\label{eq:pfmatrix1}
\min_{\bm{Z} \in \mathcal{M}_1} \frac{1}{2m} \| \mathcal{A} (\bm{Z})- \bm{y} \|_2^2, 
\end{equation}
where $\mathcal{M}_1$ comprises all Hermitian rank-$1$ matrices embedded in $\mathbb{H}^{n\times n}$ over $\mathbb{R}$ and thus forms a smooth manifold. Equipping $\mathcal{M}_1$ with the canonical ambient metric from $\mathbb{H}^{n\times n}$ and choosing a retraction operator $\mathcal{R}_{\bm{Z}}$, we implement \eqref{eq:RGDgeneral} to obtain the canonical RGD algorithm, as introduced in \cite{RefWorks:RefID:34-cai2018solving,RefWorks:RefID:33-li2021riemannian} for solving \eqref{eq:pfmatrix1}. 
\begin{itemize}
    \item \textit{Riemannian metric and Riemannian gradient.} For $\bm{Z}\in {\mathcal{M}_1}$, the tangent space of $\mathcal{M}_1$ at any rank-$1$ matrix $\bm{Z} = \beta\bm{u}\bm{u}^*$ with $\bm{u}\in \mathbb{C}^n$, $\|\bm{u}\|_2=1$, and $\beta\in\mathbb{R}\setminus\{0\}$ is
\begin{equation}\label{tangent_space11}
\mathbb{T}_{\bm{Z}} = \{\bm{u}\bm{w}^*+\bm{w}\bm{u}^*~|~\bm{w}\in\mathbb{C}^{n}\}.
\end{equation}
Then we define the inner product on ${\mathbb{T}_{\bm{Z}}}$ to align with the canonical inner product of the ambient space $\mathbb{H}^{n\times n}$ over $\mathbb{R}$, i.e.,
\begin{equation}\label{Inner_product}
\langle\bm{ A}, \bm{B}\rangle = \operatorname{tr}(\bm{A}^*\bm{B}),\qquad\forall~\bm{A},\bm{B}\in\mathbb{T}_{\bm{Z}}, 
\end{equation}
which is always real.
The Riemannian gradient of the function $F(\bm{Z}):=\frac{1}{2m} \| \mathcal{A} (\bm{Z})- \bm{y} \|_2^2$ at $\bm{Y}\in\mathbb{T}_{\bm{Z}}$ is thus given by
\begin{equation}\label{M1}
\nabla_{\mathcal{M}_1} F(\bm{Y})=\mathcal{P}_{\mathbb{T}_{\bm{Z}}} \nabla F(\bm{Y}),
\end{equation}
where $\nabla F$ is the Euclidean gradient of $F:\mathbb{H}^{n\times n}\mapsto\mathbb{R}$ and $\mathcal{P}_{\mathbb{T}_{\bm{Z}}}:\mathbb{H}^{n\times n}\mapsto\mathbb{T}_{\bm{Z}}$ represents the orthogonal projector, both with respect to the canonical inner product in $\mathbb{H}^{n\times n}$ over $\mathbb{R}$. Utilizing the definition of the tangent space $\mathbb{T}_{\bm{Z}}$ as given in $(\ref{tangent_space11})$, the orthogonal projection onto $\mathbb{T}_{\bm{Z}}$ can be represented by
\begin{equation}\label{ortho projector}
\mathcal{P}_{\mathbb{T}_{\bm{Z}}}(\bm{W})= \bm{uu}^{*}\bm{W}+\bm{W}\bm{uu}^{*}-\bm{uu}^{*}\bm{W}\bm{uu}^{*}, ~~\forall~\bm{W}\in \mathbb{H}^{n\times n}.
\end{equation}
\item \textit{Retraction.} For any $\bm{Z}_t$, we select $\mathcal{R}_{\bm{Z}_t}$ to be the $1$-truncated SVD operator $\mathcal{H}_1$, which finds the best rank-$1$ approximation and acts as the projection onto $\mathcal{M}_1$ under the canonical metric. This retraction ensures that a point from the tangent space $\mathbb{T}_{\bm{Z}_t}$ is retracted back to the Riemannian manifold $\mathcal{M}_1$. It can be proved that $\mathcal{H}_1$ satisfies all the conditions of Definition $\ref{defre}$.
\end{itemize}
Consequently, \eqref{eq:RGDgeneral} is formulated as the following, which is the canonical RGD
$$
\begin{aligned}
\bm{Z}_{t+1} & =\mathcal{H}_1\big(\bm{Z}_t-\alpha_t\mathcal{P}_{\mathbb{T}_{\bm{Z}_t}} \nabla F\left(\bm{Z}_t\right)\big) \\
& =\mathcal{H}_1 \big(\mathcal{P}_{\mathbb{T}_{\bm{Z}_t}}(\bm{Z}_t-\frac{\alpha_t}{m} \mathcal{A}^*\left(\mathcal{A} (\bm{Z}_t)-\bm{y})\right) \big)
\end{aligned}
$$
Considering the specific probabilistic models for the sampling operator $\mathcal{A}$, we replace the full Euclidean gradient $\nabla F\left(\bm{Z}_t\right)=\frac{1}{m}\sum_{k=1}^{m}\left(\langle\bm{a}_k\bm{a}_k^*,\bm{Z}_t\rangle-y_k\right)\bm{a}_k\bm{a}_k^*$ with its truncated counterpart $\nabla F_t\left(\bm{Z}_t\right)=\frac{1}{m}\sum_{k\in \mathbb{I}_t}\left(\langle\bm{a}_k\bm{a}_k^*,\bm{Z}_t\rangle-y_k\right)\bm{a}_k\bm{a}_k^*$, where $\mathbb{I}_t\subset\{1,\ldots,m\}$ is determined by certain truncation rules in accordance with the probabilistic model of $\mathcal{A}$. This leads to the algorithms proposed in \cite{RefWorks:RefID:34-cai2018solving,RefWorks:RefID:33-li2021riemannian}. We refer to them as the canonical RGD algorithm, as they employ the canonical metric derived from the ambient space of $\mathcal{M}_1$. Moreover, a similar derivation for the canonical RCG algorithm can be found in \cite{RefWorks:RefID:33-li2021riemannian}.

\subsection{Wirtinger Flow Algorithm as Riemannian Gradient Descent}
Different from the canonical RGD, the Wirtinger flow (WF) algorithm \cite{RefWorks:RefID:14-candes2015phase,RefWorks:RefID:12-chen2017solving}  solves the phase retrieval problem in the vector space $\mathbb{C}^n$. Specifically, WF finds the solution through the following least squares problem in $\mathbb{C}^n$ 
\begin{equation} \nonumber
\min_{\bm{z} \in \mathbb{C}^{n}}\frac{1}{2m}\sum_{k=1}^{m}( y_k-|\langle\bm{a}_k,\bm{z}\rangle|^2)^2
\end{equation}
by Wirtinger gradient flow
\begin{equation}\label{eq:wf2}
\begin{split}
\bm{z}_{t+1}=\bm{z}_t-\frac{\alpha_t}{\left\|\bm{z}_t\right\|^2_2} \nabla f\left(\bm{z}_t\right),
\end{split}
\end{equation} 
where $\alpha_t$ is the stepsize and $\nabla f$ is the Wirtinger gradient of $f(\bm{z}):=\frac{1}{2m}\sum_{k=1}^{m}(y_k-|\langle\bm{a}_k,\bm{z}\rangle|^2)^2$.

Although the WF algorithm is formulated to solve the least squares problem in vector space, it implicitly operates on the matrix space in the manifold $\mathcal{M}_1$. In fact, the algorithm can be interpreted as a special case of Riemannian gradient descent for phase retrieval, by choosing a proper Riemannian metric on the smooth manifold $\mathcal{M}_1$ and a retraction operator. 

To demonstrate this, we reformulate \eqref{eq:wf2} in matrix space. Define $\bm{Z}_{t} = \bm{z}_{t}\bm{z}_{t}^*$. Then the update rule becomes:
\begin{equation}\label{wf2}
\begin{split}
\bm{Z}_{t+1} &= \bm{z}_{t+1}\bm{z}_{t+1}^* \\
&= \Big(\bm{z}_t - \frac{\alpha_t}{\|\bm{z}_t\|^2_2} \nabla f(\bm{z}_t)\Big)\Big(\bm{z}_t - \frac{\alpha_t}{\|\bm{z}_t\|^2_2} \nabla f(\bm{z}_t)\Big)^* \\
&= \bm{Z}_t - \frac{\alpha_t}{\|\bm{z}_t\|^2_2} \left(\nabla f(\bm{z}_t)\bm{z}_t^* + \bm{z}_t\nabla f(\bm{z}_t)^*\right) + \frac{\alpha_t^2}{\|\bm{z}_t\|_2^4}\nabla f(\bm{z}_t)\nabla f(\bm{z}_t)^* \\
&= \bm{Z}_t - 2\alpha_t \left(\nabla F(\bm{Z}_t)\bm{u}_t\bm{u}_t^* + \bm{u}_t\bm{u}_t^*\nabla F(\bm{Z}_t)\right) + \frac{\alpha_t^2}{\|\bm{z}_t\|_2^4}\nabla f(\bm{z}_t)\nabla f(\bm{z}_t)^*,
\end{split}
\end{equation}
where $\bm{u}_t := \bm{z}_t/\|\bm{z}_t\|_2$ is the normalized direction of $\bm{z}_t$, and the final equality follows from the identity $\nabla f(\bm{z}_t) = 2\nabla F(\bm{Z}_t)\bm{z}_t$.
\begin{itemize}
\item \textit{Riemannian metric and Riemannian gradient.}  We define a linear weighting operator $\mathcal{W}: \mathbb{H}^{n \times n} \mapsto \mathbb{H}^{n \times n}$ by
    \[
    \mathcal{W}(\bm{Z}) = \bm{Z} - \frac{1}{2}\langle \bm{Z}, \bm{I}_n \rangle \bm{I}_n, \quad \forall~ \bm{Z} \in \mathbb{H}^{n \times n},
    \]
    where $\bm{I}_n$ denotes the $n$-dimensional identity matrix.
Then for the ambient space $\mathbb{H}^{n\times n}$ over $\mathbb{R}$, we can define a pseudo-inner product as following
\begin{equation}\label{metric2}
\langle \bm{A}, \bm{B} \rangle_\mathfrak{w}= \langle \mathcal{W}(\bm{A}),\bm{B} \rangle = \langle \bm{A},\bm{B}\rangle - \frac{1}{2}\operatorname{tr}(\bm{A})\operatorname{tr}(\bm{B}),  \quad \forall~ \bm{A}, \bm{B} \in \mathbb{H}^{n\times n}.  
\end{equation}
Note that $\langle \cdot, \cdot \rangle_{\mathfrak{w}}$ does not define an inner product on the ambient space $\mathbb{H}^{n\times n}$, since $\langle \bm{A}, \bm{A} \rangle_{\mathfrak{w}}$ can be negative. However, when restricted to the tangent space $\mathbb{T}_{\bm{Z}}$ (for any $\bm{Z} \in \mathcal{M}_1$), it becomes non-negative and thereby induces a Riemannian metric on $\mathcal{M}_1$.
Equipped with this Riemannian metric, $\mathcal{M}_1$ is a smooth Riemmanian manifold embedded in $\mathbb{H}^{n\times n}$ over $\mathbb{R}$. 

Now let us calculate the Riemannian gradient $\nabla_{\mathcal{M}_1}^{(\mathfrak{w})}F$ of the function $F$ on $\mathcal{M}_1$ under $\langle\cdot,\cdot\rangle_\mathfrak{w}$. For this purpose, we first introduce a lemma.
\begin{lemma}\label{lemma:trace}
Let $\bm{Z} = \beta\bm{u}\bm{u}^* \in \mathcal{M}_1$ where $\beta \in \mathbb{R}\setminus\{0\}$ and $\bm{u} \in \mathbb{C}^n$ with $\|\bm{u}\|_2 = 1$.
Then, for any $\bm{Y} \in \mathbb{T}_{\bm{Z}}$, we have
\[
\operatorname{tr}(\bm{Y}) = \bm{u}^*\bm{Y}\bm{u}.
\]
\end{lemma}
\begin{proof}
Since $\bm{Y}\in\mathbb{T}_{\bm{Z}}$, it follows from \eqref{tangent_space11} that $\langle \bm{Y},\bm{I}_n - \bm{uu}^*\rangle=0$. With it, 
a direct calculation gives $\operatorname{tr}(\bm{Y}) = \langle \bm{Y}, \bm{I}_n \rangle 
= \langle \bm{Y}, \bm{uu}^* + (\bm{I}_n - \bm{uu}^*) \rangle 
= \langle \bm{Y}, \bm{uu}^* \rangle = \bm{u}^*\bm{Y}\bm{u}$.
\end{proof}
Consider a smooth curve $\bm{Z}(s) \in \mathcal{M}_1$ for $s \in \mathbb{R}$ satisfying $\bm{Z}(0) = \bm{Z}=\beta\bm{u}\bm{u}^*$. 
Differentiating $F(\bm{Z}(s))$ with respect to $s$ at $s=0$ gives
\begin{equation}\label{differential}
\begin{split}
&\left.\frac{d}{ds} F(\bm{Z}(s))\right|_{s=0} = \langle \dot{\bm{Z}}(0), \nabla F(\bm{Z}) \rangle\\
&= \langle \dot{\bm{Z}}(0), \bm{uu}^*\nabla F(\bm{Z}) + \nabla F(\bm{Z})\bm{uu}^* \rangle 
 -\tfrac{1}{2}\langle \dot{\bm{Z}}(0), \bm{uu}^*(\bm{uu}^*\nabla F(\bm{Z}) + \nabla F(\bm{Z})\bm{uu}^*)\bm{uu}^* \rangle \\
&= \langle \dot{\bm{Z}}(0), \bm{uu}^*\nabla F(\bm{Z}) + \nabla F(\bm{Z})\bm{uu}^* \rangle -\tfrac{1}{2}\big(\bm{u}^*\dot{\bm{Z}}(0)\bm{u}\big)\cdot\big(\bm{u}^*(\bm{uu}^*\nabla F(\bm{Z}) + \nabla F(\bm{Z})\bm{uu}^*) \bm{u}\big) \\
&= \langle \dot{\bm{Z}}(0), \bm{uu}^*\nabla F(\bm{Z}) + \nabla F(\bm{Z})\bm{uu}^* \rangle  -\tfrac{1}{2}\operatorname{tr}(\dot{\bm{Z}}(0))\cdot\operatorname{tr}(\bm{uu}^*\nabla F(\bm{Z}) + \nabla F(\bm{Z})\bm{uu}^*) \\
&= \langle \dot{\bm{Z}}(0), \bm{uu}^*\nabla F(\bm{Z}) + \nabla F(\bm{Z})\bm{uu}^*\rangle_{\mathfrak{w}},
\end{split}
\end{equation}
where the second equality follows from $\langle \dot{\bm{Z}}(0), \bm{I}_n-\bm{uu}^*\rangle =0$ as $\dot{\bm{Z}}(0) \in \mathbb{T}_{\bm{Z}}$, 
and the fourth equality is due to Lemma \ref{lemma:trace}.
Since $\bm{uu}^*\nabla F(\bm{Z}) + \nabla F(\bm{Z})\bm{uu}^*\in\mathbb{T}_{\bm{Z}}$, \eqref{differential} implies 
\begin{equation}\label{gradient:wf}
\nabla_{\mathcal{M}_1}^{(\mathfrak{w})} F(\bm{Z})=\bm{uu}^*\nabla F(\bm{Z})+ \nabla F(\bm{Z})\bm{uu}^*.    
\end{equation}

\item\textit{Retraction.} For any $\bm{Z}=\sigma\bm{z}\bm{z}^*\in\mathcal{M}_1$ with $\bm{z}\in\mathbb{C}^n$ and $\sigma =\pm1$, we first define an operator $\mathcal{S}_{\bm{Z}}:\mathbb{T}_{\bm{Z}}\mapsto\mathcal{M}_1$ as follows. For any $\bm{\Xi} \in \mathbb{T}_{\bm{Z}}$, when it can be parametrized as
\begin{equation}\label{eq:repXiw}
    \bm{\Xi} = \bm{Z}+\bm{z}\bm{w}^* + \bm{w}\bm{z}^*, \quad\mbox{~for~some~}\bm{w} \in \mathbb{C}^n \mbox{~s.t.~}\bm{z}^*\bm{w}\in\mathbb{R},
\end{equation}
we define
\begin{equation}\label{retraction2}
\mathcal{S}_{\bm{Z}}(\bm{\Xi}) = \sigma(\bm{z} + \sigma\bm{w})(\bm{z} + \sigma\bm{w})^*, \quad \forall~\bm{\Xi}\in \mathbb{T}_{\bm{Z}}.
\end{equation}
\begin{thm}
The operator $\mathcal{S}_{\bm{Z}}:\mathbb{T}_{\bm{Z}}\mapsto\mathcal{M}_1$
is well-defined and a retraction operator according to Definition \ref{defre}.
\end{thm}
\begin{proof}
To show that $\mathcal{S}_{\bm{Z}}$ is well-defined, it suffices to prove the existence and uniqueness of the parametrization \eqref{eq:repXiw}.
\begin{itemize}
\item \textit{Existence:}
Since $\bm{\Xi},\bm{Z} \in \mathbb{T}_{\bm{Z}}$, the expression of the tangent space \eqref{tangent_space11} ensures that there exists $\bm{v} \in \mathbb{C}^n$ such that $\bm{\Xi} = \bm{Z} + \bm{z}\bm{v}^* + \bm{v}\bm{z}^*$. Decompose $\bm{v} = c\bm{z} + \bm{z}_\perp$ orthogonally where $c \in \mathbb{C}$ and $\bm{z}^*\bm{z}_{\perp}=0$. Then, $\bm{\Xi} = \bm{Z} + (\bar{c}+c)\bm{z}\bm{z}^* + \bm{z}\bm{z}_{\perp}^* + \bm{z}_{\perp}\bm{z}^*$. Taking $\bm{w} = \frac{\bar{c}+c}{2}\bm{z} + \bm{z}_\perp$ yields $\bm{\Xi} = \bm{Z} + \bm{z}\bm{w}^* + \bm{w}\bm{z}^*$ and $\bm{z}^*\bm{w} = \frac{c + \bar{c}}{2}\|\bm{z}\|_2^2\in\mathbb{R}$.

\item \textit{Uniqueness:}
Suppose $\bm{\Xi}$ admits two parametrizations \eqref{eq:repXiw} with $\bm{w}=\bm{w}_1,\bm{w}_2$ respectively, i.e.,
$\bm{\Xi} = \bm{Z} + \bm{z}\bm{w}_1^* + \bm{w}_1\bm{z}^*= \bm{Z} + \bm{z}\bm{w}_2^* + \bm{w}_2\bm{z}^*$
with $\bm{z}^*\bm{w}_i \in \mathbb{R}$ for $i=1,2$. Taking the difference yields
\begin{equation}\label{eq: diff tangent vector}
\bm{z}(\bm{w}_1 - \bm{w}_2)^* + (\bm{w}_1 - \bm{w}_2)\bm{z}^* = \bm{0}.    
\end{equation}
Let $c_i = \bm{z}^*\bm{w}_i \in \mathbb{R}$ and decompose $\bm{w}_i$ orthogonally as $\bm{w}_i = \frac{c_i}{\|\bm{z}\|_2^2}\bm{z} + \big(\bm{I}_n - \frac{\bm{z}\bm{z}^*}{\|\bm{z}\|_2^2}\big)\bm{w}_i$.
Substituting it into \eqref{eq: diff tangent vector} gives
\[
2\frac{c_1 - c_2}{\|\bm{z}\|_2^2}\bm{z}\bm{z}^* + \bm{z}(\bm{w}_1 - \bm{w}_2)^*\Big(\bm{I}_n - \frac{\bm{z}\bm{z}^*}{\|\bm{z}\|_2^2}\Big) + \Big(\bm{I}_n - \frac{\bm{z}\bm{z}^*}{\|\bm{z}\|_2^2}\Big)(\bm{w}_1 - \bm{w}_2)\bm{z}^* = \bm{0}.
\]
Since the three terms on the left-hand side are orthogonal, they must all be the zero matrix. This implies $c_1 = c_2$ and $\big(\bm{I}_n - \frac{\bm{z}\bm{z}^*}{\|\bm{z}\|_2^2}\big)(\bm{w}_1 - \bm{w}_2) = \bm{0}$. Thus, $\bm{w}_1 = \bm{w}_2$. 
\end{itemize}

Next, we show that $\mathcal{S}_{\bm{Z}}$ is a retraction operator according to checking conditions in Definition \ref{defre}. 
\begin{itemize}
\item 
When $\bm{\Xi} = \bm{Z}$, because $\bm{w}$ is unique in \eqref{eq:repXiw}, we must have $\bm{w} = \bm{0}$. We immediately verified $\mathcal{S}_{\bm{Z}}(\bm{Z}) = \bm{Z}$. 
\item 
Then, we show $\mathcal{J}_{\mathcal{S}_{\bm{Z}}}(\bm{Z}) = \mathcal{I}$. Consider the smooth mapping $\mathcal{G}: \mathbb{C}^{n}\rightarrow \mathbb{T}_{\bm{Z}}$ defined by $\mathcal{G}(\bm{w}) := \sigma\bm{zz}^* + \bm{zw}^* + \bm{wz}^*$. Let $\bm{w}(t)$ be a smooth curve in $\mathbb{C}^n$ parametrized by $t\in\mathbb{R}$ with $\bm{w}(0)=\bm{w}$. The differential of $\mathcal{G}$ is given by 
$$
\mathcal{J}_{\mathcal{G}}(\bm{w})[\dot{\bm{w}}] = \bm{z}\dot{\bm{w}}^* + \dot{\bm{w}}\bm{z}^*.
$$
We also have $\mathcal{S}_{\bm{Z}}\big(\mathcal{G}(\bm{w})\big)= \sigma(\bm{z} + \sigma\bm{w})(\bm{z} + \sigma\bm{w})^*$, and its differential is  
\[
\mathcal{J}_{\mathcal{S}_{\bm{Z}} \circ \mathcal{G}}(\bm{w})[\dot{\bm{w}}] = \dot{\bm{w}}(\bm{z}+\sigma\bm{w})^* + (\bm{z}+\sigma\bm{w})\dot{\bm{w}}^*.
\]

The chain rule applied to $\mathcal{S}_{\bm{Z}} \circ \mathcal{G}$ yields 
\begin{equation}\label{eq: chain rule retracion}
\mathcal{J}_{\mathcal{S}_{\bm{Z}} \circ \mathcal{G}}(\bm{w})[\dot{\bm{w}}] = \mathcal{J}_{\mathcal{S}_{\bm{Z}}}(\mathcal{G}(\bm{w}))\big[\mathcal{J}_{\mathcal{G}}(\bm{w})[\dot{\bm{w}}]\big].  
\end{equation}

Taking $\bm{w} = \bm{0}$ in \eqref{eq: chain rule retracion} yields $\bm{z}\dot{\bm{w}}^* + \dot{\bm{w}}\bm{z}^* =\mathcal{J}_{\mathcal{S}_{\bm{Z}}}(\bm{Z})[\bm{z}\dot{\bm{w}}^* + \dot{\bm{w}}\bm{z}^*]$ for any $\dot{\bm{w}} \in \mathbb{C}^n$, which implies $\mathcal{J}_{\mathcal{S}_{\bm{Z}}}(\bm{Z})$ acts as the identity on $\mathbb{T}_{\bm{Z}}$.
\end{itemize}
\end{proof}
When we take $\bm{w}=-\alpha_t\nabla f(\bm{z}_t)/\|\bm{z}_t\|_2^2$, we have $\bm{z}_t^*\bm{w}=-\alpha_t\bm{z}_t^*\nabla F(\bm{Z}_t)\bm{z}_t/\|\bm{z}_t\|_2^2 \in \mathbb{R} $. Then, by \eqref{wf2} and the definition of $\mathcal{S}_{\bm{Z}_t}$ in \eqref{retraction2},
\begin{equation}\label{eq:retraction}
 \mathcal{S}_{\bm{Z}_t}\Big(\bm{Z}_t-\frac{\alpha_t}{\|\bm{z}_t\|_2^2}\left(\nabla f(\bm{z}_t)\bm{z}_t^*+\bm{z}_t\nabla f(\bm{z}_t)^*\right)\Big) = \bm{Z}_{t+1}.   
\end{equation}
\end{itemize}

With the reformulation \eqref{wf2} of the WF algorithm, the retraction \eqref{eq:retraction}, and the Riemannian gradient \eqref{gradient:wf}, we further reformulate the WF algorithm as a special case of RGD as follows
$$
\bm{Z}_{t+1} = \mathcal{S}_{\bm{Z}_t}\big(\bm{Z}_t - \alpha_t \nabla_{\mathcal{M}}^{(\mathfrak{w})} F(\bm{Z}_t)\big).
$$
Same as that in the canonical RGD, we can replace the Riemannian gradient $\nabla_{\mathcal{M}_1}^{(\mathfrak{w})}  F\left(\bm{Z}_t\right)$ with its truncated counterpart $\nabla_{\mathcal{M}_1}^{(\mathfrak{w})}  F_t\left(\bm{Z}_t\right)$, where $F_t(\bm{Z}_t)=\frac{1}{2m}\sum_{k\in\mathbb{I}_t}( y_k-\langle\bm{a}_k\bm{a}_k^*,\bm{Z}_t\rangle)^2$  with $\mathbb{I}_t\subset\{1,\ldots,m\}$ determined by specific truncation rules following the probabilistic model of $\mathcal{A}$. This leads to the so-called algorithm Truncated Wirtinger flow proposed in \cite{RefWorks:RefID:41-li2022sampling}.
\subsection{Convergence speed}\label{subsection:convergence speed}
As we have seen previously, by selecting different Riemannian metrics and retractions, both WF-typed and canonical RGD algorithms are special cases of the RGD on the smooth manifold $\mathcal{M}_1$ for solving the phase retrieval problem. 
Specifically, we reformulate the phase retrieval problem as finding the least squares solution on $\mathcal{M}_1$:
\begin{equation}\label{least_square}
\min _{\bm{Z} \in \mathcal{M}_1}F(\bm{Z}),\quad \text{where}~F(\bm{Z}):=\frac1{2m}\|\mathcal{A}(\bm{Z})-\bm{y}\|_2^2.    
\end{equation}
When $\mathcal{M}_1$ is equipped with a Riemannian metric $\langle\cdot,\cdot\rangle_{\mathfrak{g}}$ and retraction $\mathcal{R}_{\bm{Z}} : \mathbb{T}_{\bm{Z}}\mathcal{M}_1 \mapsto \mathcal{M}_1$, the RGD for solving \eqref{least_square} iterates as:
\begin{equation}\label{riemannian_framework}
\bm{Z}_{t+1} = \mathcal{R}_{\bm{Z}_t}\left(\bm{Z}_t - \alpha_t \nabla_{\mathcal{M}_1}^{(\mathfrak{g})} F(\bm{Z}_t)\right),
\end{equation}
where $\alpha_t > 0$ is the step size, and $\nabla_{\mathcal{M}_1}^{(\mathfrak{g})} F$ denotes the Riemannian gradient of $F$ on $\mathcal{M}_1$ under the metric $\langle\cdot,\cdot\rangle_{\mathfrak{g}}$. This update consists of two steps:
\begin{equation}\label{eq:riemannian_update}
\begin{cases}
\bm{W}_t = \bm{Z}_t - \alpha_t \nabla_{\mathcal{M}_1}^{(\mathfrak{g})} F(\bm{Z}_t), \\
\bm{Z}_{t+1} = \mathcal{R}_{\bm{Z}_t}(\bm{W}_t),
\end{cases}
\end{equation}
where $\bm{W}_t$ is obtained by performing one step of gradient descent on the linearized problem 
$
\min_{\bm{Z} \in \mathbb{T}_{\bm{Z}_t}\mathcal{M}_1} F(\bm{Z})
$ staring from $\bm{Z}_t$, and $\bm{Z}_{t+1}$ results from retracting $\bm{W}_t$ back to $\mathcal{M}_1$ via $\mathcal{R}_{\bm{Z}_t}$.

As a result, the convergence speed of RGD is principally determined by two key factors. The first step of \eqref{eq:riemannian_update} is the gradient descent of a linear least square problem, which yields linear convergence at rate $1 - O(\kappa_{\mathfrak{g}}^{-1})$. Here $\kappa_{\mathfrak{g}} = \frac{C_{u,\mathfrak{g}}}{C_{l,\mathfrak{g}}}$ is the condition number of the sensing operator $\frac{1}{\sqrt{m}}\mathcal{A}$, where $C_{u,\mathfrak{g}}$ and $C_{l,\mathfrak{g}}$ are constants in the following inequality:
\begin{equation}\label{tangent_space_condition}
C_{l,\mathfrak{g}}\|\bm{W}\|_{\mathfrak{g}}^2 \leq \frac{1}{m}\|\mathcal{A}(\bm{W})\|_2^2 \leq C_{u,\mathfrak{g}}\|\bm{W}\|_{\mathfrak{g}}^2, \quad \forall~ \bm{W} \in \mathbb{T}_{\bm{Z}_t}.
\end{equation}
In the second step of \eqref{eq:riemannian_update}, the retraction operator $\mathcal{R}_{\bm{Z}_t}$ contributes only a second-order perturbation \cite{absil2008optimization} on $\bm{W}_t$ to obtain $\bm{Z}_{t+1}$. This higher-order effect is dominated by the first-order convergence dynamics governed by $\kappa_{\mathfrak{g}}$. Thus, the convergence behavior of RGD is dominated by the first substep of \eqref{eq:riemannian_update}. It will converge linearly with a rate determined by $\kappa_\mathfrak{g}$. A smaller $\kappa_\mathfrak{g}$ results in faster convergence. A complete theoretical analysis can be found in the later sections of this paper.

As both the TWF and canonical RGD are special cases of RGD, we can compare their convergence speeds by comparing their condition number $\kappa_\mathfrak{g}$.
To begin our analysis, we first introduce our measurement model, random complex Gaussian measurements, as follows. We choose $\bm{a}_k$ in \eqref{eq:pf} follows the complex i.i.d Gaussian distribution as
\begin{equation}\label{eq:Gaussian}
\bm{a}_k\mathop{\sim}^{i.i.d}\mathcal{N}(0,\bm{I}_n/2)+\imath\cdot\mathcal{N}(0,\bm{I}_n/2), \quad k=1, \ldots, m,
\end{equation}
where $\imath$ is the imaginary unit.
Then, under Gaussian measurement models, we have\cite{RefWorks:RefID:14-candes2015phase,RefWorks:RefID:34-cai2018solving}:
\begin{equation}\label{eq:exp}
    \mathbb{E}\left(\frac{1}{m}\|\mathcal{A}(\bm{W})\|_2^2\right) = \operatorname{tr}(\bm{W})^2 + \|\bm{W}\|_F^2, \quad \forall~ \bm{W} \in \mathbb{H}^{n \times n}.
\end{equation}
When the number of measurements $m$ is sufficiently large, $\frac{1}{m}\|\mathcal{A}(\bm{W})\|_2^2 \approx \operatorname{tr}(\bm{W})^2 + \|\bm{W}\|_F^2$. Thus, we can estimate the condition numbers $\kappa_{\mathfrak{g}}$ of different algorithms.
\begin{itemize}
    \item \emph{Canonical RGD.} The Riemannian metric is the canonical metric, and the induced norm is the Frobenius norm. Then we have
$$
\|\bm{W}\|_{F}^2\lesssim \frac{1}{m}\|\mathcal{A}(\bm{W})\|_2^2\lesssim 2\|\bm{W}\|_{F}^2,\qquad
\forall~ \bm{W} \in \mathbb{T}_{\bm{Z}_t}.
$$
As a result, when $m$ is sufficiently large, $\kappa$ approaches $2$.
\item \emph{Wiringer Flow.} The Riemannian metric is $\langle \cdot,\cdot\rangle_{\mathfrak{w}}$ defined in \eqref{metric2}. Thus, we have
$$
\|\bm{W}\|_{\mathfrak{w}}^2\lesssim\frac{1}{m}\|\mathcal{A}(\bm{W})\|_2^2\lesssim 4\|\bm{W}\|_{\mathfrak{w}}^2,\qquad
\forall~ \bm{W} \in \mathbb{T}_{\bm{Z}_t}.
$$
When $m$ is sufficiently large, $\kappa_{\mathfrak{w}}$ approaches 4.
\end{itemize}
The above analysis reveals mainly two key insights. First, canonical RGD demonstrates superior performance over WF in practice, which directly explains the faster convergence rates observed experimentally. Second, both methods share a fundamental limitation: their condition numbers remain strictly greater than $1$ even with arbitrarily large numbers of measurements $m$, preventing the convergence rate from approaching $0$.

This core limitation stems from the sensing operator $\frac{1}{\sqrt{m}}\mathcal{A}$ providing only a stable embedding of rank-1 matrices in $\mathbb{R}^m$ without achieving isometry under the Riemannian metrics used. Consequently, the condition number is not ideal, leading to suboptimal convergence rates. To overcome this, we propose a non-canonical Riemannian metric $\langle \cdot,\cdot\rangle_{\mathfrak{o}}$ designed to achieve near-isometry (yielding $\kappa_{\mathfrak{o}}$ approaches 1) while maintaining computational efficiency. This method not only offers theoretical guarantees of linear convergence with the fastest possible rate, approaching zero, for the associated RGD algorithm, but also demonstrates practical improvements for phase retrieval problems.

\section{Weighted Riemannian Gradient Descent}\label{sec3}
In this section, we develop the Weighted Riemannian Gradient Descent (WRGD) algorithm and its variant, which are based on a carefully designed Riemannian metric that overcomes the convergence limitations of existing algorithms identified in Section \ref{sec2} and a carefully selected initialization by
using a truncated spectral method. Section~\ref{subsection:3.1} details its derivation, while Section~\ref{subsection:3.2} establishes its theoretical guarantees, demonstrating that the truncated WRGD algorithm achieves linear convergence with a convergence rate arbitrarily close to zero, requiring sample complexity proportional to $n$ (up to logarithmic factors), thereby overcoming the lower bound barrier of the convergence rate of the canonical RGD and WF.
\subsection{(Truncated) Weighted Riemannian Descent Algorithms}\label{subsection:3.1}
Building on Section~\ref{sec2}, we develop an optimized Riemannian gradient descent method through a novel metric design. Our approach centers on constructing a Riemannian metric $\langle\cdot,\cdot\rangle_\mathfrak{o}$ on $\mathbb{H}^{n\times n}$ that induces near-isometry of $\frac{1}{\sqrt{m}}\mathcal{A}$ on tangent spaces, leading to an improved condition number $\kappa_{\mathfrak{o}}=C_{u,\mathfrak{o}}/C_{l,\mathfrak{o}}$ (see \eqref{tangent_space_condition}).

Motivated by the expectation calculations in \eqref{eq:exp} for Gaussian models, we require
\begin{equation}\label{eq:exp2}
\|\bm{W}\|_{\mathfrak{o}}^2 = \mathbb{E}\left(\frac{1}{m}\|\mathcal{A}(\bm{W})\|_2^2\right), \quad \forall~\bm{W}\in\mathbb{T}_{\bm{Z}_t}.
\end{equation}
When $\frac{1}{m}\|\mathcal{A}(\bm{W})\|_2^2$ concentrates well around its expectation, we obtain the near-isometry $\frac{1}{m}\|\mathcal{A}(\bm{W})\|_2^2 \approx \|\bm{W}\|_{\mathfrak{o}}^2$, which ultimately yields $\kappa_{\mathfrak{o}} \approx 1$.
To satisfy \eqref{eq:exp2}, we define a weighted inner product on $\mathbb{H}^{n\times n}$:
\begin{equation}\label{new metric}
\langle \bm{W}_1, \bm{W}_2 \rangle_{\mathfrak{o}} := \mathbb{E}\left(\frac{1}{m}\langle \mathcal{A}(\bm{W}_1), \mathcal{A}(\bm{W}_2) \rangle\right)=\langle \bm{W}_1,\bm{W}_2\rangle + \operatorname{tr}(\bm{W}_1)\operatorname{tr}(\bm{W}_2), \quad \forall~ \bm{W}_1, \bm{W}_2 \in\mathbb{H}^{n\times n}.
\end{equation}
Then this inner product $\langle\cdot,\cdot\rangle_\mathfrak{o}$ induces a Riemannian metric defined at each tangent space $\mathbb{T}_{\bm{Z}}$ for any $\bm{Z}\in \mathcal{M}_1$. Obviously, the norm induced by this inner product satisfies \eqref{eq:exp2}. Equipped with this Riemannian metric, $\mathcal{M}_1$ is a smooth Riemmanian manifold embedded in $\mathbb{H}^{n\times n}$ over $\mathbb{R}$.

The Riemannian gradient $\nabla_{\mathcal{M}_{1}}^{(\mathfrak{o})} F(\bm{Z})$ is given by the following Lemma.
\begin{lemma}\label{lemma:riemannian_gradient}
For $\bm{Z} = \beta\bm{u}\bm{u}^* \in \mathcal{M}_1$ with $\beta \in \mathbb{R}\setminus\{0\}$, $\bm{u}\in \mathbb{C}^{n}$, and $\|\bm{u}\|_2 = 1$, the Riemannian gradient $\nabla_{\mathcal{M}_{1}}^{(\mathfrak{o})} F(\bm{Z})$ is given by
\begin{equation}\label{M3}
\nabla_{\mathcal{M}_{1}}^{(\mathfrak{o})} F(\bm{Z})= \mathcal{T}^{(\mathfrak{o})}_{\mathbb{T}_{\bm{Z}}}(\nabla F(\bm{Z}))\in \mathbb{T}_{\bm{Z}},
\end{equation}
where $\nabla F(\bm{Z})$ is the Euclidean gradient of $F :\mathbb{H}^{n\times n} \mapsto \mathbb{R}$, and $\mathcal{T}^{(\mathfrak{o})}_{\mathbb{T}_{\bm{Z}}}$ is an operator from $\mathbb{H}^{n\times n}$ to $\mathbb{T}_{\bm{Z}}$ defined by
\begin{equation}\label{projector2}
\mathcal{T}^{(\mathfrak{o})}_{\mathbb{T}_{\bm{Z}}}(\bm{W}): = \bm{uu}^* \bm{W} + \bm{W} \bm{uu}^* - \tfrac{3}{2} \bm{uu}^* \bm{W} \bm{uu}^*, \quad \forall~ \bm{W} \in \mathbb{H}^{n \times n}.
\end{equation}
\end{lemma}
\begin{proof}
Consider a smooth curve $\bm{Z}(s) \in \mathcal{M}_1$ with $\bm{Z}(0) = \bm{Z}$. The differential satisfies:
\begin{equation*}
\begin{split}
\frac{d}{ds}& F(\bm{Z}(s))\Big|_{s=0} = \langle \dot{\bm{Z}}(0), \nabla F(\bm{Z}) \rangle\\
&= \langle \dot{\bm{Z}}(0), \bm{uu}^*\nabla F(\bm{Z}) + \nabla F(\bm{Z})\bm{uu}^* -\tfrac{3}{2}\bm{uu}^*\nabla F(\bm{Z})\bm{uu}^*\rangle +\tfrac{1}{2}\langle \dot{\bm{Z}}(0), \bm{uu}^*\nabla F(\bm{Z})\bm{uu}^* \rangle \\
&= \langle \dot{\bm{Z}}(0), \bm{uu}^*\nabla F(\bm{Z}) + \nabla F(\bm{Z})\bm{uu}^* -\tfrac{3}{2}\bm{uu}^*\nabla F(\bm{Z})\bm{uu}^*\rangle 
\\& \qquad\qquad+\langle \dot{\bm{Z}}(0), \bm{uu}^*( \bm{uu}^*\nabla F(\bm{Z}) + \nabla F(\bm{Z})\bm{uu}^* -\tfrac{3}{2}\bm{uu}^*\nabla F(\bm{Z})\bm{uu}^*)\bm{uu}^* \rangle \\
&= \langle \dot{\bm{Z}}(0), \bm{uu}^*\nabla F(\bm{Z}) + \nabla F(\bm{Z})\bm{uu}^* -\tfrac{3}{2}\bm{uu}^*\nabla F(\bm{Z})\bm{uu}^*\rangle \\&\qquad\qquad+\big(\bm{u}^*\dot{\bm{Z}}(0)\bm{u}\big)\cdot\big(\bm{u}^* (\bm{uu}^*\nabla F(\bm{Z}) + \nabla F(\bm{Z})\bm{uu}^* -\tfrac{3}{2}\bm{uu}^*\nabla F(\bm{Z})\bm{uu}^*)\bm{u}\big) \\
&=\langle \dot{\bm{Z}}(0), \bm{uu}^*\nabla F(\bm{Z}) + \nabla F(\bm{Z})\bm{uu}^* -\tfrac{3}{2}\bm{uu}^*\nabla F(\bm{Z})\bm{uu}^*\rangle  \\&\qquad\qquad+\operatorname{tr}(\dot{\bm{Z}}(0))\cdot\operatorname{tr}(\bm{uu}^*\nabla F(\bm{Z}) + \nabla F(\bm{Z})\bm{uu}^* -\tfrac{3}{2}\bm{uu}^*\nabla F(\bm{Z})\bm{uu}^*) \\
&= \langle \dot{\bm{Z}}(0), \bm{uu}^*\nabla F(\bm{Z}) + \nabla F(\bm{Z})\bm{uu}^* -\tfrac{3}{2}\bm{uu}^*\nabla F(\bm{Z})\bm{uu}^*\rangle_{\mathfrak{o}}\\
&=\langle \dot{\bm{Z}}(0), \mathcal{T}^{(\mathfrak{o})}_{\mathbb{T}_{\bm{Z}}}(\nabla F(\bm{Z}))\rangle_{\mathfrak{o}},\\
\end{split}
\end{equation*}
where the last third equality follows from Lemma~\ref{lemma:trace}, and the last second equality follows the definition of $\langle \cdot, \cdot \rangle_{\mathfrak{o}}$. 
\end{proof}

While other retraction operators are available, we select ${\mathcal{H}_1}$, $1$- truncated SVD, for simplicity. Specifically, for any $\bm{W} \in \mathbb{H}^{n\times n}$ with the SVD of $\bm{W} = \sum_i \sigma_i \bm{u}_i \bm{v}_i^{*}$ ($\sigma_1 \geq \sigma_2 \geq \cdots$), the retraction is defined as
\begin{equation}\label{R3}
\mathcal{H}_{1}(\bm{W}) := \sigma_1\bm{u}_1 \bm{v}_1^{*}.
\end{equation}
Combining the Riemannian gradient \eqref{M3} with the retraction \eqref{R3}, we obtain the Weighted Riemannian Gradient Descent (WRGD) for solving \eqref{least_square} as follows
\begin{equation}\label{WRGD_framework}
\bm{Z}_{t+1} = \mathcal{H}_{1}\left(\bm{Z}_t - \alpha_t \nabla_{\mathcal{M}_1}^{(\mathfrak{o})} F(\bm{Z}_t)\right),
\end{equation}
where $\alpha_t > 0$ is the step size. When $\bm{Z}_0$ is positive semi-definite with rank-$1$ in $\eqref{WRGD_framework}$, we can prove that $\bm{Z}_{t}$ is always positive semi-definite and rank-$1$.

Inspired by \cite{RefWorks:RefID:34-cai2018solving}, we introduce the truncated variant of WRGD, where the sensing operator is adaptively computed from $\mathcal{A}$ based on the measurement vector $\bm{y}$ and the current estimate $\bm{Z}_t$. Given any $\bm{Z}:=\bm{zz}^* \in \mathbb{H}^{n \times n}$, let $\mathcal{A}_{\bm{Z}}$ be the linear operator associated with $\bm{Z}$ defined by
\[
\mathcal{A}_{\bm{Z}}(\bm{W}) = \left\{ \left\langle \bm{W}, \bm{a}_k \bm{a}_k^{*} \right\rangle\cdot\mathbb{1}_{\mathcal{E}_0^k(\bm{y}) \cap \mathcal{E}_1^k(\bm{z}) \cap \mathcal{E}_2^k(\bm{z})}\right\}_{k=1}^m,
\]
where $\mathbb{1}$ is the indicator function and $\mathcal{E}_0^k(\bm{y}) \cap \mathcal{E}_1^k(\bm{z}) \cap \mathcal{E}_2^k(\bm{z})$ is a collection of events determined by the following
\begin{equation}\label{truncation:guassian}
  \begin{aligned}
& \mathcal{E}_0^k(\bm{y})=\left\{\sqrt{y_k} \leq \tau_0 \sqrt{\frac{\|\bm{y}\|_1}{m}}\right\}, \\
& \mathcal{E}_1^k(\bm{z})=\left\{\left|\bm{a}_k^* \bm{z}\right| \leq \tau_1\|\bm{z}\|_2\right\}, \\
& \mathcal{E}_2^k(\bm{z})=\left\{\left|y_k-\left|\bm{a}_k^* \bm{z}\right|^2\right| \leq \frac{\tau_2}{m}\left\|\bm{y}-\mathcal{A}\left(\bm{z} \bm{z}^*\right)\right\|_1 \frac{\left|\bm{a}_k^* \bm{z}\right|+\sqrt{y_k}}{\|\bm{z}\|_2}\right\},
\end{aligned}  
\end{equation}
with prescribed truncation parameters $\tau_0$, $\tau_1$, and $\tau_2$.
In other words, if $\bm{a}_k$ satisfies the above truncation rules, the $k$-th entry of $\mathcal{A}_{\bm{Z}}(\bm{W})$ is $\left\langle \bm{W}, \bm{a}_k \bm{a}_k^{*} \right\rangle$; otherwise, it is set to $0$. Note that the adjoint of $\mathcal{A}_{\bm{Z}}$ is given by
\[
\mathcal{A}_{\bm{Z}}^{*}(\bm{b}) = \frac{1}{m}\sum_{k=1}^m b_k \bm{a}_k \bm{a}_k^{*} \cdot \mathbb{1}_{\mathcal{E}_0^k(\bm{y}) \cap \mathcal{E}_1^k(\bm{z}) \cap \mathcal{E}_2^k(\bm{z})}, \quad \forall~ \bm{b} \in \mathbb{R}^m.
\]
For simplicity, we use $\mathcal{A}_t$, $\mathcal{A}_t^*$, and $\mathcal{T}^{(\mathfrak{o})}_{\mathbb{T}_t}$ to denote $\mathcal{A}_{\bm{Z}_t}$, $\mathcal{A}^*_{\bm{Z}_t}$, and $\mathcal{T}^{(\mathfrak{o})}_{\mathbb{T}_{\bm{Z}_t}}$, respectively, in the following algorithms. Then we can give the truncated versions of weighted Riemannian gradient descent-based algorithms.

The complete truncated Weighted Riemannian Gradient Descent (TWRGD) algorithm is summarized as Algorithm \ref{alg:1}. Although the algorithm involves iterations over matrices, it is actually about iterative updates of two vectors. Therefore, the computational complexity of this type of algorithms is the same as that of iterative algorithms in vector spaces, such as the WF algorithm. See the Appendix \ref{sec8.2} for details.
\begin{algorithm}[H]
    \renewcommand{\algorithmicrequire}{\textbf{Initialization:}}
    \caption{(TWRGD) Truncated Weighted Riemannian Gradient Descent }
    \label{alg:1}
    \begin{algorithmic}[1]
         \Require $\bm{Z}_0$.
          
          \ForAll{$t=0, 1, \ldots, $}
          \State $\bm{G}_t=\frac{1}{m}\mathcal{A}_{t}^{*}\left(\bm{y}-\mathcal{A}_{t}\left(\bm{Z}_t\right)\right)$
          
          
          \State $\bm{W}_{t}=\bm{Z}_{t}+\alpha_t \mathcal{T}_{\mathbb{T}_t}^{(\mathfrak{o})}(\bm{G}_t )$
          
          \State $\bm{Z}_{t+1}=\mathcal{H}_{1}\left(\bm{W}_{t}\right)$
          \EndFor

\end{algorithmic} 
\end{algorithm}

Furthermore, the step size $\alpha_t > 0$ in TWRGD can be a constant or adaptively calculated. Exact linear searches along $\mathcal{T}^{(\mathfrak{o})}_{\mathbb{T}_{t}}(\bm{G}_t)$ and $\mathcal{T}^{(\mathfrak{o})}_{\mathbb{T}_{t}}(\bm{S}_t)$ yield
 closed-form step-sizes given by 
\begin{equation}\label{step-size}
\begin{split}
&\frac{\alpha_t}{m} = \frac{\|\mathcal{T}^{(\mathfrak{o})}_{\mathbb{T}_{t}}(\bm{G}_t)\|_{\mathfrak{o}}^2}{\|\mathcal{A}_t\mathcal{T}^{(\mathfrak{o})}_{\mathbb{T}_{t}}(\bm{G}_t)\|_2^2},\quad 
\end{split}
\end{equation}
in Algorithm \ref{alg:1}.


Given the non-convex nature of the problem, selecting an initialization close to the true solution is crucial for the success of iterative updates. To this end, we adopt the truncated spectral initialization method proposed in \cite{RefWorks:RefID:12-chen2017solving} for random complex Gaussian models. This approach involves computing the leading eigenvector of the truncated version of $\frac{1}{m}\mathcal{A}^*({\bm{y}}$), formulated as follows.
\begin{algorithm}[H]
    \renewcommand{\algorithmicrequire}{\textbf{Initialization:}}
    \caption{Initialization by Truncated Spectral Method }
    \label{alg:3}
    \begin{algorithmic}[1] 
        \State Let $\bm{z}_0:=\sqrt{\frac{n\|\bm{y}\|_1}{\sum_{i=1}^{m}\|\bm{a}_i\|_2^2}}\hat{\bm{z}}$, where $\hat{\bm{z}}$ be the leading eigenvector of the matrix:
         \begin{equation*}\nonumber
         \bm{Y} :=\frac{1}{m} \sum_{k=1}^m  y_{k}\bm{a}_k \bm{a}_k^*\cdot\mathbb{1}_{U^k_{0}(\bm{y}) },\quad U^k_{0}(\bm{y}):=\bigg\{y_k\leq\frac{3}{m}\|\bm{y}\|_1\bigg\},
         \end{equation*}
          \State Set $\bm{Z}_0= \bm{z}_0 \bm{z}_0^*$.
        \end{algorithmic} 
\end{algorithm}
As established in \cite{RefWorks:RefID:12-chen2017solving}, Algorithm \ref{alg:3} yields an initialization sufficiently close to the true solution, as formally detailed in Theorem \ref{thm:1}.

\subsection{Recovery Guarantee}\label{subsection:3.2}
In this section, we present the main theoretical results for the proposed TWRGD algorithm. Specifically, we prove that both algorithms converge linearly to the global minimizer with an arbitrarily small contraction factor, assuming random complex Gaussian measurements in the noiseless setting.

We begin by establishing the theoretical guarantee for the initialization generated by Algorithm \ref{alg:3}. The following theorem states that the algorithm constructs an initialization sufficiently close to the true solution, provided that $m\geq cn$ with some absolute constant $c>0$.

\begin{thm}[Initialization]\label{thm:1}
    Let $\{\bm{a}_k\}_{k=1}^{m}$ follow the random Gaussian measurement model in (\ref{eq:Gaussian}). For any constant $\epsilon_0 \in(0,1)$, there exist positive constants $c_1, c_2$ such that if $m \geq c_1 n$, the initialization $\bm{Z}_0$ generated by Algorithm \ref{alg:3} with $\bm{y}=\mathcal{A}(\bm{X})$ satisfies 
$$
\left\|\bm{Z}_0-\bm{X}\right\|_F \leq \epsilon_0\|\bm{X}\|_F
$$
with probability at least $1-e^{-c_2m}$.
\end{thm}

\begin{proof}
According to \cite[Proposition C.3]{RefWorks:RefID:12-chen2017solving}, the vector $\bm{z}_0$ produced by Algorithm \ref{alg:3} satisfies the following property: for any $\epsilon \in(0,1)$, there exist constants $c_1, c_2 > 0$ such that
$$
\|\bm{z}_0-\bm{x}\|_2 \leq \epsilon\|\bm{x}\|_2
$$ 
holds with probability at least $1-e^{-c_2m}$, provided that $m \geq c_1 n$. Conditioning on this event and recalling that $\bm{Z}_0=\bm{z}_0\bm{z}_0^*$ and $\bm{X}=\bm{x}\bm{x}^*$, we can bound the matrix error. By choosing $\epsilon$ such that $(2+\epsilon)\epsilon= \epsilon_0$, we obtain:
\begin{align*}
\|\bm{Z}_0-\bm{X}\|_F\leq(\|\bm{z}_0\|_2+\|\bm{x}\|_2)\|\bm{z}_0-\bm{x}\|_2\leq (2+\epsilon)\epsilon\|\bm{x}\|_2^2=\epsilon_0\|\bm{X}\|_F.
\end{align*}
This completes the proof.
\end{proof}

Building on this initialization guarantee, we establish the local linear convergence of the TWRGD algorithm.
\begin{thm}[Local Convergence for TWRGD]\label{thm:3}
Let $\{\bm{a}_k\}_{k=1}^{m}$ follow the random Gaussian measurement model in (\ref{eq:Gaussian}). Let $\tau_0,\tau_1,\tau_2>0$ be sufficiently large truncation parameters. There exist positive constants $c_3,c_4$, $\epsilon_0\in(0,\frac{1}{22})$, and step sizes $\alpha_t>0$ for any $t>0$, such that: provided $m\geq c_3n$ and $\bm{Z}_0$ obeys
$$
\left\|\bm{Z}_0-\bm{X}\right\|_F \leq \epsilon_0\|\bm{X}\|_F,
$$
then with probability exceeding $1-e^{-c_4n}$, the sequence $\{\bm{Z}_t\}_{t\in\mathbb{N}}$ generated by Algorithm \ref{alg:1} using  initialization $\bm{Z}_0$, input data $\bm{y}=\mathcal{A}(\bm{X})$, and parameters $\tau_0, \tau_1, \tau_2, \alpha_t$ satisfies
\begin{equation}\nonumber
\left\|\bm{Z}_{t}-\bm{X}\right\|_F \leq \nu_1^{t}\epsilon_0\left\|\bm{X}\right\|_F,
\end{equation}
for some contraction factor $\nu_1\in(0,1)$. Moreover, $\nu_1$ decreases to $0$ as the truncation parameters $\tau_0,\tau_1,\tau_2$ increase.
\end{thm}

The proofs of Theorems \ref{thm:3} are deferred to Section \ref{sec5}. These theorems demonstrate that when the number of measurements $m$ is sufficiently large, the contraction factors of our proposed TWRGD  becomes arbitrarily small. In contrast, the contraction factors of existing algorithms—such as canonical TRGD, and TWF, possess a non-zero lower bound even as the number of measurements increases. This shows that the proposed metric on the Riemannian manifold helps improve algorithmic efficiency. It further indicates that the general framework of Riemannian gradient descent-based algorithms, with a suitable choice of metric, can achieve optimal computational performance.

Combining the preceding theorems, we establish the following exact recovery guarantees for our proposed algorithms.
\begin{thm}[Exact Recovery Guarantees]
Let $\{\bm{a}_k\}_{k=1}^{m}$ follow the random Gaussian measurement model in (\ref{eq:Gaussian}). Then there exist algorithmic parameters $\tau_0,\tau_1,\tau_2,\alpha_t$, and positive constants $c_5,c_6,\nu\in(0,1)$ such that $\{\bm{Z}_t\}_{t\in\mathbb{N}}$ generated by Algorithm \ref{alg:1}, using input data $\bm{y}=\mathcal{A}(\bm{X})$ and the chosen algorithmic parameters satisfy: if $m\geq c_5 n$, then, with probability at least $1-e^{-c_6 m}$,
\begin{equation}\nonumber
\|\bm{Z}_t-\bm{X}\|_{F}\leq \nu^{t}\cdot\epsilon_0\|\bm{X}\|_{F}.
\end{equation}
Moreover, the contraction factor $\nu$ decreases to $0$ as the truncation parameters $\tau_0,\tau_1,\tau_2$ increase.
\end{thm}
\begin{proof}
    The theorem follows directly from combining Theorems \ref{thm:1}, and \ref{thm:3}.
\end{proof}

\section{Supporting Lemmas}\label{sec4}
In this section, we present the supporting lemmas required for the proofs of the main results, Theorems \ref{thm:3}. First, in Section \ref{sec4.1}, we characterize the relationship between the inner product induced by the newly defined metric and the standard matrix inner product. Next, in Section \ref{sec4.2}, we establish several concentration inequalities arising from the random complex Gaussian model. 

For clarity and convenience, we define the following constants. Unless otherwise specified, $\xi$ refers to a complex standard normal random variable. The following constants will be used throughout the rest of the paper:
\begin{align*}
&\hat{\beta_1}:= \mathbb{E}[|\xi|^4\cdot\mathbb{1}_{|\xi|\leq\tau_1}]-\mathbb{E}[|\xi|^2\cdot\mathbb{1}_{|\xi|\leq\tau_1}],\quad
\hat{\beta_2}:= \mathbb{E}\left[|\xi|^2\cdot\mathbb{1}_{{|\xi| \leq \tau_1}}\right], \\
&\rho_1:= \left(10\tau_1^3e^{-0.49\tau_1^2}+\dfrac{10\tau_1^2}{\epsilon}e^{-0.49\epsilon^{-2}}+6\tau_1^4\epsilon\right), \\
&\rho_2:= 6 \tau_1^2 \tau_2 e^{-0.64 \tau_2^2} + 4 \tau_1^2 \tau_0 e^{-0.39 \tau_0^2} + 2\epsilon, \\
&\rho_4:= \rho_3 + \left(6 \tau_2 \tau_{h, z}^2 e^{-0.64 \tau_2^2}+\tau_{h, z}^2\epsilon\right), \\ 
&\epsilon_0:=\min \left\{\sqrt{\frac{\rho_3}{3(\tau_1^4+5 \tau_1^3+8 \tau_1^2+2 \tau_2^2)}}, \frac{\rho_3}{15 \tau_1 \tau_{h, z}}, \frac{1}{11}\right\},\\
&\delta: = 2 (\rho_1 +\rho_2),\quad
\tau_{h, z}:=\tau_1+\left(0.3 \tau_2\left(\tau_1+1.2 \tau_0\right)+\tau_1^2\right)^{1 / 2},
\end{align*}
where $\rho_3>0$ and $\epsilon>0$ are sufficiently small constant.
 \subsection{Relationships between Inner Products and Projections}\label{sec4.1}
We first establish relationships between the inner product induced by the newly defined metric $\langle \cdot,  \cdot\rangle_{\mathfrak{o}}$ in \eqref{new metric} and the standard matrix inner product $\langle\cdot , \cdot\rangle$. Recall that for any $\bm{Z}\in\mathcal{M}_1$, we denote the tangent space at $\bm{Z}$ by $\mathbb{T}_{\bm{Z}}$ \eqref{tangent_space11}, the orthogonal projection onto $\mathbb{T}_{\bm{Z}}$ (with respect to the standard inner product) by $\mathcal{P}_{\mathbb{T}_{\bm{Z}}}$ \eqref{ortho projector}, and a non-orthogonal projection onto $\mathbb{T}_{\bm{Z}}$ by $\mathcal{T}^{(\mathfrak{o})}_{\mathbb{T}_{\bm{Z}}}$ \eqref{projector2}. We thus obtain the following lemmas.

 \begin{lemma}\label{lemma:metric1}
 For $\bm{Z}\in\mathcal{M}_1$, any $\bm{A}\in\mathbb{H}^{n\times n}$, and $\bm{B} \in {\mathbb{T}_{\bm{Z}}}$, we have
 \[
 \langle\mathcal{T}^{(\mathfrak{o})}_{\mathbb{T}_{\bm{Z}}}\bm{A}, \bm{B} \rangle_{\mathfrak{o}} = \langle \bm{A}, \bm{B} \rangle.
  \]  
 \end{lemma}
 
 \begin{proof}
     Let $\bm{Z}=\beta\bm{u}\bm{u}^*$ where $\|\bm{u}\|_2=1$. We observe that $\mathcal{T}^{(\mathfrak{o})}_{\mathbb{T}_{\bm{Z}}}\bm{A}=\mathcal{P}_{\mathbb{T}_{\bm{Z}}}\bm{A}-\frac{1}{2}\bm{uu^*Auu^*}$. Then, we derive
     \begin{equation}\nonumber
         \begin{split}
         \langle\mathcal{T}^{(\mathfrak{o})}_{\mathbb{T}_{\bm{Z}}}\bm{A}, \bm{B} \rangle_{\mathfrak{o}}
&= \langle \mathcal{T}^{(\mathfrak{o})}_{\mathbb{T}_{\bm{Z}}}\bm{A}, \bm{B}\rangle + \operatorname{tr}(\mathcal{T}^{(\mathfrak{o})}_{\mathbb{T}_{\bm{Z}}}\bm{A})\operatorname{tr}(\bm{B})\\
         & = \langle \mathcal{P}_{\mathbb{T}_{\bm{Z}}}\bm{A}, \bm{B}\rangle - \frac{1}{2}\langle\bm{uu}^*\bm{A}\bm{uu}^*,\bm{B}\rangle + \operatorname{tr}(\mathcal{P}_{\mathbb{T}_{\bm{Z}}}\bm{A})\operatorname{tr}(\bm{B}) -\frac{1}{2}\operatorname{tr}(\bm{uu}^*\bm{A}\bm{uu}^*)\operatorname{tr}(\bm{B})\\
         &= \langle \bm{A}, \bm{B}\rangle -\frac{1}{2}\bm{u}^*\bm{A}\bm{u}\bm{u}^*\bm{B}\bm{u}+\bm{u}^*\bm{A}\bm{u}\bm{u}^*\bm{B}\bm{u}-\frac{1}{2}\bm{u}^*\bm{A}\bm{u}\bm{u}^*\bm{B}\bm{u}
         = \langle \bm{A}, \bm{B}\rangle,
         \end{split}
     \end{equation} 
where the second last equality follows from Lemma \ref{lemma:trace}. 
 \end{proof}

 \begin{lemma}\label{lemma:metric2}
 For any $\bm{A} \in \mathbb{H}^{n \times n}$, we have 
 $$
 \|\mathcal{T}^{(\mathfrak{o})}_{\mathbb{T}_{\bm{Z}}}(\bm{A})\|_{\mathfrak{o}} \leq \|\mathcal{P}_{\mathbb{T}_{\bm{Z}}}(\bm{A})\|_F \quad \text{and} \quad \|\mathcal{P}_{\mathbb{T}_{\bm{Z}}}(\bm{A})\|_F\leq\|\mathcal{P}_{\mathbb{T}_{\bm{Z}}}(\bm{A})\|_{\mathfrak{o}}\leq\sqrt{2}\|\mathcal{P}_{\mathbb{T}_{\bm{Z}}}(\bm{A})\|_F .
 $$
 \end{lemma}
 
 \begin{proof}
     By the definition of the metric $\mathfrak{o}$, for any $\bm{A} \in \mathbb{H}^{n \times n}$, we have
\begin{equation}\nonumber
    \begin{split}
       \|\mathcal{T}^{(\mathfrak{o})}_{\mathbb{T}_{\bm{Z}}}(\bm{A})\|_{\mathfrak{o}}^2 &= \langle \mathcal{T}^{(\mathfrak{o})}_{\mathbb{T}_{\bm{Z}}}(\bm{A}),\mathcal{T}^{(\mathfrak{o})}_{\mathbb{T}_{\bm{Z}}}(\bm{A})\rangle_ {\mathfrak{o}}
= \langle \mathcal{P}_{\mathbb{T}_{\bm{Z}}}(\bm{A}),\mathcal{T}^{(\mathfrak{o})}_{\mathbb{T}_{\bm{Z}}}(\bm{A})\rangle \\
&=\langle \mathcal{P}_{\mathbb{T}_{\bm{Z}}}(\bm{A}), \mathcal{P}_{\mathbb{T}_{\bm{Z}}}(\bm{A})\rangle - \frac{1}{2}\langle\mathcal{P}_{\mathbb{T}_{\bm{Z}}}(\bm{A}), \bm{uu^*Auu^*} \rangle\\
& = \langle \mathcal{P}_{\mathbb{T}_{\bm{Z}}}(\bm{A}), \mathcal{P}_{\mathbb{T}_{\bm{Z}}}(\bm{A})\rangle -\frac{1}{2}(\bm{u^*Au})^2 \leq \|\mathcal{P}_{\mathbb{T}_{\bm{Z}}}(\bm{A})\|_F^2.
   \end{split}
   \end{equation}
   Here, the second equality holds because $\mathcal{T}^{(\mathfrak{o})}_{\mathbb{T}_{\bm{Z}}}(\bm{A})\in {\mathbb{T}_{\bm{Z}}}$ and we apply Lemma \ref{lemma:metric1}. The second inequality follows from a similar derivation.
 \end{proof}

\subsection{Concentration Inequalities}\label{sec4.2}
In order to prove our main results, a crucial step involves estimating the uniform concentration of the random sum $\frac{1}{m}\sum_{k=1}^{m}\left(\operatorname{Re}\left(\bm{w}^* \bm{a}_k \bm{a}_k^* \bm{z}\right)\right)^2 \cdot \mathbb{1}_{\left\{\left|\bm{a}_k^{*} \bm{z}\right| \leq \tau_1\|\bm{z}\|_2\right\}}$ for all $\bm{z},\bm{w} \in \mathbb{C}^n$, as described in the following lemma. 
\begin{lemma}\label{lemma:uniform_concentration}
Fix $\tau_1>0$ and let $\epsilon \in(0,1)$ be a sufficiently small constant. There exist universal constants $C_1,C_2>0$ such that: if $m \geq C_1\epsilon^{-2} \log \epsilon^{-1} \cdot n$, then with probability at least $1-e^{-C_2m \epsilon^2}$, we have 
\begin{equation}\label{real part}\begin{split}
&\left|\frac{1}{m}\sum_{k=1}^{m}\left(\operatorname{Re}\left(\bm{w}^* \bm{a}_k \bm{a}_k^* \bm{z}\right)\right)^2 \cdot \mathbb{1}_{\left\{\left|\bm{a}_k^{*} \bm{z}\right| \leq \tau_1\|\bm{z}\|_2\right\}}- \mathbb{E}\left[\frac{1}{m}\sum_{k=1}^{m}\left(\operatorname{Re}\left(\bm{w}^* \bm{a}_k \bm{a}_k^* \bm{z}\right)\right)^2 \cdot \mathbb{1}_{\left\{\left|\bm{a}_k^{*} \bm{z}\right| \leq \tau_1\|\bm{z}\|_2\right\}}\right]\right| \\ &\leq \rho_1 \|\bm{z}\|_2^2\|\bm{w}\|_2^2,\quad \forall~\bm{z},\bm{w} \in \mathbb{C}^n.
\end{split}\end{equation} 
\end{lemma}
The proof of Lemma \ref{lemma:uniform_concentration} is deferred to the end of this section. To prove it, we first reformulate the sum as follows:
\begin{equation}
\label{reformulation}
\begin{split}
&\frac{1}{m}\sum_{k=1}^{m}\left(\operatorname{Re}\left(\bm{w}^* \bm{a}_k \bm{a}_k^* \bm{z}\right)\right)^2 \cdot \mathbb{1}_{\left\{\left|\bm{a}_k^{*} \bm{z}\right| \leq \tau_1\|\bm{z}\|_2\right\}} \\
&=\frac{1}{m} \sum_{k=1}^m\left[\begin{array}{l}\bm{w} \\ \overline{\bm{w}}\end{array}\right]^*\left[\begin{array}{ll}\left|\bm{a}_k^* \bm{x}\right|^2 \bm{a}_k \bm{a}_k^* & \left(\bm{a}_k^* \bm{x}\right)^2 \bm{a}_k \bm{a}_k^{\top} \\ \left(\overline{\bm{a}_k^* \bm{x}}\right)^2 \overline{\bm{a}}_k \bm{a}_k^* & \left|\bm{a}_k^* \bm{x}\right|^2 \overline{\bm{a}}_k \bm{a}_k^{\top}\end{array}\right]\left[\begin{array}{l}\bm{w} \\ \overline{\bm{w}}\end{array}\right]\cdot \mathbb{1}_{\left\{\left|\bm{a}_k^{*} \bm{z}\right| \leq \tau_1\|\bm{z}\|_2\right\}}.
\end{split}\end{equation}
This reformulation implies that it is sufficient to estimate the concentration of $ \frac{1}{m}\sum_{k=1}^{m}\left|\bm{a}_k^* \bm{x}\right|^2 \bm{a}_k \bm{a}_k^*\cdot \mathbb{1}_{\left\{\left|\bm{a}_k^{*} \bm{z}\right| \leq \tau_1\|\bm{z}\|_2\right\}} $ and $\frac{1}{m}\sum_{k=1}^{m}\left(\bm{a}_k^* \bm{x}\right)^2 \bm{a}_k \bm{a}_k^{\top}\cdot \mathbb{1}_{\left\{\left|\bm{a}_k^{*} \bm{z}\right| \leq \tau_1\|\bm{z}\|_2\right\}} $, respectively. 

We begin by calculating the expectations.
\begin{lemma}\label{gaussian:expectations}
Let $\bm{z},\bm{w}\in\mathbb{C}^{n}$. Assume that the measurement vectors $\{\bm{a}_k\}_{k=1}^{m}$ are drawn from the distribution specified in \eqref{eq:Gaussian}. Then the following expectations hold:
\[
\begin{split}
    &\mathbb{E}\left[\left|\bm{a}_k^{*} \bm{z}\right|^2 \bm{a}_k \bm{a}_k^{*} \mathbb{1}_{\left\{\left|\bm{a}_k^{*} \bm{z}\right| \leq \tau_1\|\bm{z}\|_2\right\}}\right] = \hat{\beta_1} \bm{z} \bm{z}^{*}+\hat{\beta_2}\|\bm{z}\|_2^2 \bm{I}_n,\\
    & \mathbb{E}\left[\left(\bm{a}_k^*\bm{z}\right)^2 \bm{a}_k \bm{a}_k^{\top}\mathbb{1}_{\left\{\left|\bm{a}_k^*\bm{z}\right| \leq \tau_1\|\bm{z}\|_2\right\}}\right]=\left(\hat{\beta_1}+\hat{\beta_2}\right) \bm{z} \bm{z}^{\top}, \\
   \end{split}\] 
and  
\[
\mathbb{E}\left[\left(\operatorname{Re}\left(\bm{w}^* \bm{a}_k \bm{a}_k^* \bm{z}\right)\right)^2 \cdot \mathbb{1}_{\left\{\left|\bm{a}_k^{*} \bm{z}\right| \leq \tau_1\|\bm{z}\|_2\right\}}\right] = \frac{1}{2}\|\bm{w}\|_2^2 \|\bm{z}\|_2^2 + (\hat{\beta_1}+\frac{1}{2}\hat{\beta_2})(\operatorname{Re}(\bm{z^*w}))^2 - \frac{1}{2}\hat{\beta_2}(\operatorname{Im}(\bm{z^*w}))^2.
\]   
\end{lemma}
\begin{proof} 
First, for each $k \in \{1, \ldots, m\}$, we can expand the real part term: 
\[
\left(\operatorname{Re}\left(\bm{w}^* \bm{a}_k \bm{a}_k^* \bm{z}\right)\right)^2 =\frac{1}{2}|\bm{a}_k^*\bm{w}|^2|\bm{a}_k^*\bm{z}|^2+\frac{1}{4}\left(\overline{(\bm{a}_k^*\bm{z})^2}(\bm{a}_k^*\bm{w})^2+\overline{(\bm{a}_k^*\bm{w})^2}(\bm{a}_k^*\bm{z})^2\right).
\]
Without loss of generality, we assume $\|\bm{z}\|_2=1$ and $\|\bm{w}\|_2=1$. Due to the unitary invariance of the distribution of $\bm{a}_k$, we can set $\bm{z}=\bm{e}_1$ and $\bm{w}=s_1e^{\imath\phi_1}\bm{e}_1+s_2e^{\imath\phi_2}\bm{e}_2$, where $s_1, s_2$ are positive real numbers such that $s_1^2+s_2^2=1$. Let $\bm{a}_k=[g_1, \ldots, g_n]^{\top}$, where $g_i \sim_{i.i.d} \mathcal{N}(0,1/2)+\imath\mathcal{N}(0,1/2)$ for $i=1, \ldots, n$. Then, we have
\[
\bm{a}_k^*\bm{z}=\overline{g_1},\quad\textrm{and}\quad
\bm{a}_k^*\bm{w}=s_1e^{\imath\phi_1}\overline{g_1}+s_2e^{\imath\phi_2}\overline{g_2},   
\]
where $\mathbb{E}[g_i]=0,~\mathbb{E}[|g_i|^2]=1,~\mathbb{E}[g_i^2]=0$ for $i=1,\ldots, n$.

For the term $\mathbb{E}\Big[\frac{1}{2}|\bm{a}_k^*\bm{w}|^2|\bm{a}_k^*\bm{z}|^2\cdot \mathbb{1}_{\left\{\left|\bm{a}_k^{*} \bm{z}\right| \leq \tau_1\|\bm{z}\|_2\right\}}\Big]$, we observe that 
\[\begin{split}
|\bm{a}_k^*\bm{w}|^2&=(s_1e^{\imath\phi_1}\overline{g_1}+s_2e^{\imath\phi_2}\overline{g_2})(s_1e^{-\imath\phi_1}g_1+s_2e^{-\imath\phi_2}g_2)\\
&=s_1^2|g_1|^2+s_2^2|g_2|^2+s_1s_2e^{\imath(\phi_2-\phi_1)}g_1\overline{g_2}+s_1s_2e^{\imath(\phi_1-\phi_2)}\overline{g_1}g_2.
\end{split}\]
Thus, we obtain 
\begin{equation}\label{concern1}
\begin{split}
&\mathbb{E}\Big[\frac{1}{2}|\bm{a}_k^*\bm{w}|^2|\bm{a}_k^*\bm{z}|^2\cdot \mathbb{1}_{\left\{\left|\bm{a}_k^{*} \bm{z}\right| \leq \tau_1\|\bm{z}\|_2\right\}}\Big]\\
&=\frac{1}{2}\mathbb{E}_{g_1}\mathbb{E}_{g_2}\Big[|g_1|^2\left(s_1^2|g_1|^2+s_2^2|g_2|^2+s_1s_2e^{\imath(\phi_2-\phi_1)}g_1\overline{g_2}+s_1s_2e^{\imath(\phi_1-\phi_2)}\overline{g_1}g_2\right)\cdot\mathbb{1}_{|g_1|\leq\tau_1}\Big]\\
&=\frac{1}{2}\mathbb{E}_{g_1}\Big[(s_1^2|g_1|^4+s_2^2|g_1|^2)\cdot\mathbb{1}_{|g_1|\leq\tau_1}\Big]\\
&=\frac{1}{2}s_1^2(\hat{\beta_1}+\hat{\beta_2})+\frac{1}{2}s_2^2\hat{\beta_2}.
\end{split}
\end{equation}

Next, for the term $\mathbb{E}\Big[\frac{1}{4}\left(\overline{(\bm{a}_k^*\bm{z})^2}(\bm{a}_k^*\bm{w})^2+\overline{(\bm{a}_k^*\bm{w})^2}(\bm{a}_k^*\bm{z})^2\right)\cdot \mathbb{1}_{\left\{\left|\bm{a}_k^{*} \bm{z}\right| \leq \tau_1\|\bm{z}\|_2\right\}}\Big]$, we note that
\[\begin{split}
\overline{(\bm{a}_k^*\bm{z})^2}(\bm{a}_k^*\bm{w})^2&=g_1^2(s_1e^{\imath\phi_1}\overline{g_1}+s_2e^{\imath\phi_2}\overline{g_2})^2\\
&=g_1^2(s_1^2e^{2\imath\phi_1}\overline{g_1}^2+s_2^2e^{2\imath\phi_2}\overline{g_2}^2+2s_1s_2e^{\imath(\phi_1+\phi_2)}\overline{g_1}\overline{g_2})\\
&=s_1^2e^{2\imath\phi_1}|g_1|^4+s_2^2e^{2\imath\phi_2}g_1^2\overline{g_2}^2+2s_1s_2e^{\imath(\phi_1+\phi_2)}|g_1|^2g_1\overline{g_2}.
\end{split}\]
Therefore, we get 
\begin{equation}\label{concern2}
\begin{split}
&\mathbb{E}\Big[\frac{1}{4}\left(\overline{(\bm{a}_k^*\bm{z})^2}(\bm{a}_k^*\bm{w})^2+\overline{(\bm{a}_k^*\bm{w})^2}(\bm{a}_k^*\bm{z})^2\right)\cdot \mathbb{1}_{\left\{\left|\bm{a}_k^{*} \bm{z}\right| \leq \tau_1\|\bm{z}\|_2\right\}}\Big]\\
&=\frac{1}{4}\mathbb{E}_{g_1}\mathbb{E}_{g_2}\Big[\left(s_1^2\cdot 2\cos(2\phi_1)|g_1|^4+s_2^2e^{2\imath\phi_2}g_1^2\overline{g_2}^2+s_2^2e^{-2\imath\phi_2}\overline{g_1}^2g_2^2\right)\mathbb{1}_{|g_1|\leq\tau_1}\Big]\\
&+\frac{1}{4}\mathbb{E}_{g_1}\mathbb{E}_{g_2}\Big[\left(2s_1s_2e^{\imath(\phi_1+\phi_2)}|g_1|^2g_1\overline{g_2}+2s_1s_2e^{-\imath(\phi_1+\phi_2)}|g_1|^2\overline{g_1}g_2\right)\mathbb{1}_{|g_1|\leq\tau_1}\Big]\\
&=\frac{1}{4}\mathbb{E}_{g_1}\bigg[s_1^2\cdot 2\cos(2\phi_1)|g_1|^4\cdot\mathbb{1}_{|g_1|\leq\tau_1}\bigg]\\
&=\frac{1}{2}s_1^2\cos(2\phi_1)(\hat{\beta_1}+\hat{\beta_2}).
\end{split}
\end{equation}
Combining equalities \eqref{concern1} and \eqref{concern2} implies that 
\begin{equation}
\begin{split}
&\mathbb{E}\Big[\left(\operatorname{Re}\left(\bm{w}^* \bm{a}_k \bm{a}_k^* \bm{z}\right)\right)^2 \cdot \mathbb{1}_{\left\{\left|\bm{a}_k^{*} \bm{z}\right| \leq \tau_1\|\bm{z}\|_2\right\}}\Big]\\
&=\frac{1}{2}s_1^2(\hat{\beta_1}+\hat{\beta_2})+\frac{1}{2}s_2^2\hat{\beta_2}+\frac{1}{2}s_1^2\cos(2\phi_1)(\hat{\beta_1}+\hat{\beta_2})\\
&=s_1^2(\hat{\beta_1}+\hat{\beta_2})\cos^2\phi_1+\frac{1}{2}s_2^2\hat{\beta_2}\\
&=s_1^2(\hat{\beta_1}+\hat{\beta_2})\cos^2\phi_1+\frac{1}{2}\hat{\beta_2}(1-s_1^2)\\
&=\frac{1}{2}\hat{\beta_2}+s_1^2(\hat{\beta_1}+\hat{\beta_2})\cos^2\phi_1-\frac{1}{2}s_1^2\hat{\beta_2}(\cos\phi_1^2+\sin\phi_1^2)\\
&=\frac{1}{2}\hat{\beta_2}+(\hat{\beta_1}+\frac{1}{2}\hat{\beta_2})(s_1\cos\phi_1)^2-\frac{1}{2}\hat{\beta_2}(s_1\sin\phi_1)^2\\
&=\frac{1}{2}\hat{\beta_2}+(\hat{\beta_1}+\frac{1}{2}\hat{\beta_2})(\operatorname{Re}(\bm{z}^*\bm{w}))^2 - \frac{1}{2}\hat{\beta_2}(\operatorname{Im}(\bm{z}^*\bm{w}))^2,
\end{split}
\end{equation}
where the last equality follows from the facts that 
$\operatorname{Re}(\bm{z}^*\bm{w})=s_1\cos\phi_1$ and $\operatorname{Im}(\bm{z}^*\bm{w})=s_1\sin\phi_1.$

The proofs for the expectations $\mathbb{E}\left[\left|\bm{a}_k^{*} \bm{z}\right|^2 \bm{a}_k \bm{a}_k^{*} \mathbb{1}_{\left\{\left|\bm{a}_k^{*} \bm{z}\right| \leq \tau_1\|\bm{z}\|_2\right\}}\right] $ and $\mathbb{E}\left[\left(\bm{a}_k^*\bm{z}\right)^2 \bm{a}_k \bm{a}_k^{\top}\mathbb{1}_{\left\{\left|\bm{a}_k^*\bm{z}\right| \leq \tau_1\|\bm{z}\|_2\right\}}\right]$ follow a similar derivation and are thus omitted for brevity.
\end{proof}
Next, in the following lemma, we estimate the concentration of $ \frac{1}{m}\sum_{k=1}^{m}\left|\bm{a}_k^* \bm{z}\right|^2 \bm{a}_k \bm{a}_k^*\cdot \mathbb{1}_{\left\{\left|\bm{a}_k^{*} \bm{z}\right| \leq \tau_1\|\bm{z}\|_2\right\}} $. This lemma is the complex counterpart to [\cite{RefWorks:RefID:34-cai2018solving}, Lemma 5.4]; as the proof technique is similar, we present the result without proof.

\begin{lemma}\label{concentration:1}
Fix $\tau_1 >0$ and let $\epsilon \in(0,1)$ be a sufficiently small constant. There exist absolute constants $C_3,C_4>0$ such that: if $m \geq C_3\epsilon^{-2} \log \epsilon^{-1} \cdot n$, then with probability at least $1-e^{-C_4m \epsilon^2}$,

$$
\begin{aligned}
& \left\| \frac{1}{m}\sum_{k=1}^m\left|\bm{a}_k^{*} \bm{z}\right|^2 \bm{a}_k \bm{a}_k^{*} \cdot\mathbb{1}_{\left\{\left|\bm{a}_k^{*} \bm{z}\right| \leq \tau_1\|\bm{z}\|_2\right\}}-\left(\hat{\beta_1} \bm{z} \bm{z}^{*}+\hat{\beta_2}\|\bm{z}\|_2^2 \bm{I}_n\right)\right\| \leq\rho_1\|\bm{z}\|_2^2
\end{aligned}
$$
holds for all $\bm{z} \in \mathbb{C}^n$.   
\end{lemma}


Next, we aim to estimate the concentration of the term $\frac{1}{m}\sum_{k=1}^{m}\Big[\left(\bm{a}_k^*\bm{z}\right)^2 \bm{a}_k \bm{a}_k^{T}\mathbb{1}_{\left\{\left|\bm{a}_k^*\bm{z}\right| \leq \tau_1\|\bm{z}\|_2\right\}}\Big]$ uniformly for all $\bm{z}\in\mathbb{C}^{n}$. We begin by showing that for a fixed $\bm{z}\in\mathbb{C}^{n}$, this matrix is concentrated around its expectation.

\begin{lemma}\label{lem: one point concentration of antidiag matirx at z}
Fix $\tau_1 > 0$ and let $\epsilon \in(0,1)$ be a sufficiently small constant. There exist absolute constants $C_5,C_6>0$ such that: if $m \geq C_5\epsilon^{-2} n$, then, for any $\bm{z}\in\mathbb{C}^{n}$, with probability at least $1-e^{-C_6m \epsilon^2}$,
\begin{equation}\label{eq: one point concentration of antidiag matrix at z}
 \left\|\frac{1}{m} \sum_{k=1}^m\left(\bm{a}_k^*\bm{z}\right)^2 \bm{a}_k \bm{a}_k^{\top}\cdot\mathbb{1}_{\left\{\left|\bm{a}_k^*\bm{z}\right| \leq \tau_1\|\bm{z}\|_2\right\}}-\left(\hat{\beta_1}+\hat{\beta_2}\right) \bm{z} \bm{z}^{\top}\right\| \leq\left(\tau_1^4+\tau_1^3+\tau_1^2\right)\epsilon\|\bm{z}\|_2^2. 
\end{equation} 
\end{lemma}

\begin{proof}
Without loss of generality, we assume that $\|\bm{z}\|_2=1$. Let us define the matrix $\bm{M}$ as:
$$\bm{M}:=\frac{1}{m} \sum_{k=1}^m\left(\bm{a}_k^*\bm{z}\right)^2 \bm{a}_k \bm{a}_k^{\top}\cdot\mathbb{1}_{\left\{\left|\bm{a}_k^*\bm{z}\right| \leq \tau_1\|\bm{z}\|_2\right\}}-\left(\hat{\beta_1}+\hat{\beta_2}\right) \bm{z} \bm{z}^{\top}.$$
We employ a standard covering argument to estimate $\|\bm{M}\|$. There exist vectors $\tilde{\bm{u}},\tilde{\bm{v}}\in\mathbb{S}^{n-1}$ such that $\|\bm{M}\|=|\tilde{\bm{u}}^*\bm{M}\tilde{\bm{v}}|$. Let $\mathcal{N}_{\frac14}$ be a $\frac14$-net for the unit sphere $\mathbb{S}^{n-1}$. Then, there exist $\bm{u}_0,\bm{v}_0\in\mathcal{N}_{\frac14}$ such that $\|\tilde{\bm{u}}-\bm{u}_0\|\leq \frac14$ and $\|\tilde{\bm{v}}-\bm{v}_0\|\leq \frac14$. Applying the triangle inequality and the definition of operator norm, we have:
\[\begin{split}
\|\bm{M}\|&=|\tilde{\bm{u}}^*\bm{M}\tilde{\bm{v}}|
\leq |\bm{u}_0^*\bm{M}\bm{v}_0|+|(\tilde{\bm{u}}-\bm{u}_0)^*\bm{M}\bm{v}_0|+|\tilde{\bm{u}}^*\bm{M}(\tilde{\bm{v}}-\bm{v}_0)|\\
&\leq |\bm{u}_0^*\bm{M}\bm{v}_0|+\|\bm{M}\|\|\tilde{\bm{u}}-\bm{u}_0\|\|\bm{v}_0\|+\|\bm{M}\|\|\tilde{\bm{v}}-\bm{v}_0\|\|\tilde{\bm{v}}\|\\
&\leq \mathop{\max}_{\bm{u},\bm{v}\in \mathcal{N}_{\frac14}}|\bm{u}^*\bm{M}\bm{v}|+\frac{2}{4}\|\bm{M}\|,
\end{split}\]
which implies that $\|\bm{M}\|\leq 2\cdot \max_{\bm{u},\bm{v}\in \mathcal{N}_{\frac14}}|\bm{u}^*\bm{M}\bm{v}|$. 
The remaining task is to estimate $\max_{\bm{u},\bm{v}\in \mathcal{N}_{\frac14}}|\bm{u}^*\bm{M}\bm{v}|$. 
For any chosen $\bm{u},\bm{v}\in \mathcal{N}_{\frac14}$, we use the following orthogonal decomposition:
$$\bm{u}=\bm{z}\bm{z}^*\bm{u}+\bm{z}_1=(\bm{z}^*\bm{u})\bm{z}+\bm{z}_1,\quad\bm{v}=\bar{\bm{z}}\bar{\bm{z}}^*\bm{v}+\bm{z}_2=(\bm{z}^{\top}\bm{v})\bar{\bm{z}}+\bm{z}_2$$ 
where $\bm{z}\perp\bm{z}_1$ and $\bar{\bm{z}}\perp\bm{z}_2$.
Then we can write:
\begin{align*}
|\bm{u}^*\bm{M}\bm{v}|=&\left|\frac{1}{m}\sum_{k=1}^m\left(\bm{a}_k^*\bm{z}\right)^2 \bm{u}^*\bm{a}_k \bm{a}_k^{\top}\bm{v}\cdot\mathbb{1}_{\left\{\left|\bm{a}_k^*\bm{z}\right| \leq\tau_1\|\bm{z}\|_2\right\}}-\left(\hat{\beta_1}+\hat{\beta_2}\right)\bm{u}^* \bm{z} \bm{z}^{\top}\bm{v}\right|\\
=&\left|\frac{1}{m}\sum_{k=1}^m\left(\bm{a}_k^*\bm{z}\right)^2 ((\bm{z}^*\bm{u})\bm{z}+\bm{z}_1)^*\bm{a}_k \bm{a}_k^{\top}((\bm{z}^{\top}\bm{v})\bar{\bm{z}}+\bm{z}_2)\cdot\mathbb{1}_{\left\{\left|\bm{a}_k^*\bm{z}\right| \leq \tau_1\|\bm{z}\|_2\right\}}-\left(\hat{\beta_1}+\hat{\beta_2}\right)\bm{u}^* \bm{z}\bm{z}^{\top}\bm{v}\right|\\
\leq&
\left|\frac{1}{m}\sum_{k=1}^m(\bm{u}^*\bm{z})(\bm{z}^{\top}\bm{v})|\bm{a}_k^*\bm{z}|^4\cdot\mathbb{1}_{\left\{\left|\bm{a}_k^*\bm{z}\right| \leq \tau_1\|\bm{z}\|_2\right\}}-\left(\hat{\beta_1}+\hat{\beta_2}\right)\bm{u}^* \bm{z}\bm{z}^{\top}\bm{v}\right|
\\
&\quad+
\left|\frac{1}{m}\sum_{k=1}^m(\bm{u}^*\bm{z})|\bm{a}_k^*\bm{z}|^2\bm{a}_k^*\bm{z}\bm{a}_k^{\top}\bm{z}_2\cdot\mathbb{1}_{\left\{\left|\bm{a}_k^*\bm{z}\right| \leq\tau_1\|\bm{z}\|_2\right\}}\right|\\
&\quad+
\left|\frac{1}{m}\sum_{k=1}^m(\bm{z}^{\top}\bm{v})|\bm{a}_k^*\bm{z}|^2\bm{a}_k^*\bm{z}\bm{z}_1^*\bm{a}_k\cdot\mathbb{1}_{\left\{\left|\bm{a}_k^*\bm{z}\right| \leq \tau_1\|\bm{z}\|_2\right\}}\right|\\
&\quad+
\left|\frac{1}{m}\sum_{k=1}^m\bm{z}_1^*\bm{a}_k\bm{a}_k^{\top}\bm{z}_2(\bm{a}_k^*\bm{z})^2
\cdot\mathbb{1}_{\left\{\left|\bm{a}_k^*\bm{z}\right| \leq\tau_1\|\bm{z}\|_2\right\}}\right|\\
:= &\mathbf{I}_1+\mathbf{I}_2+\mathbf{I}_3+\mathbf{I}_4.
\end{align*}

For $\mathbf{I}_1$, using the first expectation from Lemma \ref{gaussian:expectations}, we have
\begin{equation}\label{expectation:1}
\mathbb{E}\Big[(\bm{u}^*\bm{z})(\bm{z}^{\top}\bm{v})|\bm{a}_k^*\bm{z}|^4\cdot\mathbb{1}_{\left\{\left|\bm{a}_k^*\bm{z}\right| \leq \tau_1\|\bm{z}\|_2\right\}}\Big]=(\hat{\beta_1}+\hat{\beta_2})(\bm{u}^*\bm{z})(\bm{z}^{\top}\bm{v}).
\end{equation}
Since both $|\bm{u}^*\bm{z}|$ and $|\bm{z}^{\top}\bm{v}|$ are bounded by 1, we have $$\left|(\bm{u}^*\bm{z})(\bm{z}^{\top}\bm{v})|\bm{a}_k^*\bm{z}|^4\cdot\mathbb{1}_{\left\{\left|\bm{a}_k^*\bm{z}\right| \leq \tau_1\|\bm{z}\|_2\right\}}\right|\leq\tau_1^4.$$ Thus, by Hoeffding's inequality, we obtain
$$\mathbb{P}(\mathbf{I}_1\geq\textstyle{\frac{1}{2}}\tau_1^4\epsilon)\leq2\exp(-c_1'm\epsilon^2),$$
where $c_1'>0$ is an absolute constant.

For $\mathbf{I}_2$, given that $\bar{\bm{z}}\perp\bm{z}_2$, it implies that $\bm{a}_k^*\bm{z}$ and $\bm{a}_k^{\top}\bm{z}_2$ are independent. Therefore,
\begin{equation}\label{expectation:2}
\mathbb{E}\Big[(\bm{u}^*\bm{z})|\bm{a}_k^*\bm{z}|^2\bm{a}_k^*\bm{z}\bm{a}_k^{\top}\bm{z}_2\cdot\mathbb{1}_{\left\{\left|\bm{a}_k^*\bm{z}\right| \leq\tau_1\|\bm{z}\|_2\right\}}\Big]=0.
\end{equation}
Note that $|\bm{a}_k^*\bm{z}|^2\bm{a}_k^*\bm{z}\cdot\mathbb{1}_{\left\{\left|\bm{a}_k^*\bm{z}\right| \leq\tau_1\|\bm{z}\|_2\right\}}$ is bounded and $\bm{a}_k^{\top}\bm{z}_2$ follows a Gaussian distribution. Let $\|\cdot\|_{\psi_2}$ denote the sub-Gaussian norm. Then we have 
$$\left\|(\bm{u}^*\bm{z})|\bm{a}_k^*\bm{z}|^2\bm{a}_k^*\bm{z}\bm{a}_k^{\top}\bm{z}_2\cdot\mathbb{1}_{\left\{\left|\bm{a}_k^*\bm{z}\right| \leq\tau_1\|\bm{z}\|_2\right\}}\right\|_{\psi_2}\leq\tau_1^3\|\bm{a}_k^{\top}\bm{z}_2\|_{\psi_2}.$$
By Hoeffding's inequality, we have
$$\mathbb{P}(\mathbf{I}_2\geq\textstyle{\frac{1}{4}}\tau_1^3\epsilon)\leq2\exp(-c_2'm\epsilon^2),$$
where $c_2'>0$ is an absolute constant.

For $\mathbf{I}_3$, since $\bm{z}\perp\bm{z}_1$, similar to $\mathbf{I}_2$, we have
\begin{equation}\label{expectation:3}
\mathbb{E}\Big[(\bm{z}^{\top}\bm{v})|\bm{a}_k^*\bm{z}|^2\bm{a}_k^*\bm{z}\bm{z}_1^*\bm{a}_k\cdot\mathbb{1}_{\left\{\left|\bm{a}_k^*\bm{z}\right| \leq \tau_1\|\bm{z}\|_2\right\}}\Big]=0 ,   
\end{equation}
and
$$\mathbb{P}(\mathbf{I}_3\geq\textstyle{\frac{1}{4}}\tau_1^3\epsilon)\leq2\exp(-c_3'm\epsilon^2),$$
where $c_3'>0$ is an absolute constant.

For $\mathbf{I}_4$, by the second equality of Lemma \ref{gaussian:expectations}, we have 
\[
\mathbb{E}\Big[\left(\bm{a}_k^*\bm{z}\right)^2 \bm{u}^*\bm{a}_k \bm{a}_k^{\top}\bm{v}\cdot\mathbb{1}_{\left\{\left|\bm{a}_k^*\bm{z}\right| \leq\tau_1\|\bm{z}\|_2\right\}}\Big]=\left(\hat{\beta_1}+\hat{\beta_2}\right)\bm{u}^* \bm{z} \bm{z}^{\top}\bm{v}.
\]
This, combined with expectations \eqref{expectation:1}, \eqref{expectation:2}, and \eqref{expectation:3}, yields  
\begin{equation}\label{expectation:4}
\mathbb{E}[\bm{z}_1^*\bm{a}_k\bm{a}_k^{\top}\bm{z}_2(\bm{a}_k^*\bm{z})^2
\cdot\mathbb{1}_{\left\{\left|\bm{a}_k^*\bm{z}\right| \leq\tau_1\|\bm{z}\|_2\right\}}]=0.    
\end{equation}
Furthermore, $\bm{z}_1^*\bm{a}_k\bm{a}_k^{\top}\bm{z}_2$ follows a sub-exponential distribution, and $(\bm{a}_k^*\bm{z})^2\cdot\mathbb{1}_{\left\{\left|\bm{a}_k^*\bm{z}\right| \leq\tau_1\|\bm{z}\|_2\right\}}$ is bounded. Let $\|\cdot\|_{\psi_1}$ denote the sub-exponential norm. Then
$$\|\bm{z}_1^*\bm{a}_k\bm{a}_k^{\top}\bm{z}_2(\bm{a}_k^*\bm{z})^2\cdot\mathbb{1}_{\left\{\left|\bm{a}_k^*\bm{z}\right| \leq\tau_1\|\bm{z}\|_2\right\}}\|_{\psi_1}
\leq\tau_1^2\|\bm{z}_1^*\bm{a}_k\bm{a}_k^{\top}\bm{z}_2\|_{\psi_1}
\leq\tau_1^2\|\bm{z}_1^*\bm{a}_k\|_{\psi_2}\|\bm{a}_k^{\top}\bm{z}_2\|_{\psi_2}.$$
By Bernstein's inequality, for $\epsilon<1$, we have
$$\mathbb{P}(\mathbf{I}_4\geq\textstyle{\frac{1}{2}}\tau_1^2\epsilon)\leq2\exp(-c_4'm\epsilon^2),$$
where $c_4'>0$ is an absolute constant.

Combining the probabilities for the four terms yields
$\left|\bm{u}^*\bm{M}\bm{v}\right|\leq\frac{1}{2}(\tau_1^4+\tau_1^3+\tau_1^2)\epsilon$
with probability at least $1-e^{-2C_6m\epsilon^2}$ for some universal constant $C_6$.
By applying a union bound over $\bm{u}\in\mathcal{N}_{\frac14}$ and $\bm{v}\in\mathcal{N}_{\frac14}$, we obtain
$$\|\bm{M}\|\leq 2\cdot \max_{\bm{u},\bm{v}\in \mathcal{N}_{\frac14}}|\bm{u}^*\bm{M}\bm{v}|\leq(\tau_1^4+\tau_1^3+\tau_1^2)\epsilon,$$
which holds with probability at least
$
1-\big|\mathcal{N}_{\frac{1}{4}}\big|^{2}e^{-2C_6 m\epsilon^2}\geq1-e^{-2C_6 m\epsilon^2+2n\log 12} \geq 1- e^{-C_6 m\epsilon^2} ,
$
provided $m\ge C_5\epsilon^{-2}n$ for some positive constant $C_5$ satisfying $C_5C_6\geq2\log 12$.
\end{proof}

In the following, we extend Lemma \ref{lem: one point concentration of antidiag matirx at z} to hold uniformly for all $\bm{z}\in\mathbb{S}^{n-1}$. To achieve this, we first introduce two auxiliary lemmas.

\begin{lemma}\label{spectral}
Fix $\gamma \geq 1$ and let $\epsilon \in(0,1)$ be a sufficiently small constant. There exist absolute constants $C_7,C_8>0$ such that: if $m \geq C_7\epsilon^{-2} \log \epsilon^{-1} \cdot n$, then with probability at least $1-e^{-C_8m \epsilon^2}$, it holds that
$$
\left\|\frac{1}{m}\sum_{k=1}^m \bm{a}_k \bm{a}_k^{*} \cdot\mathbb{1}_{\left\{\left|\bm{a}_k^{*} \bm{z}\right|>\gamma\|\bm{z}\|_2\right\}}\right\| \leq 5 \gamma e^{-0.49 \gamma^2}+\epsilon,
\quad\forall \bm{z} \neq \bm{0}.
$$    
\end{lemma}

This lemma is the complex counterpart of [\cite{RefWorks:RefID:34-cai2018solving}, Lemma 5.3]. Its proof is similar, and thus is omitted.

\begin{lemma}\label{lem: bound of spectral norm of sum of a_ka_k^T}
Let $\{c_k\}_{k=1}^m\subset\mathbb{C}$, $\{\bm{a}_k\}_{k=1}^m\subset\mathbb{C}^n$ and $\{\sigma_{k}\}_{k=1}^m\in\{0,1\}$. Then, the following inequality holds:
\begin{equation}
\Big\| \frac{1}{m}\sum_{k=1}^m c_k\sigma_k\bm{a}_k\bm{a}_k^{\top}\Big\|
\leq
\max_{k}|c_k|\cdot\Big\|\frac{1}{m} \sum_{k=1}^m\sigma_k\bm{a}_k\bm{a}_k^*\Big\|.
\end{equation}
\end{lemma}
\begin{proof} 
We first consider the case where $\sigma_k=1$ for all $k \in \{1, \ldots, m\}$. The spectral norm can be bounded as follows:
\begin{align*}
&\Big\| \frac{1}{m}\sum_{k=1}^m c_k\bm{a}_k\bm{a}_k^{\top}\Big\|
=\sup_{\bm{u},\bm{v}\in\mathbb{S}^{n-1}}
\Big\langle \frac{1}{m}\sum_{k=1}^m c_k\bm{a}_k\bm{a}_k^{\top}, \bm{uv}^*\Big\rangle
=
\sup_{\bm{u},\bm{v}\in\mathbb{S}^{n-1}}
\frac{1}{m}\sum_{k=1}^m c_k
\cdot \Big\langle\bm{a}_k\bm{a}_k^{\top},\bm{uv}^*\Big\rangle\\
&\leq
\frac{1}{m}\sup_{\bm{u},\bm{v}\in\mathbb{S}^{n-1}}
\max_{k}|c_k|\cdot \sum_{k=1}^m\Big|\Big\langle\bm{a}_k\bm{a}_k^{\top},\bm{uv}^*\Big\rangle\Big|\leq
\frac{1}{m}\sup_{\bm{u},\bm{v}\in\mathbb{S}^{n-1}}
\max_{k}|c_k|\cdot \sum_{k=1}^m|\bm{u}^*\bm{a}_k||\bm{a}_k^{\top}\bm{v}|\\
&\leq
\sup_{\bm{u},\bm{v}\in\mathbb{S}^{n-1}}
\max_{k}|c_k|\cdot
\sqrt{ \frac{1}{m}\sum_{k=1}^m|\bm{u}^*\bm{a}_k|^2}
\sqrt{ \frac{1}{m}\sum_{k=1}^m|\bm{a}_k^{\top}\bm{v}|^2}\\
&=
\sup_{\bm{u},\bm{v}\in\mathbb{S}^{n-1}}
\max_{k}|c_k|\cdot
\sqrt{ \frac{1}{m}\sum_{k=1}^m\langle\bm{a}_k\bm{a}_k^*,\bm{uu}^*\rangle}
\sqrt{ \frac{1}{m}\sum_{k=1}^m\langle\overline{\bm{a}_k\bm{a}_k^*},\bm{vv}^*\rangle}\\
&\leq
\frac{1}{m}\max_{k}|c_k|\Big\| \sum_{k=1}^m\bm{a}_k\bm{a}_k^*\Big\|.
\end{align*}
For the general case where $\sigma_k \in \{0,1\}$, since $\sigma_k^2=\sigma_k$, we can define effective vectors $\sigma_k\bm{a}_k$. Applying the result from the previous case (where all $\sigma_k=1$) to the terms $c_k(\sigma_k\bm{a}_k)(\sigma_k\bm{a}_k^{\top})$, we directly obtain:
$$\Big\| \frac{1}{m}\sum_{k=1}^m c_k\sigma_k\bm{a}_k\bm{a}_k^{\top}\Big\|=\Big\| \frac{1}{m}\sum_{k=1}^m c_k(\sigma_k\bm{a}_k)(\sigma_k\bm{a}_k^{\top})\Big\|
\leq
\max_{k}|c_k|\cdot\Big\|\frac{1}{m} \sum_{k=1}^m\sigma_k\bm{a}_k\bm{a}_k^*\Big\|.$$
\end{proof}



With the help of the above two lemmas, we can now extend Lemma \ref{lem: one point concentration of antidiag matirx at z} to hold uniformly for all $\bm{z}\in\mathbb{S}^{n-1}$.

\begin{lemma}\label{concentration:2}
Fix $\tau_1 \geq 2$ and let $\epsilon \in(0,1)$ be a sufficiently small constant. There exist universal constants $C_9,C_{10}>0$ such that: if $m \geq C_9\epsilon^{-2} \log \epsilon^{-1} \cdot n$, then with probability at least $1-e^{-C_{10}m \epsilon^2}$, it holds that
$$
\begin{aligned}
& \left\| \frac{1}{m}\sum_{k=1}^m\left(\bm{a}_k^*\bm{z}\right)^2 \bm{a}_k \bm{a}_k^{\top}\cdot\mathbb{1}_{\left\{\left|\bm{a}_k^*\bm{z}\right| \leq \tau_1\|\bm{z}\|_2\right\}}-\left(\hat{\beta_1}+\hat{\beta_2}\right) \bm{z} \bm{z}^{\top}\right\| 
\leq\rho_1\|\bm{z}\|_2^2, \quad \forall\bm{z} \in \mathbb{C}^n.
\end{aligned}
$$
\end{lemma}
\begin{proof}
Let $\mathcal{N}_{\epsilon^2}$ be an $\epsilon^2$-net of the unit sphere $\mathbb{S}^{n-1}$. For any $\bm{z}\in\mathbb{S}^{n-1}$, we choose an element $\bm{z}_0\in \mathcal{N}_{\epsilon^2}$ such that $\|\bm{z}-\bm{z}_0\|_2\leq\epsilon^2$. To bound the difference between the expression evaluated at $\bm{z}$ and $\bm{z}_0$, we decompose it as follows:
\begin{align*}
&\left\|\frac{1}{m}\sum_{k=1}^m(\bm{a}_k^*\bm{z})^2\bm{a}_k\bm{a}_k^{\top}\cdot\mathbb{1}_{\left\{\left|\bm{a}_k^*\bm{z}\right| \leq \tau_1\|\bm{z}\|_2\right\}}-
\frac{1}{m}\sum_{k=1}^m(\bm{a}_k^*\bm{z}_0)^2\bm{a}_k\bm{a}_k^{\top}\cdot\mathbb{1}_{\left\{\left|\bm{a}_k^*\bm{z}_0\right| \leq \tau_1\|\bm{z}_0\|_2\right\}}\right\|\\
\leq&
\left\|\frac{1}{m}\sum_{k=1}^m\left((\bm{a}_k^*\bm{z})^2-(\bm{a}_k^*\bm{z}_0)^2\right)\bm{a}_k\bm{a}_k^{\top}\cdot\mathbb{1}_{\left\{\left|\bm{a}_k^*\bm{z}\right| \leq \tau_1\|\bm{z}\|_2\right\}}
\mathbb{1}_{\left\{\left|\bm{a}_k^*\bm{z}_0\right| \leq \tau_1\|\bm{z}_0\|_2\right\}}\right\|\\
&\quad+
\left\|\frac{1}{m}\sum_{k=1}^m(\bm{a}_k^*\bm{z})^2\bm{a}_k\bm{a}_k^{\top}\cdot\mathbb{1}_{\left\{\left|\bm{a}_k^*\bm{z}\right| \leq \tau_1\|\bm{z}\|_2\right\}}
\mathbb{1}_{\left\{\left|\bm{a}_k^*\bm{z}_0\right| > \tau_1\|\bm{z}_0\|_2\right\}}\right\|\\
&\quad+
\left\|\frac{1}{m}\sum_{k=1}^m(\bm{a}_k^*\bm{z}_0)^2\bm{a}_k\bm{a}_k^{\top}\cdot\mathbb{1}_{\left\{\left|\bm{a}_k^*\bm{z}\right| >\tau_1\|\bm{z}\|_2\right\}}
\mathbb{1}_{\left\{\left|\bm{a}_k^*\bm{z}_0\right| \leq \tau_1\|\bm{z}_0\|_2\right\}}\right\|\\
\leq&
\left\|\frac{1}{m}\sum_{k=1}^m\left((\bm{a}_k^*\bm{z})^2-(\bm{a}_k^*\bm{z}_0)^2\right)\bm{a}_k\bm{a}_k^{\top}\cdot
\mathbb{1}_{\left\{\left|\bm{a}_k^*\bm{z}\right| \leq \tau_1\|\bm{z}\|_2\right\}}
\mathbb{1}_{\left\{\left|\bm{a}_k^*\bm{z}_0\right| \leq \tau_1\|\bm{z}_0\|_2\right\}}
\mathbb{1}_{\left\{\left|\bm{a}_k^*(\bm{z}-\bm{z}_0)\right| \leq \epsilon\right\}}
\right\|\\
&\quad+
\left\|
\frac{1}{m}\sum_{k=1}^m\left((\bm{a}_k^*\bm{z})^2-(\bm{a}_k^*\bm{z}_0)^2\right)\bm{a}_k\bm{a}_k^{\top}
\cdot
\mathbb{1}_{\left\{\left|\bm{a}_k^*\bm{z}\right| \leq \tau_1\|\bm{z}\|_2\right\}}
\mathbb{1}_{\left\{\left|\bm{a}_k^*\bm{z}_0\right| \leq \tau_1\|\bm{z}_0\|_2\right\}}
\mathbb{1}_{\left\{\left|\bm{a}_k^*(\bm{z}-\bm{z}_0)\right| > \epsilon\right\}}
\right\|\\
&\quad+
\left\|\frac{1}{m}\sum_{k=1}^m(\bm{a}_k^*\bm{z})^2\bm{a}_k\bm{a}_k^{\top}\cdot\mathbb{1}_{\left\{\left|\bm{a}_k^*\bm{z}\right| \leq \tau_1\|\bm{z}\|_2\right\}}
\mathbb{1}_{\left\{\left|\bm{a}_k^*\bm{z}_0\right| > \tau_1\|\bm{z}_0\|_2\right\}}\right\|\\
&\quad+
\left\|\frac{1}{m}\sum_{k=1}^m(\bm{a}_k^*\bm{z}_0)^2\bm{a}_k\bm{a}_k^{\top}\cdot
\mathbb{1}_{\left\{\left|\bm{a}_k^*\bm{z}\right| >\tau_1\|\bm{z}\|_2\right\}}
\mathbb{1}_{\left\{\left|\bm{a}_k^*\bm{z}_0\right| \leq \tau_1\|\bm{z}_0\|_2\right\}}\right\|\\
&:= \mathbf{I}_1+\mathbf{I}_2+\mathbf{I}_3+\mathbf{I}_4.
\end{align*}
We now proceed to estimate the upper bounds for each of these four terms.

For $\mathbf{I}_1$, we note that the coefficient $\left((\bm{a}_k^*\bm{z})^2-(\bm{a}_k^*\bm{z}_0)^2\right)$ can be factored as $\bm{a}_k^*(\bm{z}-\bm{z}_0)(\bm{a}_k^*\bm{z}+\bm{a}_k^*\bm{z}_0)$. Given the indicator functions $\mathbb{1}_{\left\{\left|\bm{a}_k^*\bm{z}\right| \leq \tau_1\|\bm{z}\|_2\right\}}$ and $\mathbb{1}_{\left\{\left|\bm{a}_k^*\bm{z}_0\right| \leq \tau_1\|\bm{z}_0\|_2\right\}}$, and noting $\|\bm{z}\|_2 = \|\bm{z}_0\|_2 = 1$, we have $|\bm{a}_k^*\bm{z}| \leq \tau_1$ and $|\bm{a}_k^*\bm{z}_0| \leq \tau_1$. Thus, $|\bm{a}_k^*\bm{z}+\bm{a}_k^*\bm{z}_0| \leq |\bm{a}_k^*\bm{z}| + |\bm{a}_k^*\bm{z}_0| \leq 2\tau_1$. Additionally, the third indicator $\mathbb{1}_{\left\{\left|\bm{a}_k^*(\bm{z}-\bm{z}_0)\right| \leq \epsilon\right\}}$ restricts $|\bm{a}_k^*(\bm{z}-\bm{z}_0)| \leq \epsilon$. Therefore, the absolute value of the scalar coefficient is bounded by $\epsilon \cdot 2\tau_1 = 2\tau_1\epsilon$. Applying Lemma \ref{lem: bound of spectral norm of sum of a_ka_k^T} and noting that $\|\frac{1}{m}\sum_{k=1}^m\bm{a}_k\bm{a}_k^*\|$ is bounded (e.g., by 2 for sufficiently large $m$ with high probability, as established by Lemma \ref{lem8.1}), we get:
$$\mathbf{I}_1\leq\ 2\tau_1\epsilon\left\|\frac{1}{m}\sum_{k=1}^m\bm{a}_k\bm{a}_k^*\right\|\leq2\tau_1^2\epsilon.$$
This holds if $m \geq 2 \epsilon^{-2} n$, with probability at least $1-2 e^{-m \epsilon^2 / 2}$.

For $\mathbf{I}_2$, the scalar coefficient $\left((\bm{a}_k^*\bm{z})^2-(\bm{a}_k^*\bm{z}_0)^2\right)$ restricted by $\mathbb{1}_{\left\{\left|\bm{a}_k^*\bm{z}\right| \leq \tau_1\|\bm{z}\|_2\right\}}
\cdot\mathbb{1}_{\left\{\left|\bm{a}_k^*\bm{z}_0\right| \leq \tau_1\|\bm{z}_0\|_2\right\}}
$ is bounded in magnitude by $|(\bm{a}_k^*\bm{z})^2| + |(\bm{a}_k^*\bm{z}_0)^2| \leq \tau_1^2\|\bm{z}\|_2^2 + \tau_1^2\|\bm{z}_0\|_2^2 = 2\tau_1^2$. Applying Lemma \ref{lem: bound of spectral norm of sum of a_ka_k^T} and substituting $\|\bm{z}-\bm{z}_0\|_2\leq\epsilon^2$ yields:
$$\mathbf{I}_2\leq2\tau_1^2\left\| \frac{1}{m}\sum_{k=1}^m\bm{a}_k\bm{a}_k^*\mathbb{1}_{\left\{\left|\bm{a}_k^*(\bm{z}-\bm{z}_0)\right| > \epsilon\right\}}
\right\|.$$
Since $\|\bm{z}-\bm{z}_0\|_2 \leq \epsilon^2$, the condition $\left|\bm{a}_k^*(\bm{z}-\bm{z}_0)\right| > \epsilon$ implies $\frac{\left|\bm{a}_k^*(\bm{z}-\bm{z}_0)\right|}{\|\bm{z}-\bm{z}_0\|_2} > \frac{\epsilon}{\|\bm{z}-\bm{z}_0\|_2} \ge \frac{\epsilon}{\epsilon^2} = \frac{1}{\epsilon}$. Applying Lemma \ref{spectral} with $\gamma = 1/\epsilon$ then gives:
$$\mathbf{I}_2\leq2\tau_1^2\left(5 \epsilon^{-1} e^{-0.49 \epsilon^{-2}}+\epsilon\right)$$
with probability at least $1-e^{-C_8m \epsilon^2}$ provided that $m \geq C_7\epsilon^{-2} \log \epsilon^{-1} \cdot n$.

For $\mathbf{I}_3$, the scalar coefficient is $|(\bm{a}_k^*\bm{z})^2|$ restricted by $\mathbb{1}_{\left\{\left|\bm{a}_k^*\bm{z}\right| \leq \tau_1\|\bm{z}\|_2\right\}}$, so its magnitude is bounded by $\tau_1^2$. Applying Lemma \ref{lem: bound of spectral norm of sum of a_ka_k^T} and Lemma \ref{spectral} (with $\gamma = \tau_1$) yields:
$$\mathbf{I}_3\leq\tau_1^2\left\|\frac{1}{m} \sum_{k=1}^m\bm{a}_k\bm{a}_k^*\mathbb{1}_{\left\{\left|\bm{a}_k^*\bm{z}_0\right| > \tau_1\|\bm{z}_0\|_2\right\}}
\right\|
\leq5 \tau_1^3 e^{-0.49 \tau_1^2}+\tau_1^2\epsilon$$
with probability at least $1-e^{-C_8m \epsilon^2}$ provided that $m \geq C_7\epsilon^{-2} \log \epsilon^{-1} \cdot n$.

For $\mathbf{I}_4$, following a similar argument as for $\mathbf{I}_3$, we have:
$$\mathbf{I}_4\leq\tau_1^2\left\|\frac{1}{m}\sum_{k=1}^m\bm{a}_k\bm{a}_k^*\mathbb{1}_{\left\{\left|\bm{a}_k^*\bm{z}\right| > \tau_1\|\bm{z}\|_2\right\}}
\right\|
\leq5 \tau_1^3 e^{-0.49 \tau_1^2}+\tau_1^2\epsilon
,$$
with probability at least $1-e^{-C_8m \epsilon^2}$ provided that $m \geq C_7\epsilon^{-2} \log \epsilon^{-1} \cdot n$.

Combining the upper bounds for these four terms, we obtain:
\begin{align}
&\left\|\frac{1}{m}\sum_{k=1}^m(\bm{a}_k^*\bm{z})^2\bm{a}_k\bm{a}_k^{\top}\cdot\mathbb{1}_{\left\{\left|\bm{a}_k^*\bm{z}\right| \leq \tau_1\|\bm{z}\|_2\right\}}-
\frac{1}{m}\sum_{k=1}^m(\bm{a}_k^*\bm{z}_0)^2\bm{a}_k\bm{a}_k^{\top}\cdot\mathbb{1}_{\left\{\left|\bm{a}_k^*\bm{z}_0\right| \leq \tau_1\|\bm{z}_0\|_2\right\}}\right\|\notag\\
\leq&10\tau_1^3e^{-0.49\tau_1^2}+\dfrac{10\tau_1^2}{\epsilon}e^{-0.49\epsilon^{-2}}+6\tau_1^2\epsilon\label{eq: error of antidiag when pretrub z}.
\end{align}
This inequality holds with probability at least $1-5e^{-C_8m \epsilon^2}$, provided $m \geq C_{11}\epsilon^{-2} \log \epsilon^{-1} \cdot n$, where $C_{11}=\max\{2,C_7\}$.

Next, by applying Lemma \ref{lem: one point concentration of antidiag matirx at z} and a union bound over the $\epsilon^2$-net $\mathcal{N}_{\epsilon^2}$, we obtain:
\begin{equation}\label{netuniform}
\begin{split}
&\left\|\frac{1}{m}\sum_{k=1}^m(\bm{a}_k^*\bm{z}_0)^2\bm{a}_k\bm{a}_k^{\top}\cdot\mathbb{1}_{\left\{\left|\bm{a}_k^*\bm{z}_0\right| \leq \tau_1\|\bm{z}_0\|_2\right\}}-(\hat{\beta_1}+\hat{\beta_2}) \bm{z}_0 \bm{z}_0^{\top}
\right\|\\
&\leq\sup_{\bm{u}\in\mathcal{N}_{\epsilon^2}}\left\|\frac{1}{m}\sum_{k=1}^m(\bm{a}_k^*\bm{u})^2\bm{a}_k\bm{a}_k^{\top}\cdot\mathbb{1}_{\left\{\left|\bm{a}_k^*\bm{u}\right| \leq \tau_1\|\bm{u}\|_2\right\}}-(\hat{\beta_1}+\hat{\beta_2}) \bm{u} \bm{u}^{\top}
\right\|\leq(\tau_1^4+\tau_1^3+\tau_1^2)\epsilon,\\
\end{split}
\end{equation}
with probability at least
$
1-\big|\mathcal{N}_{\epsilon^2}|e^{-2C_{10} m\epsilon^2}\geq1-e^{-2C_{10} m\epsilon^2+n\log3\epsilon^{-2}} \geq 1- e^{-C_{10}m\epsilon^2} ,
$
provided $m\ge C_{9}\epsilon^{-2}\log\epsilon^{-1}\cdot n$  for some positive constant $C_9$ such that $C_9C_{10}\geq2+\log3/\log\epsilon^{-1}$.

Thus, by combining \eqref{eq: error of antidiag when pretrub z} and \eqref{netuniform}, using the triangle inequality:
\begin{align}
&\left\|\frac{1}{m}
\sum_{k=1}^m(\bm{a}_k^*\bm{z})^2\bm{a}_k\bm{a}_k^{\top}\cdot\mathbb{1}_{\left\{\left|\bm{a}_k^*\bm{z}\right| \leq \tau_1\|\bm{z}\|_2\right\}}
-
(\hat{\beta_1}+\hat{\beta_2}) \bm{z} \bm{z}^{\top}
\right\|\notag\\
\leq&\left\|
\frac{1}{m}\sum_{k=1}^m(\bm{a}_k^*\bm{z})^2\bm{a}_k\bm{a}_k^{\top}\cdot\mathbb{1}_{\left\{\left|\bm{a}_k^*\bm{z}\right| \leq \tau_1\|\bm{z}\|_2\right\}}
-
\frac{1}{m}\sum_{k=1}^m(\bm{a}_k^*\bm{z}_0)^2\bm{a}_k\bm{a}_k^{\top}\cdot\mathbb{1}_{\left\{\left|\bm{a}_k^*\bm{z}_0\right| \leq \tau_1\|\bm{z}_0\|_2\right\}}
\right\|\notag\\
&+\left\|\frac{1}{m}\sum_{k=1}^m(\bm{a}_k^*\bm{z}_0)^2\bm{a}_k\bm{a}_k^{\top}\cdot\mathbb{1}_{\left\{\left|\bm{a}_k^*\bm{z}_0\right| \leq \tau_1\|\bm{z}_0\|_2\right\}}-(\hat{\beta_1}+\hat{\beta_2}) \bm{z}_0 \bm{z}_0^{\top}
\right\|+\Big\|(\hat{\beta_1}+\hat{\beta_2})(\bm{z}\bm{z}^{\top}-\bm{z}_0 \bm{z}_0^{\top})\Big\|
\notag\\
\leq&10\tau_1^3e^{-0.49\tau_1^2}+\dfrac{10\tau_1^2}{\epsilon}e^{-0.49\epsilon^{-2}}+(\tau_1^4+\tau_1^3+7\tau_1^2)\epsilon+4\epsilon^2\notag
\leq \rho_1.
\end{align}
This completes the proof of Lemma \ref{concentration:2}.
\end{proof}

Thus, based on Lemma \ref{concentration:1} and Lemma \ref{concentration:2}, we can now proceed to prove Lemma \ref{lemma:uniform_concentration}.

\begin{proof}[Proof of Lemma \ref{lemma:uniform_concentration}]
Starting from the reformulation in \eqref{reformulation}, we can express the difference between the sample sum and its expectation as follows:
\begin{small}
    $$\begin{aligned}  &\left|\frac{1}{m}\sum\limits_{k=1}^{m}\left(\operatorname{Re}\left(\bm{w}^* \bm{a}_k \bm{a}_k^* \bm{z}\right)\right)^2 \cdot \mathbb{1}_{\left\{\left|\bm{a}_k^{*} \bm{z}\right| \leq \tau_1\|\bm{z}\|_2\right\}} - \mathbb{E}\left[\frac{1}{m}\left(\operatorname{Re}\left(\bm{w}^* \bm{a}_k \bm{a}_k^* \bm{z}\right)\right)^2\cdot\mathbb{1}_{\left\{\left|\bm{a}_k^{*} \bm{z}\right| \leq \tau_1\|\bm{z}\|_2\right\}}\right]\right|\\ & =\Bigg|\frac{1}{4m} \sum_{k=1}^m\left[\begin{array}{l}\bm{w} \\ \overline{\bm{w}}\end{array}\right]^*\left[\begin{array}{ll}\left|\bm{a}_k^* \bm{x}\right|^2 \bm{a}_k \bm{a}_k^* -\mathbb{E}(\left|\bm{a}_k^* \bm{x}\right|^2 \bm{a}_k \bm{a}_k^* )& \left(\bm{a}_k^* \bm{x}\right)^2 \bm{a}_k \bm{a}_k^{\top} -\mathbb{E}(\left(\bm{a}_k^* \bm{x}\right)^2 \bm{a}_k \bm{a}_k^{\top} )\\ \left(\overline{\bm{a}_k^* \bm{x}}\right)^2 \overline{\bm{a}}_k \bm{a}_k^*-\mathbb{E}(\left(\overline{\bm{a}_k^* \bm{x}}\right)^2 \overline{\bm{a}}_k \bm{a}_k^*) & \left|\bm{a}_k^* \bm{x}\right|^2 \overline{\bm{a}}_k \bm{a}_k^{\top} - \mathbb{E}(\left|\bm{a}_k^* \bm{x}|\right)^2 \overline{\bm{a}}_k \bm{a}_k^{\top})\end{array}\right]\left[\begin{array}{l}\bm{w} \\ \overline{\bm{w}}\end{array}\right]\cdot \mathbb{1}_{\left\{\left|\bm{a}_k^{*} \bm{z}\right| \leq \tau_1\|\bm{z}\|_2\right\}}\Bigg|\\ &\leq \frac{1}{2}\Big\|\frac{1}{m} \sum_{k=1}^m\left|\bm{a}_k^{*} \bm{z}\right|^2 \bm{a}_k \bm{a}_k^{*} \cdot\mathbb{1}_{\left\{\left|\bm{a}_k^{*} \bm{z}\right| \leq \tau_1\|\bm{z}\|_2\right\}}-\left(\hat{\beta_1} \bm{z} \bm{z}^{*}+\hat{\beta_2}\|\bm{z}\|_2^2 \bm{I}_n\right)\Big\|\|\bm{w}\|_2^2 \\ &+
    \frac{1}{2} \Big\|\frac{1}{m}\sum_{k=1}^m\left(\bm{a}_k^*\bm{z}\right)^2\bm{a}_k \bm{a}_k^{\top}\cdot\mathbb{1}_{\left\{\left|\bm{a}_k^*\bm{z}\right| \leq \tau_1\|\bm{z}\|_2\right\}}
-\left(\hat{\beta_1}+\hat{\beta_2}\right)\bm{z} \bm{z}^{\top}\Big\|\|\bm{w}\|_2^2
\leq \rho_1\|\bm{w}\|_2^2\|\bm{z}\|_2^2,  \end{aligned}
$$ 
\end{small}
where the final inequality is obtained by applying the bounds from Lemma \ref{concentration:1} and Lemma \ref{concentration:2}. This concludes the proof.
\end{proof}

\section{Proofs of the Main Results}\label{sec5}
In this section, we prove Theorem \ref{thm:3}. In Section \ref{sec5.1}, we present the key propositions used for the proof of Theorem \ref{thm:3}. In Section \ref{sec5.2}, we establish the full proof of the main result of Theorem \ref{thm:3}.

\subsection{Key Ingredients}\label{sec5.1}
The concentration of $\frac{1}{\sqrt{m}}\mathcal{A}$ restricted to ${\mathbb{T}_{\bm{Z}}}$ is fundamental to solving the phase retrieval problem. To the best of our knowledge, however, the existing choice of metric allows one to establish only a restricted well-conditionedness for $\frac{1}{\sqrt{m}}\mathcal{A}$ on the tangent space $\mathbb{T}_{\bm{Z}}$, rather than a restricted nearly isometry property. This limitation prevents the derivation of optimal theoretical guarantees for signal recovery. Consequently, it is crucial to design a suitable metric that ensures the operator $\frac{1}{\sqrt{m}}\mathcal{A}$, when restricted to $\mathbb{T}_{\bm{Z}}$ and subject to certain truncations, acts as a near isometry.

We first prove that our proposed metric \eqref{new metric} establishes a restricted near-isometry for $\frac{1}{\sqrt{m}}\mathcal{A}$, subject to the truncation defined in \eqref{truncation:guassian}, under the complex Gaussian model \eqref{eq:Gaussian}. While the restricted near-isometry property for real Gaussian measurements can be derived by slightly modifying the proof of \cite[Proposition 4.1]{RefWorks:RefID:34-cai2018solving}, extending this to complex random Gaussian measurements presents specific challenges. The primary obstacle arises from the need to estimate $\left(\operatorname{Re}\left(\bm{w}^* \bm{a}_k \bm{a}_k^* \bm{z}\right)\right)^2$. Because the imaginary component is discarded, significant energy is lost, rendering the lower bound of the isometry difficult to establish.

\begin{pro}[Restricted Near Isometry Property] \label{RIP:Gaussian}
Let $\{\bm{a}_k\}_{k=1}^{m}$ follow the random Gaussian model in \eqref{eq:Gaussian}. Let $\tau_0, \tau_1, \tau_2>0$ be truncation parameters and $\epsilon > 0$ be sufficiently small. Then, there exist constants $C_{12}, C_{13} > 0$ such that whenever the sample size satisfies $m \geq C_{12} \epsilon^{-2}\log(1/\epsilon)n$, the following event occurs with probability at least $1-5e^{-C_{13} m\epsilon^2}$: For all positive semi-definite $\bm{Z}\in\mathcal{M}_1$ satisfying $\left\|\bm{Z} -\bm{X} \right\|_F \leq \frac{1}{14}\|\bm{X}\|_F$,
\begin{equation}\label{rip}
(\hat{\beta_1}-\delta)\|\bm{W}\|_{\mathfrak{o}}^2 \leq \frac{1}{m}\|\mathcal{A}_{\bm{Z}}(\bm{W})\|_2^2 \leq (\hat{\beta_2}+\delta)\|\bm{W}\|_{\mathfrak{o}}^2,
\qquad \forall~\bm{W}\in\mathbb{T}_{\bm{Z}},
\end{equation} 
where $\delta$ is defined as $\delta = 2(\rho_1+\rho_2)$.
\end{pro}

\begin{proof} 
For any positive semi-definite $\bm{Z}\in\mathcal{M}_1$, we can write $\bm{Z}=\bm{zz}^*$ for some $\bm{z}\in\mathbb{C}^{n}$. Subsequently, for any $\bm{W}\in\mathbb{T}_{\bm{Z}}$, it holds that $\bm{W}=\bm{z}\bm{w}^*+\bm{w}\bm{z}^*$ for some $\bm{w}\in\mathbb{C}^{n}.$
Without loss of generality, we can assume $\bm{z}^*\bm{w} \in \mathbb{R}$. This is permissible because for any $\bm{w}\in \mathbb{C}^n$, we can find a $\lambda \in \mathbb{R}$ such that the vector $\bm{w}' = \bm{w}-\imath\lambda\bm{z}$ satisfies $\bm{z}^*\bm{w}' \in \mathbb{R}$. An easy calculation confirms that substituting $\bm{w}'$ for $\bm{w}$ leaves $\bm{W}$ unchanged:
\begin{align*}
\bm{z}\bm{w}^*+\bm{w}\bm{z}^* = \bm{z}(\bm{w}-\imath\lambda\bm{z})^*+(\bm{w}-\imath\lambda\bm{z})\bm{z}^*.
\end{align*}
Given this simplification, the norm $\|\bm{W}\|_{\mathfrak{o}}^2$ is given by:
\begin{align*}
\|\bm{W}\|_{\mathfrak{o}}^2 = 2\|\bm{z}\|_2^2\|\bm{w}\|_2^2+6\left|\bm{z}^{*} \bm{w}\right|^2.
\end{align*}
Furthermore, by direct calculation, the scaled norm of $\mathcal{A}_{\bm{Z}}(\bm{W})$ can be expressed as:
\begin{align*}
\frac{1}{m}\|\mathcal{A}_{\bm{Z}}(\bm{W})\|_2^2=\frac{4}{m}\sum\limits_{k=1}^{m}\left(\operatorname{Re}\left(\bm{w}^* \bm{a}_k \bm{a}_k^* \bm{z}\right)\right)^2 \cdot \mathbb{1}_{\mathcal{E}_0^k(\bm{y}) \cap \mathcal{E}_1^k(\bm{z}) \cap \mathcal{E}_2^k(\bm{z})}.
\end{align*}

\begin{itemize}
    \item \textbf{Upper Bound:}
    To derive the upper bound, we utilize the inequality $\mathbb{1}_{\mathcal{E}_0^k(\bm{y}) \cap \mathcal{E}_1^k(\bm{z}) \cap \mathcal{E}_2^k(\bm{z})} \leq \mathbb{1}_{\mathcal{E}_1^k(\bm{z})}$. This allows us to express the term as:
    \[\begin{split}
   \frac{1}{m}\|\mathcal{A}_{\bm{Z}}(\bm{W})\|_2^2& \leq \frac{4}{m}\sum\limits_{k=1}^{m}\left(\operatorname{Re}\left(\bm{w}^* \bm{a}_k \bm{a}_k^* \bm{z}\right)\right)^2 \cdot \mathbb{1}_{\mathcal{E}_1^k(\bm{z})}\\
   &= \frac{4}{m}\sum\limits_{k=1}^{m}\left(\left(\operatorname{Re}\left(\bm{w}^* \bm{a}_k \bm{a}_k^* \bm{z}\right)\right)^2 \cdot \mathbb{1}_{\{{\left|\bm{a}_k^{*} \bm{z}\right| \leq \tau_1\|\bm{z}\|_2}\}} - \mathbb{E}\left[\left(\operatorname{Re}\left(\bm{w}^* \bm{a}_k \bm{a}_k^* \bm{z}\right)\right)^2\cdot\mathbb{1}_{\{{\left|\bm{a}_k^{*} \bm{z}\right| \leq \tau_1\|\bm{z}\|_2}\}} \right]\right)\\ 
   &+\frac{4}{m}\sum\limits_{k=1}^{m} \mathbb{E}\left[\left(\operatorname{Re}\left(\bm{w}^* \bm{a}_k \bm{a}_k^* \bm{z}\right)\right)^2\cdot\mathbb{1}_{\{{\left|\bm{a}_k^{*} \bm{z}\right| \leq \tau_1\|\bm{z}\|_2}\}} \right] \\ & :=  4\mathbf{I}_{11} +  4\mathbf{I}_{12}\end{split}\]
    Applying Lemma \ref{lemma:uniform_concentration}, there exist universal constants $C_1 ,C_2>0$ such that if $m \geq C_1\epsilon^{-2} \log \epsilon^{-1} n$, then with probability at least $1-e^{-C_2m\epsilon^2}$, we have:
    \begin{equation}\label{I11}
   -\rho_1 \|\bm{z}\|_2^2 \|\bm{w}\|_2^2\leq \mathbf{I}_{11} \leq \rho_1 \|\bm{z}\|_2^2 \|\bm{w}\|_2^2. 
    \end{equation}
    Furthermore, from Lemma \ref{gaussian:expectations}, we obtain:
    \begin{equation}\label{I12}
    \mathbf{I}_{12} = \frac{1}{2} \hat{\beta_2} \|\bm{z}\|_2^2 \|\bm{w}\|_2^2 + (\hat{\beta_1}+\frac{\hat{\beta_2}}{2})|\bm{z^*}\bm{w}|^2.
    \end{equation}
    Combining these results, we establish the upper bound:
    \begin{equation}\label{upper bound}
    \begin{split}
    \frac{1}{m}\|\mathcal{A}_{\bm{Z}}(\bm{W})\|_2^2&\leq 4\rho_1 \|\bm{z}\|_2^2 \|\bm{w}\|_2^2 +2 \hat{\beta_2} \|\bm{z}\|_2^2 \|\bm{w}\|_2^2 + 4(\hat{\beta_1}+\frac{\hat{\beta_2}}{2})|\bm{z^*}\bm{w}|^2 \\ &\leq  4\rho_1 \|\bm{z}\|_2^2 \|\bm{w}\|_2^2 +2 \hat{\beta_2}\|\bm{z}\|_2^2 \|\bm{x}\|_2^2 + 6\hat{\beta_2}|\bm{z^*}\bm{w}|^2 \\ &\leq (2\rho_1+\hat{\beta_2})\|\bm{W}\|_\mathfrak{o}^2
    \leq (\hat{\beta_2} +\delta)\|\bm{W}\|_\mathfrak{o}^2.
    \end{split}
    \end{equation}
   \item \textbf{Lower Bound:}
    For the lower bound, let $\bm{h}=\bm{z}-\bm{x}$. By Lemma \ref{uuu11}, the condition $\left\|\bm{z} \bm{z}^{*}-\bm{x} \bm{x}^{*}\right\|_F \leq \frac{1}{14}\|\bm{X}\|_F$ implies $\|\bm{h}\|_2 \leq \frac{1}{11}\|\bm{z}\|_2$.
    Furthermore, Lemma \ref{lem8.1} implies that if $\epsilon\in(0,0.01)$, then $\delta_1=4\epsilon^2+4\epsilon\in(0,0.05)$. Consequently, if $m \geq 2 \epsilon^{-2} n$, then with probability at least $1-2 e^{-m \epsilon^2 / 2}$, we have:
 \[
\left\{\left|\bm{a}_k^{*} \bm{x}\right|\leq0.9 \tau_0\|\bm{x}\|_2\right\}\subset\mathcal{E}_0^k(\bm{y})\subset\left\{\left|\bm{a}_k^{*} \bm{x}\right|\leq1.1 \tau_0\|\bm{x}\|_2\right\}.
 \]  
    Similarly, drawing from \cite[Lemma 5.2]{RefWorks:RefID:34-cai2018solving}, if $\epsilon\in(0, 0.01)$, there exist universal constants $c',\bar{c}'>0$ such that if $m\geq c'\epsilon^{-2}\log \epsilon^{-1}n$, then with probability at least $1-e^{-\bar{c}'\epsilon^2m}$, we obtain:
\[\begin{split}
&\left\{\left|y_k-\left|\bm{a}_k^* \bm{z}\right|^2\right| \leq 1.15\tau_2\|\bm{h}\|_2 \left(\left|\bm{a}_k^* \bm{z}\right|+\sqrt{y_k}\right)\right\}\subset\mathcal{E}_2^k(\bm{z})\subset\left\{\left|y_k-\left|\bm{a}_k^* \bm{z}\right|^2\right| \leq 3\tau_2\|\bm{h}\|_2 \left(\left|\bm{a}_k^* \bm{z}\right|+\sqrt{y_k}\right)\right\},\\
&\left\{\left|\bm{a}_k^* \bm{h}\right|\leq 1.15\tau_2\|\bm{h}\|_2\right\}\subset\left\{\left|y_k-\left|\bm{a}_k^* \bm{z}\right|^2\right| \leq 1.15\tau_2\|\bm{h}\|_2 \left(\left|\bm{a}_k^* \bm{z}\right|+\sqrt{y_k}\right)\right\}.
\end{split}
\]
    Thus, by combining these conditions and setting $C_{14}=\max\{2,c'\}$ and $\overline{C_{14}} = \min\{1,\bar{c}'\}$, if $m\geq C_{14}\epsilon^{-2}\log\epsilon^{-1}n$, then with probability at least $1-2e^{-\overline{C_{14}}\epsilon^2m}$, we establish:
\begin{equation}\label{event1}
\begin{aligned}
\mathbb{1}_{\mathcal{E}_0^k(\bm{y}) \cap \mathcal{E}_1^k(\bm{z}) \cap \mathcal{E}_2^k(\bm{z})} & =\mathbb{1}_{\mathcal{E}_1^k(\bm{z})}-\mathbb{1}_{\mathcal{E}_1^k(\bm{z})} \mathbb{1}_{\left(\mathcal{E}_2^k(\bm{z})\right)^c \cup\left(\mathcal{E}_0^k(\bm{y})\right)^c} \\
& \geq \mathbb{1}_{\mathcal{E}_1^k(\bm{z})}-\mathbb{1}_{\mathcal{E}_1^k(\bm{z})} \mathbb{1}_{\left(\mathcal{E}_2^k(\bm{z})\right)^c}-\mathbb{1}_{\mathcal{E}_1^k(\bm{z})} \mathbb{1}_{\left(\mathcal{E}_0^k(\bm{y})\right)^c} \\
& \geq \mathbb{1}_{\mathcal{E}_1^k(\bm{z})}-\mathbb{1}_{\mathcal{E}_1^k(\bm{z})} \mathbb{1}_{\left\{\left|\bm{a}_k^{*} \bm{h}\right|>1.15 \tau_2\|\bm{h}\|_2\right\}}-\mathbb{1}_{\mathcal{E}_1^k(\bm{z})} \mathbb{1}_{\left\{\left|\bm{a}_k^{*} \bm{x}\right|>0.9 \tau_0\|\bm{x}\|_2\right\}}.
\end{aligned}
\end{equation}
    This lower bound for the indicator function implies that:
\begin{equation}
\begin{split}
\frac{1}{m}\|\mathcal{A}_{\bm{Z}}(\bm{W})\|_2^2&=\frac{4}{m}\sum\limits_{k=1}^{m}\left(\operatorname{Re}\left(\bm{w}^* \bm{a}_k \bm{a}_k^* \bm{z}\right)\right)^2 \cdot \mathbb{1}_{\mathcal{E}_0^k(\bm{x}) \cap \mathcal{E}_1^k(\bm{z}) \cap \mathcal{E}_2^k(\bm{z})} \\&\geq \frac{4}{m}\sum\limits_{k=1}^{m}\left(\operatorname{Re}\left(\bm{w}^* \bm{a}_k \bm{a}_k^* \bm{z}\right)\right)^2 \cdot \mathbb{1}_{\{{\left|\bm{a}_k^{*} \bm{z}\right| \leq \tau_1\|\bm{z}\|_2}\}}\\&-\frac{4}{m}\sum\limits_{k=1}^{m}\left(\operatorname{Re}\left(\bm{w}^* \bm{a}_k \bm{a}_k^* \bm{z}\right)\right)^2 \cdot\mathbb{1}_{\{{\left|\bm{a}_k^{*} \bm{z}\right| \leq \tau_1\|\bm{z}\|_2}\}} \mathbb{1}_{\left\{\left|\bm{a}_k^{*} \bm{h}\right|>1.15 \tau_2\|\bm{h}\|_2\right\}}\\ &-\frac{4}{m}\sum\limits_{k=1}^{m}\left(\operatorname{Re}\left(\bm{w}^* \bm{a}_k \bm{a}_k^* \bm{z}\right)\right)^2 \cdot \mathbb{1}_{\{{\left|\bm{a}_k^{*} \bm{z}\right| \leq \tau_1\|\bm{z}\|_2}\}}\mathbb{1}_{\left\{\left|\bm{a}_k^{*} \bm{x}\right|>0.9 \tau_0\|\bm{x}\|_2\right\}} \\ & := 4\mathbf{I}_{21} -4\mathbf{I}_{22}-4\mathbf{I}_{23}.
\end{split}
\end{equation}

Next, we establish the lower bound for $\mathbf{I}_{21}$ and the upper bounds for $\mathbf{I}_{22}$ and $\mathbf{I}_{23}$ respectively.

For $\mathbf{I}_{21}$, by leveraging equations \eqref{I11} and \eqref{I12} from the upper bound analysis, we have:
\begin{equation}\nonumber
\begin{split}
\mathbf{I}_{21} &= \frac{1}{m}\sum\limits_{k=1}^{m}\left(\left(\operatorname{Re}\left(\bm{w}^* \bm{a}_k \bm{a}_k^* \bm{z}\right)\right)^2 \cdot \mathbb{1}_{\mathcal{E}_1^k(\bm{z})} - \mathbb{E}\left[\left(\operatorname{Re}\left(\bm{w}^* \bm{a}_k \bm{a}_k^* \bm{z}\right)\right)^2\cdot\mathbb{1}_{\mathcal{E}_1^k(\bm{z})} \right]\right)\\ &+\frac{1}{m}\sum\limits_{k=1}^{m} \mathbb{E}\left[\left(\operatorname{Re}\left(\bm{w}^* \bm{a}_k \bm{a}_k^* \bm{z}\right)\right)^2\cdot\mathbb{1}_{\mathcal{E}_1^k(\bm{z})} \right] \\&\geq -\rho_1 \|\bm{z}\|_2^2\|\bm{w}\|_2^2 + \frac{1}{2}\hat{\beta_2}\|\bm{z}\|_2^2\|\bm{w}\|_2^2 +(\hat{\beta_1}+\frac{\hat{\beta_2}}{2})|\bm{z^*}\bm{w}|^2 \\ &\geq  -\rho_1 \|\bm{z}\|_2^2 \|\bm{x}\|_2^2 +\frac{1}{2} \hat{\beta_1}\|\bm{z}\|_2^2 \|\bm{x}\|_2^2 + \frac{3}{2}\hat{\beta_1}|
\bm{z^*}\bm{w}|^2 \\& \geq \frac{1}{4}(\hat{\beta_1}-2\rho_1) \|\bm{W}\|_{\mathfrak{o}}. 
\end{split}
\end{equation}
The final inequality in this derivation holds due to $\hat{\beta_2}\geq \hat{\beta_1}$ and the definition of $\|\bm{W}\|_{\mathfrak{o}}$.

For $\mathbf{I}_{22}$, we note its expression as:
\begin{align*}
\mathbf{I}_{22} &=\frac{1}{m}\sum\limits_{k=1}^{m}\left(\operatorname{Re}\left(\bm{w}^* \bm{a}_k \bm{a}_k^* \bm{z}\right)\right)^2 \cdot \mathbb{1}_{\{{\left|\bm{a}_k^{*} \bm{z}\right| \leq \tau_1\|\bm{z}\|_2}\}} \mathbb{1}_{\left\{\left|\bm{a}_k^{*} \bm{h}\right|>1.15 \tau_2\|\bm{h}\|_2\right\}}  \\ & = \frac{1}{4m}\sum_{k=1}^m\left[\begin{array}{l}\bm{w} \\ \overline{\bm{w}}\end{array}\right]^*\left[\begin{array}{ll}\left|\bm{a}_k^* \bm{z}\right|^2 \bm{a}_k \bm{a}_k^* & \left(\bm{a}_k^* \bm{z}\right)^2 \bm{a}_k \bm{a}_k^{\top} \\ \left(\overline{\bm{a}_k^* \bm{z}}\right)^2 \overline{\bm{a}}_k \bm{a}_k^* & \left|\bm{a}_k^* \bm{z}\right|^2 \overline{\bm{a}}_k \bm{a}_k^{\top} \end{array}\right]\left[\begin{array}{l}\bm{w} \\ \overline{\bm{w}}\end{array}\right] \cdot \mathbb{1}_{\{{\left|\bm{a}_k^{*} \bm{z}\right| \leq \tau_1\|\bm{z}\|_2}\}} \mathbb{1}_{\left\{\left|\bm{a}_k^{*} \bm{h}\right|>1.15 \tau_2\|\bm{h}\|_2\right\}}\\ &\leq \frac{1}{2}\left\| \frac{1}{m}\sum_{k=1}^m\left|\bm{a}_k^* \bm{z}\right|^2 \bm{a}_k \bm{a}_k^{*} \cdot \mathbb{1}_{\{{\left|\bm{a}_k^{*} \bm{z}\right| \leq \tau_1\|\bm{z}\|_2}\}} \mathbb{1}_{\left\{\left|\bm{a}_k^{*} \bm{h}\right|>1.15 \tau_2\|\bm{h}\|_2\right\}}\right\|\|\bm{w}\|_2^2 \\ &+
    \frac{1}{2}\left\|\frac{1}{m}\sum_{k=1}^m\left(\bm{a}_k^*\bm{z}\right)^2\bm{a}_k \bm{a}_k^{*}\cdot \mathbb{1}_{\{{\left|\bm{a}_k^{*} \bm{z}\right| \leq \tau_1\|\bm{z}\|_2}\}} \mathbb{1}_{\left\{\left|\bm{a}_k^{*} \bm{h}\right|>1.15 \tau_2\|\bm{h}\|_2\right\}}\right\|\|\bm{w}\|_2^2 \\& \leq \frac{1}{2}\tau_{1}^2\left\| \frac{1}{m}\sum_{k=1}^m\bm{a}_k \bm{a}_k^{*} \cdot\mathbb{1}_{\left\{\left|\bm{a}_k^{*} \bm{h}\right|>1.15 \tau_2\|\bm{h}\|_2\right\}}\right\|\|\bm{z}\|_2^2\|\bm{w}\|_2^2 \\&+
\frac{1}{2}\tau_{1}^2\left\|\frac{1}{m}\sum_{k=1}^m\bm{a}_k \bm{a}_k^{*}\cdot\mathbb{1}_{\left\{\left|\bm{a}_k^{*} \bm{h}\right|>1.15 \tau_2\|\bm{h}\|_2\right\}}\right\|\|\bm{z}\|_2^2\|\bm{w}\|_2^2. 
\end{align*}
The last inequality in the above derivation follows from Lemma \ref{lem: bound of spectral norm of sum of a_ka_k^T}.
Subsequently, by applying Lemma \ref{spectral}, for some universal constants $C_{15},\overline{C_{15}}>0$, if $\epsilon>0$ is sufficiently small and $m \geq C_{15} \epsilon^{-2} \log \epsilon^{-1}n$, then with probability exceeding $1-e^{-\overline{C_{15}}m \epsilon^2}$, we have:
\begin{align*}
\mathbf{I}_{22}  &\leq \tau_{1}^2\left(5.75 \tau_2 e^{-0.64 \tau_2^2}+\epsilon\right)\|\bm{w}\|^2\|\bm{z}\|_2^2.
\end{align*}
Similarly, $\mathbf{I}_{23}$ can be bounded from above. For some universal constants $C_{15}',\overline{C_{15}}'>0$, if $\epsilon>0$ is sufficiently small and $m \geq C_{15}' \epsilon^{-2} \log \epsilon^{-1}n$, then with probability exceeding $1-e^{-\overline{C_{15}}'m \epsilon^2}$, we get:
$$
\mathbf{I}_{23} \leq  \tau_1^2\left(3.6 \tau_0 e^{-0.39 \tau_0^2}+\epsilon\right)\|\bm{w}\|_2^2\|\bm{z}\|_2^2,
$$
holds for all $\bm{z}\neq 0,\bm{w}\in\mathbb{C}^{n}$. 

Altogether, combining the bounds for $\mathbf{I}_{21}, \mathbf{I}_{22}$ and $\mathbf{I}_{23}$, and defining $C_{16}= \max\{C_{14},C_{15},C_{15}'\}$ and $\overline{C_{16}}=\min\{\overline{C_{14}},\overline{C_{15}},\overline{C_{15}}'\}$, when $m\geq C_{16}\epsilon^{-2} \log \epsilon^{-1}n$, then with probability at least $1-5e^{-\overline{C_{16}}m \epsilon^2}$, we obtain the overall lower bound:
\begin{equation}\label{lower bound}
\begin{split}
\frac{1}{m}\|\mathcal{A}_{\bm{Z}}(\bm{W})\|_2^2 &\geq  \hat{\beta_1}\|\bm{W}\|_{\mathfrak{o}} -2\left(\
\rho_1+ \tau_{1}^2\left(5.75 \tau_2 e^{-0.64 \tau_2^2}+3.6 \tau_0 e^{-0.39 \tau_0^2}+2\epsilon\right)\right)\|\bm{z}\|_2^2 \|\bm{x}\|_2^2 \\&\geq (\hat{\beta_1}- 2(\rho_1+\rho_2))\|\bm{W}\|_{\mathfrak{o}}=(\hat{\beta_1}- \delta)\|\bm{W}\|_{\mathfrak{o}},
\end{split}
\end{equation}
which holds for all $\bm{W}\in \mathbb{T}_{\bm{Z}}$. 
\end{itemize}

Combining the results from \eqref{upper bound} and \eqref{lower bound}, we establish that for some constants $C_{12},C_{13}>0$, if $m\geq C_{12}\epsilon^{-2} \log \epsilon^{-1}n$, then with probability at least $1-5e^{-C_{13}m \epsilon^2}$, the following inequality
\[
(\hat{\beta_1}-\delta)\|\bm{W}\|_{\mathfrak{o}}^2 \leq \frac{1}{m}\|\mathcal{A}_{\bm{Z}}(\bm{W})\|_2^2 \leq (\hat{\beta_2}+\delta)\|\bm{W}\|_{\mathfrak{o}}^2
\]
holds for all positive semi-definite $\bm{Z}\in\mathcal{M}_1$ satisfying $\left\|\bm{Z} -\bm{X} \right\|_F \leq \frac{1}{14}\|\bm{X}\|_F$ and any $\bm{W}\in\mathbb{T}_{\bm{Z}}$. This completes the proof of Proposition \ref{RIP:Gaussian}.

\end{proof}

From Proposition \ref{RIP:Gaussian}, it is evident that the condition number of $\frac{1}{\sqrt{m}}\mathcal{A}_{\bm{Z}}$ restricted to the tangent space $\mathbb{T}_{\bm{Z}}$ is $\kappa_{\mathfrak{o}} =\frac{\hat{\beta_2}+\delta}{\hat{\beta_1}-\delta}$. By increasing the values of $\tau_0,\tau_1,\tau_2$ and decreasing $\epsilon$, $\kappa_{\mathfrak{o}}$ approaches $1$. This behavior indicates that $\frac{1}{\sqrt{m}}\mathcal{A}_{\bm{Z}}$ is near-isometric when restricted to the tangent space $\mathbb{T}_{\bm{Z}}$.

\begin{pro}[Restricted Weak Correlation]
\label{WRC}
Let $\{\bm{a}_k\}_{k=1}^{m}$ follow the random Gaussian model in \eqref{eq:Gaussian}. Let $\tau_0, \tau_1, \tau_2 >0$ be truncation parameters and let $\epsilon \in (0,1)$ be sufficiently small. Then, there exist constants $C_{12}, C_{13} > 0$ such that whenever the sample size satisfies $m \geq C_{12} \epsilon^{-2}\log(1/\epsilon)n$, the following event occurs with probability at least $1-30e^{-C_{13} m\epsilon^2}$: For all positive semi-definite $\bm{Z}\in\mathcal{M}_1$ satisfying $\left\|\bm{Z} -\bm{X} \right\|_F \leq \epsilon_0\|\bm{X}\|_F$,
\begin{equation}\label{wrc1}
\frac{1}{m}\|\mathcal{T}_{\mathbb{T}_{\bm{Z}}}^{(\mathfrak{o})}\mathcal{A}_{\bm{Z}}^*\mathcal{A}_{\bm{Z}}(\mathcal{I}-\mathcal{P}_{\mathbb{T}_{\bm{Z}}})(\bm{Z}-\bm{X})\|_{\mathfrak{o}} \leq\sqrt{2(\hat{\beta_2}+\delta)\rho_4}\|\bm{Z}-\bm{X}\|_{F}.
\end{equation}
\end{pro}
\begin{proof} 
For any positive semi-definite $\bm{Z}\in\mathcal{M}_1$, we can write $\bm{Z}=\bm{zz}^*$ for some $\bm{z}\in\mathbb{C}^{n}$. Let $\bm{h}=\bm{z}-\bm{x}$. By Lemma \ref{lemma:metric2}, the following derivation holds:
$$\begin{aligned}
&\frac{1}{m}\|\mathcal{T}_{\mathbb{T}_{\bm{Z}}}^{(\mathfrak{o})}\mathcal{A}_{\bm{Z}}^*\mathcal{A}_{\bm{Z}}(\mathcal{I}-\mathcal{P}_{\mathbb{T}_{\bm{Z}}})(\bm{Z}-\bm{X})\|_{\mathfrak{o}
   }
   =\frac{1}{m}\|\mathcal{T}_{\mathbb{T}_{\bm{Z}}}^{(\mathfrak{o})}\mathcal{A}_{\bm{Z}}^*\mathcal{A}_{\bm{Z}}(\mathcal{I}-\mathcal{P}_{\mathbb{T}_{\bm{Z}}})(\bm{zz}^*-\bm{xx}^*)\|_{\mathfrak{o}
   }\\
&\leq\frac{1}{m}\|\mathcal{P}_{\mathbb{T}_{\bm{Z}}}\mathcal{A}_{\bm{Z}}^*\mathcal{A}_{\bm{Z}}(\mathcal{I}-\mathcal{P}_{\mathbb{T}_{\bm{Z}}})(\bm{zz}^*-\bm{xx}^*)\|_F
\leq \frac{1}{m}\|\mathcal{P}_{\mathbb{T}_{\bm{Z}}}\mathcal{A}_{\bm{Z}}^*\|\|\mathcal{A}_{\bm{Z}}(\mathcal{I}-\mathcal{P}_{\mathbb{T}_{\bm{Z}}})(\bm{zz}^*-\bm{xx}^*)\|_F.
\end{aligned}$$
Furthermore, from the proof of Proposition \ref{RIP:Gaussian}, it follows that for chosen parameters $\tau_0, \tau_1, \tau_2>0$ and a sufficiently small $\epsilon > 0$, there exist constants $C_{12}, C_{13} > 0$. If $m \geq C_{12} \epsilon^{-2}\log{\epsilon^{-1}}n$, then with probability at least $1-5e^{-C_{13}m\epsilon^2}$, the following conditions hold:
\begin{equation}\label{rip}
(\hat{\beta_1}-\delta)\|\bm{W}\|_{\mathfrak{o}}^2 \leq \frac{1}{m}\|\mathcal{A}_{\bm{Z}}(\bm{W})\|_2^2 \leq (\hat{\beta_2}+\delta)\|\bm{W}\|_{\mathfrak{o}}^2,
\qquad \forall~\bm{W}\in\mathbb{T}_{\bm{Z}},
\end{equation} 
\begin{equation}\label{event_h}
\frac{1}{m} \sum_{k=1}^m \left|\bm{a}_k^{*} \bm{h}\right|^2\leq1.05\|\bm{h}\|_2^2, \quad \forall \bm{h}\in\mathbb{C}^{n},
\end{equation}
and
\begin{equation}\label{small_z}
\left\|\frac{1}{m}\sum_{k=1}^m \bm{a}_k \bm{a}_k^{*} \cdot\mathbb{1}_{\left\{\left|\bm{a}_k^{*} \bm{z}\right|>1.15\tau_2\|\bm{z}\|_2\right\}}\right\| \leq 6\tau_2 e^{-0.64 \tau_2^2}+\epsilon. 
\end{equation}

Given the event \eqref{rip}, we can establish the following upper bound:
\begin{equation}\label{eq:prop4}
\begin{aligned}
\frac{1}{\sqrt{m}}\left\|\mathcal{P}_{\mathbb{T}_{\bm{Z}}} \mathcal{A}_{\bm{Z}}^{*}\right\| & =\sup _{\|\bm{b}\|_2=1,\|\bm{W}\|_F=1} \left|\left\langle\bm{W}, \frac{1}{\sqrt{m}}\mathcal{P}_{\mathbb{T}_{\bm{Z}}} \mathcal{A}_{\bm{Z}}^{*} (\bm{b})\right\rangle\right|  =\sup _{\|\bm{b}\|_2=1,\|\bm{W}\|_F=1}\left|\left\langle\frac{1}{\sqrt{m}}\mathcal{A}_{\bm{Z}} \mathcal{P}_{\mathbb{T}_{\bm{Z}}}(\bm{W}), \bm{b}\right\rangle\right| \\
&\leq\sqrt{(\hat{\beta_2}+\delta)}\|\mathcal{P}_{\mathbb{T}_{\bm{Z}}}(\bm{W})\|_{\mathfrak{o}}\|\bm{b}\|_2
\leq\sqrt{2(\hat{\beta_2}+\delta)}\|\mathcal{P}_{\mathbb{T}_{\bm{Z}}}(\bm{W})\|_{F}\|\bm{b}\|_2\\
& \leq \sqrt{2(\hat{\beta_2}+\delta)}\|\bm{W}\|_{F}\|\bm{b}\|_2\leq\sqrt{2(\hat{\beta_2}+\delta)}.
\end{aligned}
\end{equation}
The second inequality in the above derivation arises from the fact that $\|\mathcal{P}_{\mathbb{T}_{\bm{Z}}}(\bm{W})\|_{\mathfrak{o}}\leq \sqrt{2}\|\mathcal{P}_{\mathbb{T}_{\bm{Z}}}(\bm{W})\|_F$, which is a consequence of \eqref{new metric}.

Then, to prove \eqref{wrc1}, it remains to establish $\frac{1}{m}\|\mathcal{A}_{\bm{Z}}(\mathcal{I}-\mathcal{P}_{\mathbb{T}_{\bm{Z}}})(\bm{zz}^*-\bm{xx}^*)\|_F\leq \sqrt{\rho_4}\|\bm{Z}-\bm{X}\|_{F}$. To this end, observe that
\begin{equation}\nonumber
\begin{split}
&(\mathcal{I}-\mathcal{P}_{\mathbb{T}_{\bm{Z}}})(\bm{zz}^*-\bm{xx}^*)=-\Big(\frac{\bm{z}^*\bm{h}}{\|\bm{z}\|^2_{2}}\bm{z}-\bm{h}\Big)\Big(\frac{\bm{z}^*\bm{h}}{\|\bm{z}\|^2_{2}}\bm{z}-\bm{h}\Big)^*.\\
\end{split}
\end{equation}
It follows that
\begin{equation}\label{prop3:333}
\begin{aligned} 
&\frac{1}{m}\|\mathcal{A}_{\bm{Z}}(\mathcal{I}-\mathcal{P}_{\mathbb{T}_{\bm{Z}}})(\bm{zz}^*-\bm{xx}^*)\|_F^2\\
&=\frac{1}{m}\sum_{k=1}^m \left|\bm{a}_k^{*}\left(\frac{\bm{z}^{*} \bm{h}}{\|\bm{z}\|_2^2} \bm{z}-\bm{h}\right)\right|^4 \mathbb{1}_{\mathcal{E}_0^k(\bm{y}) \cap \mathcal{E}_1^k(\bm{z}) \cap \mathcal{E}_2^k(\bm{z})} \\ & \leq \frac{1}{m}\sum_{k=1}^m \frac{\left|\bm{a}_k^{*} \bm{z}\right|^4\left|\bm{z}^{*} \bm{h}\right|^4}{\|\bm{z}\|_2^8} \mathbb{1}_{\mathcal{E}_0^k(\bm{y}) \cap \mathcal{E}_1^k(\bm{z}) \cap \mathcal{E}_2^k(\bm{z})} \\
& +\frac{4}{m}\sum_{k=1}^m \frac{\left|\bm{z}^{*} \bm{h}\right|^3}{\|\bm{z}\|_2^6}\left|\bm{a}_k^{*} \bm{z}\right|^3\left|\bm{a}_k^{*} \bm{h}\right| \mathbb{1}_{\mathcal{E}_0^k(\bm{y}) \cap \mathcal{E}_1^k(\bm{z}) \cap \mathcal{E}_2^k(\bm{z})} \\ & +\frac{6}{m}\sum_{k=1}^m \frac{\left|\bm{z}^{*} \bm{h}\right|^2}{\|\bm{z}\|_2^4}\left|\bm{a}_k^{*} \bm{z}\right|^2\left|\bm{a}_k^{*} \bm{h}\right|^2 \mathbb{1}_{\mathcal{E}_0^k(\bm{y}) \cap \mathcal{E}_1^k(\bm{z}) \cap \mathcal{E}_2^k(\bm{z})} \\ & +\frac{4}{m} \sum_{k=1}^m \frac{\left|\bm{z}^{*} \bm{h}\right|}{\|\bm{z}\|_2^2}\left|\bm{a}_k^{*} \bm{z}\right|\left|\bm{a}_k^{*} \bm{h}\right|^3 \mathbb{1}_{\mathcal{E}_0^k(\bm{y}) \cap \mathcal{E}_1^k(\bm{z}) \cap \mathcal{E}_2^k(\bm{z})}\\&+ \frac{1}{m}\sum_{k=1}^m\left|\bm{a}_k^{*} \bm{h}\right|^4 \mathbb{1}_{\mathcal{E}_0^k(\bm{y}) \cap \mathcal{E}_1^k(\bm{z}) \cap \mathcal{E}_2^k(\bm{z})}\\& 
:=\mathbf{I}_1+\mathbf{I}_2+\mathbf{I}_3+\mathbf{I}_4+\mathbf{I}_5 .
\end{aligned}
\end{equation}

We now proceed to bound each term $\mathbf{I}_i$ for $i=1, \dots, 5$, utilizing \eqref{event_h} and \eqref{small_z}.

\paragraph{Upper bound of $\mathbf{I}_1$.} Direct calculation yields:
$$\mathbf{I}_1 \leq \frac{1}{m}\sum_{k=1}^m \frac{|\bm{a}_k^{*} \bm{z}|^4|\bm{z}^{*} \bm{h}|^4}{\|\bm{z}\|_2^8} \mathbb{1}_{\mathcal{E}_1^k(\bm{z}) }\leq \tau_1^4\|\bm{h}\|_2^4.$$

\paragraph{Upper bound of $\mathbf{I_2}$.} The term $\mathbf{I_2}$ can be bounded as:
\begin{align*}
    \mathbf{I_2}
    & \leq \frac{4}{m} \sum_{k=1}^m \frac{\left|\bm{z}^{*} \bm{h}\right|^3}{\|\bm{z}\|_2^6}\left|\bm{a}_k^{*} \bm{z}\right|^3\left|\bm{a}_k^{*} \bm{h}\right| \mathbb{1}_{\mathcal{E}_1^k(\bm{z})} 
     \leq \frac{4}{m} \tau_1^3\|\bm{h}\|_2^3\sum_{k=1}^m\left|\bm{a}_k^{*} \bm{h}\right| \\ & \leq 4 \tau_1^3\|\bm{h}\|_2^3 \sqrt{ \frac{1}{m}\sum_{k=1}^m\left|\bm{a}_k^{*} \bm{h}\right|^2}  \leq 4\sqrt{(1+\delta_1)} \tau_1^3\|\bm{h}\|_2^4\leq 5 \tau_1^3\|\bm{h}\|_2^4,
\end{align*}
where the last inequality follows from \eqref{event_h}.

\paragraph{Upper bound of $\mathbf{I_3}$.} Similarly, $\mathbf{I_3}$ is bounded by:
$$\begin{aligned} \mathbf{I}_3 & \leq \frac{6}{m}\sum_{k=1}^m \frac{\left|\bm{z}^{*} \bm{h}\right|^2}{\|\bm{z}\|_2^4}\left|\bm{a}_k^{*} \bm{z}\right|^2\left|\bm{a}_k^{*} \bm{h}\right|^2 \mathbb{1}_{ \mathcal{E}_1^k(\bm{z}) }  \leq \frac{6}{m} \tau_1^2\|\bm{h}\|_2^2 \sum_{k=1}^m\left|\bm{a}_k^{*} \bm{h}\right|^2  \leq 8 \tau_1^2\|\bm{h}\|_2^4 ,\end{aligned}$$
where the last inequality follows from \eqref{event_h}.
\paragraph{Upper bound of $\mathbf{I}_4$.} Under event \eqref{event_h}, and following calculations analogous to those in \cite[Lemma 5.2]{RefWorks:RefID:34-cai2018solving}, we obtain:
\begin{equation}\label{usefultruncation}
    \mathbb{1}_{\mathcal{E}_0^k(\bm{y}) \cap \mathcal{E}_1^k(\bm{z}) \cap \mathcal{E}_2^k(\bm{z})} \leq \mathbb{1}_{\left\{\left|\bm{a}_k^{*} \bm{h}\right| \leq \tau_{h, z}\|\bm{z}\|_2\right\}},\end{equation}
where $\tau_{h,z}: =\tau_1+\left(0.3 \tau_2\left(\tau_1+1.2 \tau_0\right)+\tau_1^2\right)^{1 / 2}$. Consequently, this leads to:

$$
\begin{aligned}
\mathbf{I}_4 
& \leq \frac{4}{m}\tau_1\|\bm{h}\|_2\left(\sum_{k=1}^m\left|\bm{a}_k^{*} \bm{h}\right|^3 \mathbb{1}_{\mathcal{E}_0^k(\bm{y}) \cap \mathcal{E}_1^k(\bm{z}) \cap \mathcal{E}_2^k(\bm{z})}\right) 
\leq \frac{4}{m} \tau_1\|\bm{h}\|_2\left(\sum_{k=1}^m\left|\bm{a}_k^{*} \bm{h}\right|^3 \mathbb{1}_{\left\{\left|\bm{a}_k^{*} \bm{h}\right| \leq \tau_{h, z}\|\bm{z}\|_2\right\}}\right) \\& \leq \frac{4}{m}\tau_1\tau_{h,z}\|\bm{h}\|_2\|\bm{z}\|_2\sum_{k=1}^m\left|\bm{a}_k^{*} \bm{h}\right|^2 
\leq 5 \tau_1\tau_{h,z}\|\bm{h}\|_2^3\|\bm{z}\|_2,
\end{aligned}
$$
where the last inequality follows from \eqref{event_h}.

\paragraph{Upper bound of $\mathbf{I}_5$}\quad For $\mathbf{I}_5$, by again applying \eqref{usefultruncation}, we find:
\begin{align*} \mathbf{I}_5 & \leq\frac{1}{m}\sum_{k=1}^m\left|\bm{a}_k^{*} \bm{h}\right|^4 \mathbb{1}_{\left\{\left|\bm{a}_k^{*} \bm{h}\right| \leq \tau_{h, z}\|\bm{z}\|_2\right\}} \\ & =\frac{1}{m}\sum_{k=1}^m\left|\bm{a}_k^{*} \bm{h}\right|^4 \mathbb{1}_{\left\{\left\{\left|\bm{a}_k^{*} \bm{h}\right| \leq \tau_{h, z}\|\bm{z}\|_2\right\} \cap \left\{\left|\bm{a}_k^{*} \bm{h}\right| \leq 1.15 \tau_2\|\bm{h}\|_2\right\}\right\}}+\frac{1}{m}\sum_{k=1}^m\left|\bm{a}_k^{*} \bm{h}\right|^4 \mathbb{1}_{\left\{\left\{\left|\bm{a}_k^{*} \bm{h}\right| \leq \tau_{h, z}\|\bm{z}\|_2\right\} \cap\left\{\left|\bm{a}_k^{*} \bm{h}\right| \leq 1.15 \tau_2\|\bm{h}\|_2\right\}^c\right\}} \\ & \leq \frac{1}{m}\sum_{k=1}^m\left|\bm{a}_k^{*} \bm{h}\right|^4 \mathbb{1}_{\left\{\left|\bm{a}_k^{*} \bm{h}\right| \leq 1.15 \tau_2\|\bm{h}\|_2\right\}}+\frac{1}{m}\sum_{k=1}^m\left|\bm{a}_k^{*} \bm{h}\right|^4 \mathbb{1}_{\left\{\left\{\left|\bm{a}_k^{*} \bm{h}\right| \leq \tau_{h, z}\|\bm{z}\|_2\right\} \cap\left\{\left|\bm{a}_k^{*} \bm{h}\right|>1.15 \tau_2\|\bm{h}\|_2\right\}\right\}} \\ & \leq 1.15^2\|\bm{h}\|_2^2 \tau_2^2\left(\frac{1}{m} \sum_{k=1}^m\left|\bm{a}_k^{*} \bm{h}\right|^2\right)+\tau_{h, z}^2\|\bm{z}\|_2^2 \bm{h}^{*}\cdot\left(\frac{1}{m} \sum_{k=1}^m \bm{a}_k \bm{a}_k^{*} \mathbb{1}_{\left\{\left|\bm{a}_k^{*} \bm{h}\right|>1.15 \tau_2\|\bm{h}\|_2\right\}}\right) \cdot\bm{h} \\ &\leq 2\tau_2^2\|\bm{h}\|_2^4+\left(6 \tau_2 \tau_{h, z}^2 e^{-0.64 \tau_2^2}+\tau_{h, z}^2\epsilon\right)\|\bm{z}\|_2^2\|\bm{h}\|_2^2,
\end{align*}
where the last inequality follows from \eqref{event_h} and \eqref{small_z}.

Finally, by combining all the derived bounds, we obtain:
\begin{equation}\label{square1}
\begin{aligned}
&\frac{1}{m}\|\mathcal{A}_{\bm{Z}}(\mathcal{I}-\mathcal{P}_{\mathbb{T}_{\bm{Z}}})(\bm{zz}^*-\bm{xx}^*)\|_F^2 \\&\leq \left((\tau_1^4+5\tau_1^3+8\tau_1^2+2\tau_2^2)\frac{\|\bm{h}\|_2^2}{\|\bm{z}\|_2^2}+5 \tau_1 \tau_{h, z} \frac{\|\bm{h}\|_2}{\|\bm{z}\|_2}+\left(6 \tau_2 \tau_{h, z}^2 e^{-0.64 \tau_2^2}+\tau_{h, z}^2\epsilon\right)\right)\|\bm{z}\|_2^2\|\bm{h}\|_2^2\\&\leq \frac{5}{4} \left((\tau_1^4+5\tau_1^3+8\tau_1^2+2\tau_2^2)\frac{\|\bm{h}\|_2^2}{\|\bm{z}\|_2^2}+5 \tau_1 \tau_{h, z} \frac{\|\bm{h}\|_2}{\|\bm{z}\|_2}+\left(6 \tau_2 \tau_{h, z}^2 e^{-0.64 \tau_2^2}+\tau_{h, z}^2\epsilon\right)\right)\|\bm{zz^*}-\bm{xx^*}\|_F^2
\end{aligned}
\end{equation}
The second inequality in \eqref{square1} follows from Lemma \ref{uuu11}. Given that $\|\bm{zz}^*-\bm{xx}^*\|_F\leq \epsilon_0\|\bm{X}\|_F$, and applying Lemma \ref{uuu11} once more, we have $\frac{\|\bm{h}\|_2}{\|\bm{z}\|_2} \leq \frac{2}{\sqrt{5}}\epsilon_0$. Substituting this inequality into \eqref{square1}, we get:
\begin{align*}
&\frac{1}{m}\|\mathcal{A}_{\bm{Z}}(\mathcal{I}-\mathcal{P}_{\mathbb{T}_{\bm{Z}}})(\bm{zz}^*-\bm{xx}^*)\|_F^2 \\
&\leq \frac{5}{4}\left(\frac{4}{5}(\tau_1^4+5\tau_1^3+8\tau_1^2+2\tau_2^2)\epsilon_0^2+\frac{2}{\sqrt{5}}\cdot5 \tau_1 \tau_{h, z} \epsilon_0+\left(6 \tau_2 \tau_{h, z}^2 e^{-0.64 \tau_2^2}+\tau_{h, z}^2\epsilon\right)\right)\|\bm{zz^*}-\bm{xx^*}\|_F^2 \\
&\leq \rho_4\|\bm{zz^*}-\bm{xx^*}\|_F^2
\end{align*}
Combining this result with \eqref{eq:prop4} concludes the proof.

\end{proof}

\subsection{Proof of the Theorem \ref{thm:3}}\label{sec5.2}
Based on the Proposition \ref{RIP:Gaussian} and \ref{WRC}, we shall prove Theorem \ref{thm:3} to establish the local linear convergence for TWRGD under random complex Gaussian measurements.

\begin{proof}[Proof of Theorem \ref{thm:3}]
Noting that $\mathcal{A}_t^{*}(\bm{y})=\mathcal{A}_t^{*} \mathcal{A}_t(\bm{X})$ and substituting the update rule $\bm{W}_{t}=\bm{Z}_{t}+\frac{\alpha_t}{m} \mathcal{T}_{\mathbb{T}_{t}}^{(\mathfrak{o})} \mathcal{A}_t^{*}\mathcal{A}_t\left(\bm{X}-\bm{Z}_t\right)$, we can expand the error term as follows:
\begin{equation}\label{convergence1}
    \begin{aligned}
    \left\|\bm{W}_t-\bm{X}\right\|_F&= \left\|\left(\bm{Z}_t-\bm{X}\right)-\frac{\alpha_t}{m}\mathcal{T}_{\mathbb{T}_{t}}^{(\mathfrak{o})} \mathcal{A}_t^{*} \mathcal{A}_t\left(\bm{Z}_t-\bm{X}\right)\right\|_{F} \\ 
&\leq  \left\|\left(\mathcal{P}_{\mathbb{T}_{t}}-\frac{\alpha_t}{m} {\mathcal{T}_{\mathbb{T}_{\bm{Z}}}^{(\mathfrak{o})}} \mathcal{A}_t^{*} \mathcal{A}_t \mathcal{P}_{\mathbb{T}_{t}}\right) \left(\bm{Z}_t-\bm{X}\right)\right\|_{\mathfrak{o}}+\left\|(\mathcal{I}-\mathcal{P}_{\mathbb{T}_{t}})  (\bm{Z}_t-\bm{X})\right\|_{F} \\ 
& +\frac{\alpha_t}{m}\left\|\mathcal{T}_{\mathbb{T}_{t}}^{(\mathfrak{o})} \mathcal{A}_t^{*} \mathcal{A}_t(\mathcal{I}-\mathcal{P}_{\mathbb{T}_{t}})(\bm{Z}_t-\bm{X})\right\|_{\mathfrak{o}}.
    \end{aligned}
\end{equation}
To analyze the first term on the right-hand side, we first establish a self-adjointness property. By Lemma \ref{lemma:metric1}, for any $\bm{A},\bm{B}\in\mathbb{T}_{\bm{Z}}$, we obtain
\[\begin{split}
\langle(\mathcal{P}_{\mathbb{T}_{\bm{Z}}}-\frac{\alpha_t}{m}{\mathcal{T}_{\mathbb{T}_{\bm{Z}}}^{(\mathfrak{o})}} \mathcal{A}_{\bm{Z}}^{*} \mathcal{A}_{\bm{Z}} \mathcal{P}_{\mathbb{T}_{\bm{Z}}})(\bm{A}), \bm{B}\rangle_{\mathfrak{o}}
&=\langle\mathcal{P}_{\mathbb{T}_{\bm{Z}}}(\bm{A}), \bm{B}\rangle_{\mathfrak{o}}-\frac{\alpha_t}{m}\langle\mathcal{P}_{\mathbb{T}_{\bm{Z}}} \mathcal{A}_{\bm{Z}}^{*} \mathcal{A}_{\bm{Z}} \mathcal{P}_{\mathbb{T}_{\bm{Z}}}(\bm{A}), \bm{B}\rangle\\
&=\langle\mathcal{P}_{\mathbb{T}_{\bm{Z}}}(\bm{A}), \mathcal{P}_{\mathbb{T}_{\bm{Z}}}(\bm{B})\rangle_{\mathfrak{o}}-\frac{\alpha_t}{m}\langle\mathcal{P}_{\mathbb{T}_{\bm{Z}}} \mathcal{A}_{\bm{Z}}^{*} \mathcal{A}_{\bm{Z}} \mathcal{P}_{\mathbb{T}_{\bm{Z}}}(\bm{A}), \mathcal{P}_{\mathbb{T}_{\bm{Z}}}(\bm{B})\rangle\\
&=\langle\bm{A},(\mathcal{P}_{\mathbb{T}_{\bm{Z}}}-\frac{\alpha_t}{m}{\mathcal{T}_{\mathbb{T}_{\bm{Z}}}^{(\mathfrak{o})}} \mathcal{A}_{\bm{Z}}^{*} \mathcal{A}_{\bm{Z}} \mathcal{P}_{\mathbb{T}_{\bm{Z}}})(\bm{B})\rangle_{\mathfrak{o}}.
\end{split}\]
This implies that the operator $\mathcal{P}_{\mathbb{T}_{\bm{Z}}}-\frac{\alpha_t}{m}{\mathcal{T}_{\mathbb{T}_{\bm{Z}}}^{(\mathfrak{o})}} \mathcal{A}_{\bm{Z}}^{*} \mathcal{A}_{\bm{Z}} \mathcal{P}_{\mathbb{T}_{\bm{Z}}}$ is self-adjoint in $\mathbb{T}_{\bm{Z}}$. Consequently, we can relate the operator norm to the spectral radius. Applying Proposition \ref{RIP:Gaussian}, we have
\begin{equation}
\begin{split}
\mathop{\sup}_{\bm{W}\in \mathbb{T}_{\bm{Z}}}\frac{\left\|\big(\mathcal{P}_{\mathbb{T}_{t}}-\frac{\alpha_t}{m} {\mathcal{T}_{\mathbb{T}_{\bm{Z}}}^{(\mathfrak{o})}} \mathcal{A}_t^{*} \mathcal{A}_t \mathcal{P}_{\mathbb{T}_{t}}\right) \left(\bm{W}\big)\right\|_{\mathfrak{o}}}{\|\bm{W}\|_{\mathfrak{o}}}
&=\mathop{\sup}_{\bm{W}\in \mathbb{T}_{\bm{Z}}}\frac{\left|\left\langle\left(\mathcal{P}_{\mathbb{T}_{t}}-\frac{\alpha_t}{m} {\mathcal{T}_{\mathbb{T}_{\bm{Z}}}^{(\mathfrak{o})}} \mathcal{A}_t^{*} \mathcal{A}_t \mathcal{P}_{\mathbb{T}_{t}}\right) \left(\bm{W}\right),\bm{W}\right\rangle_{\mathfrak{o}}\right|}{\|\bm{W}\|_{\mathfrak{o}}^2}\\
&=\mathop{\sup}_{\bm{W}\in \mathbb{T}_{\bm{Z}}}\frac{\left|\|\bm{W}\|_{\mathfrak{o}}^2-\frac{\alpha_t}{m}\left\langle\mathcal{T}_{\mathbb{T}_{\bm{Z}}}^{(\mathfrak{o})} \mathcal{A}_{\bm{Z}}^{*} \mathcal{A}_{\bm{Z}} \mathcal{P}_{\mathbb{T}_{\bm{Z}}}(\bm{W}), \bm{W}\right\rangle_{\mathfrak{o}}\right|}{\|\bm{W}\|_{\mathfrak{o}}^2 }\\
&=\mathop{\sup}_{\bm{W}\in \mathbb{T}_{\bm{Z}}}\frac{\left|\|\bm{W}\|_{\mathfrak{o}}^2-\frac{\alpha_t}{m}\left\|\mathcal{A}_{\bm{Z}}(\bm{W})\right\|_2^2\right|}{\|\bm{W}\|_{\mathfrak{o}}^2 }\\
 & \leq \max \big\{|1-\alpha_t (\hat{\beta_1}-\delta )|,|1-\alpha_t (\hat{\beta_2}+\delta)|\big\}.
\end{split}
\end{equation}
Therefore, the first term on the right-hand side of \eqref{convergence1} is bounded by 
\[
\begin{split}
\left\|\left(\mathcal{P}_{\mathbb{T}_{t}}-\frac{\alpha_t}{m}\mathcal{T}_{\mathbb{T}_{t}}^{(\mathfrak{o})} \mathcal{A}_t^{*} \mathcal{A}_t\mathcal{P}_{\mathbb{T}_{t}}\right)\left(\bm{Z}_t-\bm{X}\right)\right\|_{\mathfrak{o}}
&\leq \max \{|1-\alpha_t (\hat{\beta_2}+\delta )|,|1-\alpha_t (\hat{\beta_1}-\delta)|\}\|\mathcal{P}_{\mathbb{T}_{t}}(\bm{Z}_t-\bm{X})\|_{\mathfrak{o}}\\
&\leq \sqrt{2}\cdot\max \{|1-\alpha_t (\hat{\beta_2}+\delta )|,|1-\alpha_t (\hat{\beta_1}-\delta)|\}\|\mathcal{P}_{\mathbb{T}_{t}}(\bm{Z}_t-\bm{X})\|_F.
\end{split}
\]
Applying Proposition \ref{WRC}, we bound the second term on the right-hand side of \eqref{convergence1}: 
$$\alpha_t\left\|\frac{1}{m}\mathcal{P}_{\mathbb{T}_{t}}^{(\mathfrak{o})} \mathcal{A}_t^{*} \mathcal{A}_t(\mathcal{I}-\mathcal{P}_{\mathbb{T}_{t}})(\bm{Z}_t-\bm{X})\right\|_{\mathfrak{o}} \leq \alpha_t\sqrt{2\rho_4(\hat{\beta_2}+\delta)}\|\bm{Z}_t-\bm{X}\|_F.$$
Additionally, relying on [Lemma 4.1, \cite{RefWorks:RefID:31-wei2016guarantees}], we bound the third term on the right-hand side of \eqref{convergence1}: 
$$\left\|\left(\mathcal{I}-\mathcal{P}_{\mathbb{T}_{t}}\right) \bm{X}\right\|_F \leq \frac{\left\|\bm{Z}_t-\bm{X}\right\|_F^2}{\|\bm{X}\|_F} \leq \epsilon_0\left\|\bm{Z}_t-\bm{X}\right\|_F.$$
Substituting these three bounds into \eqref{convergence1} yields
    $$
     \left\|\bm{W}_t-\bm{X}\right\|_F \leq (\mu_1+\epsilon_0)\left\|\bm{Z}_t-\bm{X}\right\|_F   $$
    where $\mu_1:=\sqrt{2}\Big(\max\{|1-\alpha_t (\hat{\beta_1}-\delta )|,|1-\alpha_t (\hat{\beta_2}+\delta)|\}+\alpha_t \sqrt{\rho_4(\hat{\beta_2}+\delta)}\Big)$. Furthermore, it follows from [Lemma 4.4, \cite{RefWorks:RefID:34-cai2018solving}] that
$$
\left\|\bm{Z}_{t+1}-\bm{X}\right\|_F \leq (\mu_1+\epsilon_0) \sqrt{1+16 (\mu_1+\epsilon_0)^2 \epsilon_0^2}\left\|\bm{Z}_t-\bm{X}\right\|_F:=\nu_1\left\|\bm{Z}_t-\bm{X}\right\|_F.
$$

We must now demonstrate that the convergence rate satisfies $\nu_1<1$. This condition holds if
\begin{equation}\label{con_rate_1} 
\mu_1 \leq \varepsilon_{up}:=\frac{1-\epsilon_0 \sqrt{1+16 \epsilon_0^2}}{ \sqrt{1+16 \epsilon_0^2}}.  
\end{equation} 
To analyze this, note that \begin{equation}\nonumber
\max \big\{|1-\alpha_t(\hat{\beta_2}+\delta)|,|1-\alpha_t (\hat{\beta_1}-\delta)|\big\}=\Bigg\{\begin{array}{ll}\alpha_t (\hat{\beta_2}+\delta)-1 & \text { if } \alpha_t \geq\frac{2}{\hat{\beta_1}+\hat{\beta_2}}  \\ 1-\alpha_t (\hat{\beta_1}-\delta) & \text { if } \alpha_t \leq\frac{2}{\hat{\beta_1}+\hat{\beta_2}}\end{array}.
\end{equation}
In the case where $\alpha_t\geq\frac{2}{\hat{\beta_1}+\hat{\beta_2}}$, the parameter $\mu_1$ becomes
$\mu_1 
=\sqrt{2}\Big(\alpha_t(\hat{\beta_2}+\delta)-1+ \alpha_t\sqrt{\rho_4(\hat{\beta_2}+\delta)}\Big).
$
Consequently, \eqref{con_rate_1} is satisfied provided that
$$
\alpha_t\leq\frac{\varepsilon_{up}+\sqrt{2}}{\sqrt{2}\Big((\hat{\beta_2}+\delta)+\sqrt{\rho_4(\hat{\beta_2}+\delta)}\Big)}.
$$
Conversely, if $\alpha_t\leq\frac{2}{\hat{\beta_1}+\hat{\beta_2}}$, we have
$
\mu_1=\sqrt{2}\Big(1-\alpha_t (\hat{\beta_1}-\delta)+ \alpha_t\sqrt{\rho_4(\hat{\beta_2}+\delta)}\Big).
$
In this scenario, \eqref{con_rate_1} holds if
$$
\alpha_t\geq\frac{\sqrt{2}-\varepsilon_{up}}{\sqrt{2}\Big((\hat{\beta_1}-\delta)-\sqrt{\rho_4(\hat{\beta_2}+\delta)}\Big)}.
$$
Combining these cases, we conclude that choosing $\alpha_t \in \Big[\frac{\sqrt{2}-\varepsilon_{up}}{\sqrt{2}\big((\hat{\beta_1}-\delta)-\sqrt{\rho_4(\hat{\beta_2}+\delta)}\big)}, \frac{\varepsilon_{up}+\sqrt{2}}{\sqrt{2}\big((\hat{\beta_2}+\delta)+\sqrt{\rho_4(\hat{\beta_2}+\delta)}\big)}\Big]$ guarantees \eqref{con_rate_1}.

On the other hand, given the step size definition $\frac{\alpha_t}{m}= \frac{\|\mathcal{T}^{(\mathfrak{o})}_{\mathbb{T}_{t}}(\bm{G}_t)\|_{\mathfrak{o}}^2}{\|\mathcal{A}_t\mathcal{T}^{(\mathfrak{o})}_{\mathbb{T}_{t}}(\bm{G}_t)\|_2^2},$
Proposition \ref{RIP:Gaussian} implies the bounds:
\[
\begin{split}
\frac{1}{\hat{\beta_2}+\delta}=\frac{\|\mathcal{T}^{(\mathfrak{o})}_{\mathbb{T}_{t}}(\bm{G}_t)\|_{\mathfrak{o}}^2}{(\hat{\beta_2}+\delta)\|\mathcal{T}^{(\mathfrak{o})}_{\mathbb{T}_{t}}(\bm{G}_t)\|_{\mathfrak{o}}^2}\leq \alpha_t
\leq\frac{\|\mathcal{T}^{(\mathfrak{o})}_{\mathbb{T}_{t}}(\bm{G}_t)\|_{\mathfrak{o}}^2}{(\hat{\beta_1}-\delta)\|\mathcal{T}^{(\mathfrak{o})}_{\mathbb{T}_{t}}(\bm{G}_t)\|_{\mathfrak{o}}^2}=\frac{1}{\hat{\beta_1}-\delta}.
\end{split}
\]
Therefore, we can select a sufficiently small $\rho_3>0$, $\epsilon_0 >0$ and sufficiently large truncation parameters $\tau_0,\tau_1,\tau_2 >0$ such that the following inequalities hold:
$$
(\hat{\beta_1}-\delta)> \sqrt{\rho_4(\hat{\beta_2}+\delta)},
$$
$$\hat{\beta_2}+\delta+ \sqrt{\rho_4(\hat{\beta_2}+\delta) } \leq \frac{\sqrt{2}+\varepsilon_\mathrm{up}}{\sqrt{2}}(\hat{\beta_1}-\delta),$$ 
and
$$
\hat{\beta_1}-\delta+ \sqrt{\rho_4(\hat{\beta_2}+\delta) } \geq \frac{\sqrt{2}-\varepsilon_\mathrm{up}}{\sqrt{2}}(\hat{\beta_2}+\delta).
$$
These selections ensure the inclusion
\begin{equation}\label{stepsize1}
 \alpha_t \in\left[\frac{1}{\hat{\beta_2}+\delta},\frac{1}{\hat{\beta_1}-\delta}\right] \subset \Big[\frac{\sqrt{2}-\varepsilon_{up}}{\sqrt{2}\big((\hat{\beta_1}-\delta)-\sqrt{\rho_4(\hat{\beta_2}+\delta)}\big)}, \frac{\varepsilon_{up}+\sqrt{2}}{\sqrt{2}\big((\hat{\beta_2}+\delta)+\sqrt{\rho_4(\hat{\beta_2}+\delta)}\big)}\Big].  
\end{equation}   
Consequently, the condition \eqref{con_rate_1} is established.
This completes the proof of Theorem \ref{thm:3}.

\end{proof}

\section{Numerical Experiments}\label{sec6}
\indent In this section, we present numerical experiments to evaluate the practical performance of WRGD and TWRGD, with steepest descent step sizes. Unless otherwise noted, initialization of all algorithms is performed using 100 power iterations via Algorithm~\ref{alg:3}. The measurements $\{\bm{a}_k\}_{k=1}^{m}$ follow the model in \eqref{eq:Gaussian}. We conduct experiments on noise-free signals of dimension $n = 1000$, where the entries of $\bm{x}$ are generated independently from a standard complex Gaussian distribution, i.e., $\bm{x} \sim \mathcal{N}(0,\bm{I}_n/2)+\imath\cdot\mathcal{N}(0,\bm{I}_n/2).$ For TWRGD, the truncation parameters are set to $\tau_0=7$, $\tau_1=4.5$, and $\tau_2=8$. Furthermore, the comparative algorithms evaluated in these experiments utilize the optimal parameters recommended in~\cite{RefWorks:RefID:12-chen2017solving,RefWorks:RefID:34-cai2018solving,wang2017solving}.

We compare the experimental results of different algorithms by the distance between vectors up to a global phase:  
\[
\textrm{dist}(\bm{z}_t, \bm{x}):=\mathop{\min}_{\phi\in[0, 2\pi]}\|\bm{z}_t-e^{\imath\phi}\bm{x}\|_2,
\]
following the definition in \cite{RefWorks:RefID:12-chen2017solving}. As shown in Lemma \ref{equivalence}, this distance is equivalent to $\|\bm{Z}_t-\bm{X}\|_F$ in the convergence analysis of our method. We define a relative mean squared error (MSE) as
$\textrm{dist}(\bm{z}_t, \bm{x})/{\|\bm{x}\|_2}.$

\subsection{Computational Efficiency}
In this section, we compare the efficiency of the proposed algorithm TWRGD with TRGD \cite{RefWorks:RefID:34-cai2018solving}, TWF \cite{RefWorks:RefID:12-chen2017solving}, and TAF \cite{wang2017solving}.

In Figure~\ref{fig:mse_iteration} and Figure \ref{fig:mse_time}, we plot the MSE of the TWRGD, TRGD\cite{RefWorks:RefID:34-cai2018solving}, TWF\cite{RefWorks:RefID:12-chen2017solving}, and TAF\cite{wang2017solving} algorithms against the number of iterations and CPU time over 50 experiments with $m=6n$, respectively. The solid curves represent the mean MSE, while the shaded regions indicate the standard deviation. The results clearly demonstrate the superior performance of TWRGD compared to its non-weighted counterpart, TRGD, as well as TWF and TAF. Notably, TWRGD achieves a substantial improvement over TRGD, significantly reducing the number of iterations required for convergence.

\begin{figure}[H]
    \centering
    \begin{minipage}{0.48\textwidth}
        \centering
        \includegraphics[width=\textwidth]{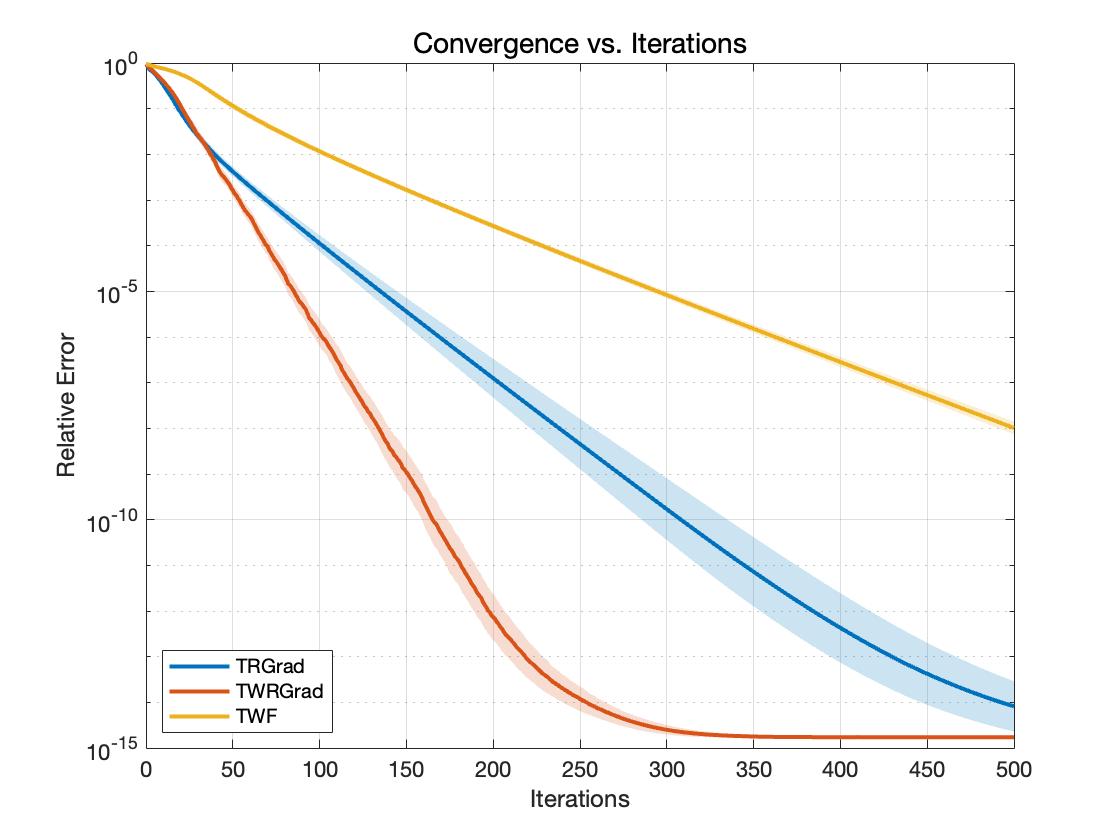}
        \caption{Relative MSE versus iteration count using: i) TWF; ii) TRGD; iii) TWRGD.}
        \label{fig:mse_iteration}
    \end{minipage}
    \hfill
    \begin{minipage}{0.48\textwidth}
        \centering
        \includegraphics[width=\textwidth]{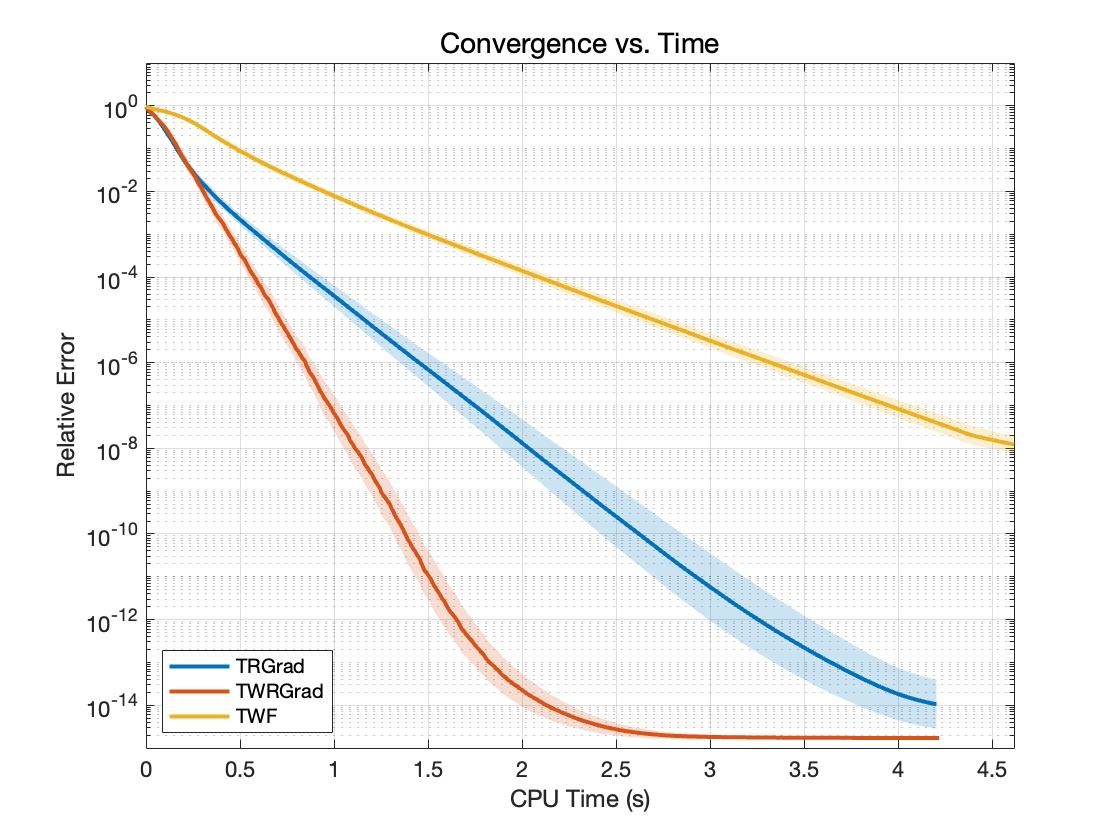}
        \caption{Relative MSE versus time (seconds) using: i) TWF; ii) TRGD; iii) TWRGD.}
        \label{fig:mse_time}
    \end{minipage}
\end{figure}

Next, we investigate the sensitivity of all the algorithms to variations in the number of measurements. We explore a range of measurement counts, $m$, selected from the set $\{10n, 12n, \dots, 30n\}$. We report the number of iterations and the computational time required to achieve MSE less than $10^{-3}$. Then we depict the iteration count and computation time of TWRGD, TRGD and TWF versus $m/n$ in Figure \ref{fig:iterations} and Table \ref{tab:time_movern} respectively.

The results presented in Figure~\ref{fig:iterations} and Table~\ref{tab:time_movern} demonstrate that the TWRGD algorithm consistently outperforms all competing methods in terms of both iteration count and computational time required to reach MSE of $10^{-3}$ across all tested values of $m$. As illustrated in Section~\ref{subsection:convergence speed}, the condition numbers for both TRGD and TWF remain strictly larger than $1$, resulting in slower convergence even when $m$ is large. In contrast, the condition number of TWRGD approaches $1$, enabling significantly faster convergence. The results of our experiment corroborate this: even when $m =30n$, TWRGD has the fastest convergence, compared to TRGD and TWF; besides, TWF has the slowest convergence speed, which is due to its largest condition number.

\begin{figure}[H]
    \centering
    \includegraphics[width=0.7\textwidth]{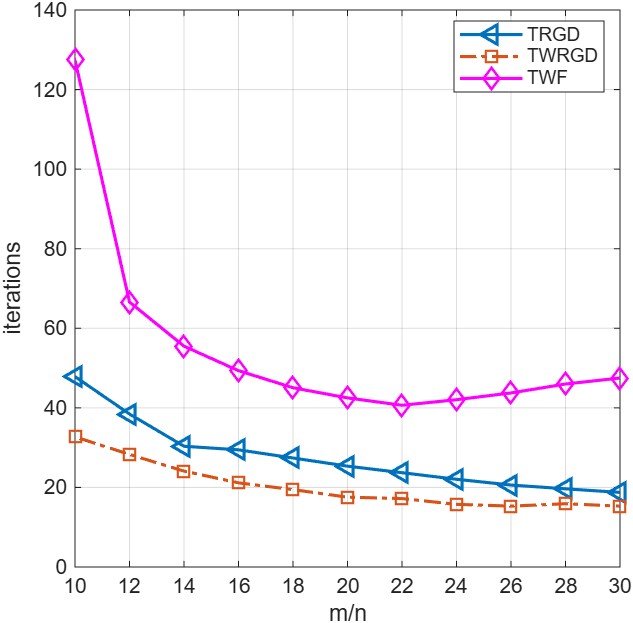}
  \caption{Iterations versus $m/n$ using: i) TRGD; ii) TWRGD; iii) TWF. }
   \label{fig:iterations}
\end{figure}

\begin{table}[t]
\centering
\caption{Computational CPU time (in seconds) versus $m/n$ for TRGD, TWRGD, and TWF.}
\label{tab:time_movern}
\begin{tabular}{c|ccc}
\hline
$m/n$ & TRGD & TWRGD & TWF \\
\hline
10 & 0.5629 & \textbf{0.3823} & 8.1845 \\
12 & 0.5785 & \textbf{0.4186} & 5.0432 \\
14 & 0.5453 & \textbf{0.4345} & 4.9137 \\
16 & 0.6246 & \textbf{0.4427} & 4.9910 \\
18 & 0.6549 & \textbf{0.4538} & 5.1251 \\
20 & 0.6686 & \textbf{0.4706} & 5.3459 \\
22 & 0.6764 & \textbf{0.4717} & 5.6182 \\
24 & 0.6854 & \textbf{0.4885} & 6.3042 \\
26 & 0.7141 & \textbf{0.5147} & 7.1025 \\
28 & 0.7610 & \textbf{0.6135} & 8.1221 \\
30 & 0.7793 & \textbf{0.6301} & 8.9180 \\
\hline
\end{tabular}
\end{table}

\subsection{Success Rates} 
This section evaluates the empirical probability over 50 independent trials considering varying numbers of measurements. We choose the ratio $m/n$ ranging from $2$ to $10$ in integer increments. A trial is deemed successful if the relative error satisfies MSE less than $10^{-3}$ within 500 steps.

Figure~\ref{fig:successful_probability 1} depicts the success rates of these algorithms versus $m/n$. It demonstrates that TWRGD and TRGD results in high-probability exact recovery once the number of measurements reaches approximately $m = 5n$. It is superior to TWF, but inferior to TAF. When $m\geq5n$, TWRGD, TRGD, TAF, and TWF all achieve successful recovery with a probability approaching one.
\begin{figure}[H]
    \centering
    \includegraphics[width=0.7\textwidth]{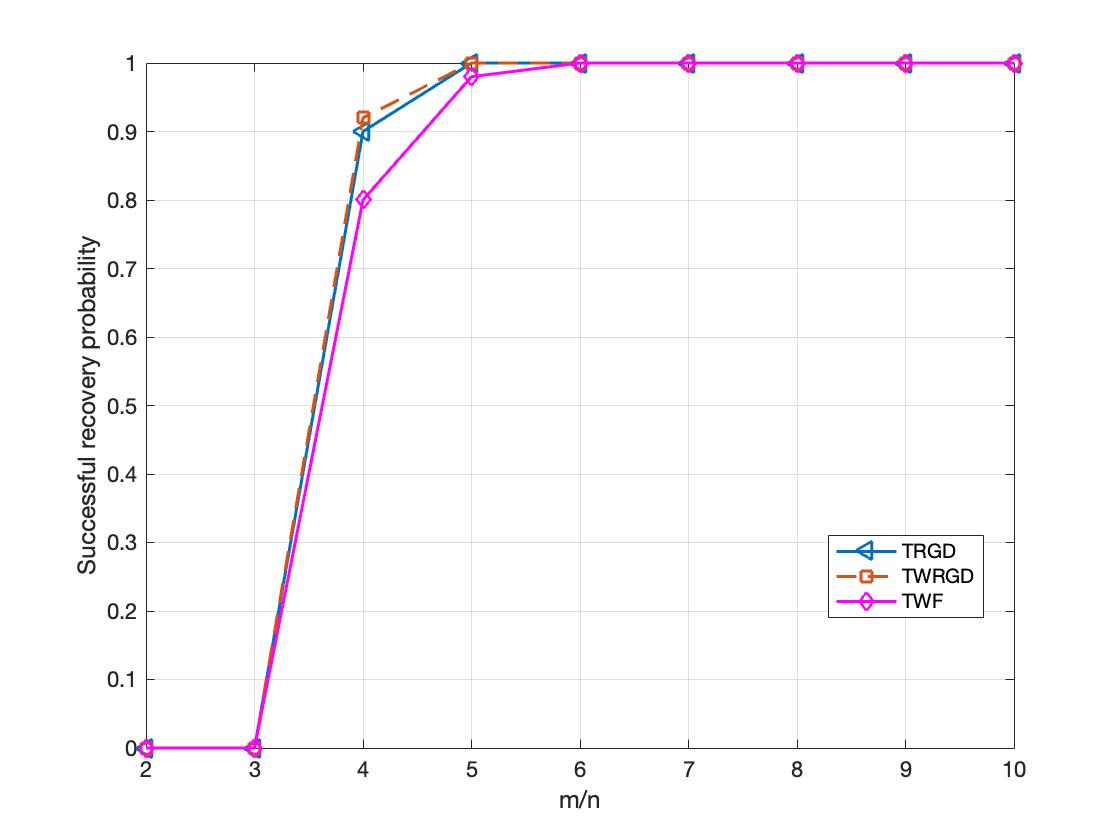}
\caption{Successful probability versus $m/n$ using: i) TWF; ii) TRGD; iii) TWRGD; iv) TAF. TWF and TAF use fixed step sizes.}
   \label{fig:successful_probability 1}
\end{figure}

\section{Conclusions}\label{sec7}
In this paper, we present two Weighted Riemannian Gradient Descent algorithms for solving a system of phaseless equations by defining a novel weighted metric. We have rigorously established an exact recovery guarantee for their truncated versions, demonstrating their capability to achieve successful recovery while maintaining optimal sampling complexity: $O(n)$ for Gaussian measurements. Compared with the canonical RGD and TWF, under the new Riemannian metric, the convergence speed of TWRGD for solving phase retrieval problems has improved and tends towards the optimal rate. Furthermore, empirical evaluations confirm the competitiveness of our algorithm compared to other state-of-the-art first-order methods. 

Finally, we conclude this paper by raising several intriguing questions and potential directions for future research. Firstly, we consider the exploration of alternative loss functions beyond the conventional least-square criterion \eqref{least_square}, such as the \textit{amplitude-based} loss introduced in \cite{wang2017solving}. Secondly, the investigation of sparse phase retrieval, as explored in works like \cite{Lu:2011} and \cite{Babak:2012}, presents an interesting and challenging direction for further studies. Lastly, extending the innovative weighted metric to address other problems, such as completion of the Euclidean distance matrix \cite{Alfakih1999SolvingED} and blind deconvolution \cite{li2016rapid}, holds significant potential and relevance.
\bibliographystyle{plain}
\bibliography{ref}
\section{Appendix}
\subsection{Useful Lemmas}
In this section, we provide some useful lemmas for supporting our results.
\begin{lemma}[Lemma 7, \cite{RefWorks:RefID:33-li2021riemannian}]\label{uuu11}
For any $\bm{x},\bm{ z} \in \mathbb{C}^n$ obeying $\operatorname{Re}(\bm{z}^*\bm{x})\geq0$, we have
$$
 \left\|\bm{z z}^*-\bm{x x}^*\right\|_F^2 \geq \frac{4}{5}\|\bm{z}-\bm{x}\|_2^2 \cdot\|\bm{x}\|_2^2. 
$$

\end{lemma}

\begin{lemma}[Lemma 3.1, \cite{RefWorks:RefID:11-candes2013phaselift}]\label{lem8.1}
Let $\{\bm{a}_k\}_{k=1}^{m}$ follow \eqref{eq:Gaussian}. For any $\epsilon>0$, define $\delta_1:=4(\epsilon^2+\epsilon)$. If $m \geq 2 \epsilon^{-2} n$, then with probability at least $1-2 e^{-m \epsilon^2 / 2}$,
\begin{equation} \label{gaussian:rip1}
(1-\delta_1) \leq \frac{1}{m}\left\|\mathcal{A}\left(\bm{u} \bm{u}^*\right)\right\|_1 \leq(1+\delta_1)    
\end{equation}
holds for all unit vectors $\bm{u}\in\mathbb{C}^{n}$. 
\end{lemma}
\begin{rem}
    Notice that $\eqref{gaussian:rip1}$ can be rewritten as $$(1-\delta_1)\leq \frac{1}{m}\sum_{k=1}^m\left|\bm{a}_k^{*} \bm{u}\right|^2 \leq(1+\delta_1). $$\end{rem}
\begin{lemma}\label{equivalence}
Let $\bm{z},\bm{x}\in\mathbb{C}^{n}$ obeying $\textrm{Re}(\bm{z}^*\bm{x})\geq0$ and $\|\bm{z}-\bm{x}\|_2\leq\delta_2\|\bm{x}\|_2$ for some $\delta_2\in(0,\frac{1}{2})$. Then we have 
$$
\frac{1}{2}\|\bm{x}\|_2\|\bm{z}-\bm{x}\|_2\leq\|\bm{zz}^*-\bm{xx}^*\|_F\leq(2+\delta_2)\|\bm{x}\|_2\|\bm{z}-\bm{x}\|_2.
$$
\end{lemma}

\begin{proof}
On the one hand, by Lemma \ref{uuu11}, we have $\|\bm{x}\|_2\|\bm{z}-\bm{x}\|_2\leq2\|\bm{zz}^*-\bm{xx}^*\|_F$ as $\textrm{Re}(\bm{z}^*\bm{x})\geq0$. On the other hand, as $\|\bm{z}-\bm{x}\|_2\leq\delta_2\|\bm{x}\|_2$ with $\delta_2\in(0,\frac{1}{2})$, we can get
$$
\|\bm{zz}^*-\bm{xx}^*\|_F^2\leq \|\bm{z}-\bm{x}\|_2^4+4\|\bm{z}-\bm{x}\|_2^3\|\bm{x}\|_2+4\|\bm{z}-\bm{x}\|_2^2\|\bm{x}\|_2^2=\|\bm{z}-\bm{x}\|_2^2(2\|\bm{z}-\bm{x}\|_2+\|\bm{x}\|_2)^2,
$$
which implies that $\|\bm{zz}^*-\bm{xx}^*\|_F\leq (2+\delta_2)\|\bm{x}\|_2\|\bm{z}-\bm{x}\|_2$. This completes the proof.

\end{proof}

\subsection{Computational Complexity}\label{sec8.2}
In this section, we analyze the computational complexity of Algorithms \ref{alg:1} in detail. As will be shown below, these two algorithms essentially rely on the iterative update of two vectors, with their overall computational complexity concentrated on matrix-vector multiplications. Notably, they do not require the large-scale operations (e.g., SVD decomposition and matrix multiplication) depicted in the statements of Algorithms \ref{alg:1}. This not only demonstrates the efficiency of the proposed algorithms but also provides theoretical support for the experimental simulations presented in Section \ref{sec6}.
For simplicity, we use the non-truncated version of WRGD as an example for illustration. Denote $\bm{A}=[\bm{a}_1^*,\ldots,\bm{a}_m^*]^{\top}$ as the random complex Gaussian measurements. At the $t$-th iteration, set $\bm{w}_t:=\bm{A}\bm{z}_t$ with each entry $w_{t}^{k}:=\bm{a}_k^*\bm{z}_t,(k=1,\ldots,m)$, $\bm{g}_t:=\bm{G}_t\bm{u}_t$ with $\bm{u}_t=\frac{\bm{z}_t}{\|\bm{z}_t\|_2}$, and $\theta_t:=\bm{u}_t^*\bm{G}_t\bm{u}_t$. Then $\bm{q}_t=\bm{g}_t-\theta_t\bm{u}_t$. As is shown in Algorithm \ref{alg:1}, we have 
\[\begin{split}
\mathcal{T}_{\mathbb{T}_t}^{(\mathfrak{o})}(\bm{G}_t )&=  \bm{u}_t\bm{u}_t^* \bm{G}_t+ \bm{G}_t \bm{u}_t\bm{u}_t^* - \tfrac{3}{2} \bm{u}_t\bm{u}_t^* \bm{G}_t\bm{u}_t\bm{u}_t^*\\
&=\bm{u}_t\bm{u}_t^* \bm{G}_t(\bm{I}_n- \bm{u}_t\bm{u}_t^*)+(\bm{I}_n- \bm{u}_t\bm{u}_t^*)\bm{G}_t\bm{u}_t\bm{u}_t^*+\frac{1}{2}\bm{u}_t\bm{u}_t^* \bm{G}_t\bm{u}_t\bm{u}_t^* \\
&=\bm{u}_t\bm{q}_t^*+\bm{q}_t\bm{u}_t^*+\frac{1}{2}\theta_t\bm{u}_t\bm{u}_t^*,
\end{split}
\]
This yields that 
\[
\begin{split}
\bm{W}_t&=\bm{Z}_t+\alpha_t\mathcal{T}_{\mathbb{T}_t}^{(\mathfrak{o})}(\bm{G}_t )\\
&=\bm{z}_t\bm{z}_t^*+\alpha_t\bigg(\bm{u}_t\bm{q}_t^*+\bm{q}_t\bm{u}_t^*+\frac{1}{2}\theta_t\bm{u}_t\bm{u}_t^*\bigg)\\
&=\left[\begin{array}{ll}\bm{u}_t & \frac{\bm{q}_t}{\|\bm{q}_t\|_2}\end{array}\right]\left[\begin{array}{ll}\|\bm{z}_t\|_2^2+\frac{\alpha_t\theta_t}{2}& \alpha_t\|\bm{q}_t\|_2\\ \alpha_t\|\bm{q}_t\|_2& 0\end{array}\right]\left[\begin{array}{l}\bm{u}_t^* \\ \frac{\bm{q}_t^*}{\|\bm{q}_t\|_2}\end{array}\right].
\end{split}
\]
and $\bm{W}_t$ is a semi-definite positive matrix with rank at most $2$.
Thus, 
\[
\begin{split}
\bm{Z}_{t+1}&=\mathcal{H}_1(\bm{W}_t)=\left[\begin{array}{ll}\bm{u}_t & \frac{\bm{q}_t}{\|\bm{q}_t\|_2}\end{array}\right]\mathcal{H}_1\bigg(\left[\begin{array}{ll}\|\bm{z}_t\|_2^2+\frac{\alpha_t\theta_t}{2}& \alpha_t\|\bm{q}_t\|_2\\ \alpha_t\|\bm{q}_t\|_2& 0\end{array}\right]\bigg)\left[\begin{array}{l}\bm{u}_t^* \\ \frac{\bm{q}_t^*}{\|\bm{q}_t\|_2}\end{array}\right]
\end{split}
\]
and we can calculate $\mathcal{H}_1(\bm{W}_t)$ based on the eigen-decomposition of the $2\times2$ matrix $\left[\begin{array}{ll}\|\bm{z}_t\|_2^2+\frac{\alpha_t\theta_t}{2}& \alpha_t\|\bm{q}_t\|_2\\ \alpha_t\|\bm{q}_t\|_2& 0\end{array}\right]$. Then we can derive its largest eigenvalue and the corresponding eigenvector as
\[\begin{split}
&\lambda_t:=\frac{(\|\bm{z}_t\|_2^2+\frac{\alpha_t\theta_t}{2})+\sqrt{(\|\bm{z}_t\|_2^2+\frac{\alpha_t\theta_t}{2})^2+4\alpha_t^2\|\bm{q}_t\|_2^2}}{2}\\
&\bm{\nu}_t=c_t\cdot\bigg[\begin{array}{ll} \alpha_t\|\bm{q}_t\|_2&  \frac{\sqrt{(\|\bm{z}_t\|_2^2+\frac{\alpha_t\theta_t}{2})^2+4\alpha_t^2\|\bm{q}_t\|_2^2}-(\|\bm{z}_t\|_2^2+\frac{\alpha_t\theta_t}{2})}{2}\end{array}\bigg]^{\top},
\end{split}
\]
where $c_t:=\bigg( \alpha_t^2\|\bm{q}_t\|_2^2+( \frac{\sqrt{(\|\bm{z}_t\|_2^2+\frac{\alpha_t\theta_t}{2})^2+4\|\bm{q}_t\|_2^2}-(\|\bm{z}_t\|_2^2+\frac{\alpha_t\theta_t}{2})}{2})^2 \bigg)^{-\frac{1}{2}} $.
This implies that 
\[\begin{split}
\bm{Z}_{t+1}&=\lambda\left[\begin{array}{ll}\bm{u}_t & \frac{\bm{q}_t}{\|\bm{q}_t\|_2}\end{array}\right]\bm{\nu}_{t}\bm{\nu}_{t}^*\left[\begin{array}{l}\bm{u}_t^* \\ \frac{\bm{q}_t^*}{\|\bm{q}_t\|_2}\end{array}\right]\\
&=c_t^2\lambda_t\bigg(\alpha_t^2\|\bm{q}_t\|_2^2\bm{u}_t+\frac{\sqrt{(\|\bm{z}_t\|_2^2+\frac{\alpha_t\theta_t}{2})^2+4\alpha_t^2\|\bm{q}_t\|_2^2}-(\|\bm{z}_t\|_2^2+\frac{\alpha_t\theta_t}{2})}{2\|\bm{q}_t\|_2}\bm{q}_t\bigg)\\
&\quad\cdot\bigg(\alpha_t^2\|\bm{q}_t\|_2^2\bm{u}_t+\frac{\sqrt{(\|\bm{z}_t\|_2^2+\frac{\alpha_t\theta_t}{2})^2+4\alpha_t^2\|\bm{q}_t\|_2^2}-(\|\bm{z}_t\|_2^2+\frac{\alpha_t\theta_t}{2})}{2\|\bm{q}_t\|_2}\bm{q}_t\bigg)^*\\
&:=\bm{z}_{t+1}\bm{z}_{t+1}^*
\end{split}\]
with rank-$1$, where 
$$
\bm{z}_{t+1}:=c\sqrt{\lambda}\bigg(\alpha_t^2\|\bm{q}_t\|_2^2\bm{u}_t+\frac{\sqrt{(\|\bm{z}_t\|_2^2+\frac{\alpha_t\theta_t}{2})^2+4\alpha_t^2\|\bm{q}_t\|_2^2}-(\|\bm{z}_t\|_2^2+\frac{\alpha_t\theta_t}{2})}{2\|\bm{q}_t\|_2}\bm{q}_t\bigg).
$$
Therefore, based on the above analysis, we find that the WRGD-typed algorithm essentially conducts iterative updates on two vectors—a design that significantly enhances its computational efficiency, rendering it even more efficient than vector-space iterative algorithms (e.g., WF algorithm). Moreover, the equivalent form of TWRGD is presented as follows.
\begin{algorithm}[H]
    \renewcommand{\algorithmicrequire}{\textbf{Initialization:}}
    \caption{Equivalent Vector Form of TWRGD Algorithm }
    \label{alg:4}
    \begin{algorithmic}[1]
         \Require $\bm{Z}=\bm{z}_0\bm{z}_0^*$.
          
          \ForAll{$t=0, 1, \ldots, $}
          \State $\bm{w}_t=\bm{A}\bm{z}_t$
          \State $\mathbb{I}_t=\{k|\mathcal{E}_0^k(\bm{y})\cap \mathcal{E}_1^k(\bm{z}_t)\cap \mathcal{E}_2^k(\bm{z}_t)\}$
          \State $\bm{g}_t=\frac{1}{\|\bm{z}_t\|_2} \sum_{k\in \mathbb{I}_t}(y_k-|w_t^{k}|^2 )w_t^k\bm{a}_k$ 
         \State $\theta_t:=\bm{u}_t^*\bm{G}_t\bm{u}_t$
         \State $\bm{q}_t=\bm{g}_t-\frac{\theta_t}{\|\bm{z}_t\|_2}\bm{z}_t$
         \State $a_t=\|\bm{z}\|_2^2+\frac{\alpha_t\theta_t}{2} $
         \State $b_t=2\alpha_t\|\bm{q}_t\|_2$
         \State $\lambda_t=\frac{a_t+\sqrt{b_t^2+a_t^2}}{2}$
         \State $c_t=\bigg(b_t^2+ (\frac{\sqrt{b_t^2+a_t^2}-\alpha_t}{2})^2\bigg)^{-\frac{1}{2}}$
         \State $\bm{z}_{t+1}=\sqrt{\lambda_t}c_t\bigg(\frac{b_t^2}{\|\bm{z}_t\|}\bm{z}_t+\frac{\sqrt{a_t^2+b_t^2}-a_t}{2\|\bm{q}_t\|_2}\bm{q}_t\bigg)$

          \EndFor

\end{algorithmic} 
\end{algorithm}

\end{document}